\documentclass[a4paper,onecolumn,11pt,accepted=2025-01-18]{quantumarticle}
\pdfoutput=1
\usepackage[utf8]{inputenc}
\usepackage[english]{babel}
\usepackage[T1]{fontenc}

\usepackage{latexsym}
\usepackage{amsmath}
\usepackage{amssymb}
\usepackage{amsthm,amsfonts}
\usepackage{accents}
\usepackage{tikz}
\usetikzlibrary{quantikz}

\usepackage{algorithm}
\usepackage{algorithmicx}
\usepackage{algpseudocode}

\usepackage{chngcntr}
\counterwithin{figure}{section}

\usepackage{tikz}
\usetikzlibrary{shapes,backgrounds}
\usepackage{tikzsymbols}
\usepackage{float} 

\usepackage{hyperref}
\usepackage{cleveref}
\hypersetup{
   breaklinks=true,   
   colorlinks=true,   
   pdfusetitle=true,  
}

\usepackage[sorting=none,sortcites=true, citestyle=numeric-comp]{biblatex}
\newtheorem{theorem}{Theorem}
\newtheorem{corollary}[theorem]{Corollary}
\newtheorem{lemma}[theorem]{Lemma}

\newtheorem{proposition}[theorem]{Proposition}
\newtheorem{definition}[theorem]{Definition}
\newtheorem{problem}[theorem]{Problem}
\newtheorem{claim}[theorem]{Claim}

\newcommand{\ketbra}[2]{|#1\rangle\! \langle #2|}

\newcommand{\Id}{\mathbb{I}}

\newcommand{\norm}[1]{\left\lVert#1\right\rVert}
\newcommand{\abs}[1]{\left\lvert#1\right\rvert}
\newcommand{\avg}[1]{\left\langle#1\right\rangle}

\advance \topmargin by -\headheight
\advance \topmargin by -\headsep
\evensidemargin \oddsidemargin
\begin{document}

\title{Hamiltonian Learning via Shadow Tomography of Pseudo-Choi States}
\author{Juan Castaneda}
\affiliation{Department of Physics, University of Toronto, Toronto ON, Canada}
\author{Nathan Wiebe}
\affiliation{Department of Computer Science, University of Toronto, Toronto ON, Canada}
\affiliation{ Pacific Northwest National Laboratory, Richland WA, USA}
\affiliation{ Canadian Institute for Advanced Research, Toronto ON, Canada}

\maketitle

\begin{abstract}
    We introduce a new approach to learn Hamiltonians through a resource that we call the pseudo-Choi state, which encodes the Hamiltonian in a state using a procedure that is analogous to the Choi-Jamiolkowski isomorphism. We provide an efficient method for generating these pseudo-Choi states by querying a time evolution unitary of the form $e^{-iHt}$ and its inverse, and show that for a Hamiltonian with $M$ terms the Hamiltonian coefficients can be estimated via classical shadow tomography within error $\epsilon$ in the $2$-norm using $\widetilde{O}\left(\frac{M}{t^2\epsilon^2}\right)$ queries to the state preparation protocol, where $t \leq \frac{1}{2\norm{H}}$. We further show an alternative approach that eschews classical shadow tomography in favor of quantum mean estimation that reduces this cost (at the price of many more qubits) to $\widetilde{O}\left(\frac{M}{t\epsilon}\right)$. Additionally, we show that in the case where one does not have access to the state preparation protocol, the Hamiltonian can be learned using $\widetilde{O}\left(\frac{\alpha^4M}{\epsilon^2}\right)$ copies of the pseudo-Choi state. The constant $\alpha$ depends on the norm of the Hamiltonian, and the scaling in terms of $\alpha$ can be improved quadratically if using pseudo-Choi states of the normalized Hamiltonian. Finally, we show that our learning process is robust to errors in the resource states and to errors in the Hamiltonian class.  Specifically, we show that if the true Hamiltonian contains more terms than we believe are present in the reconstruction, then our methods give an indication that there are Hamiltonian terms that have not been identified and will still accurately estimate the known terms in the Hamiltonian.
\end{abstract}
\tableofcontents

\section{Introduction}

Identifying the system Hamiltonian is a central problem in physics, as the Hamiltonian uniquely describes the dynamics of any closed quantum system. Additionally, as quantum computers evolve, the need to characterize errors in their behaviour is becoming increasingly important, and as a result, several classes of quantum certification protocols have been designed to ensure quantum devices behave as intended. There exist a wide variety of such protocols, each providing different amounts and types of information about the processes in question, and operating under different assumptions with varying resource costs~\cite{eisertCertificationReview}. In contrast to efficient protocols like randomized benchmarking, which only provide information about the average error rate of a quantum gate set \cite{knill2008, proctor2017, helsen2022}, Hamiltonian learning can provide information about the systematic error in a unitary evolution under some fixed but erroneous Hamiltonian, without requiring exponential resources like protocols such as full process tomography which can provide more information about noise processes. Learning the Hamiltonian yields low-level information about the quantum dynamics that can be directly used for quantum control, and thus is an indispensable task for calibration and certification of quantum devices since it can be used to debug quantum Hamiltonian simulation algorithms by learning the simulated Hamiltonian and comparing the result to the ideal target.

Despite the importance of this problem, systematic methods for inferring the Hamiltonian and bounding the sample complexity needed to do so has only recently become a major focus in quantum information science. The most important thrusts include direct inference from short time experiments~\cite{da2011practical}, approximate Bayesian inference~\cite{Granade_2012, Wiebe2014, Wang2017, evans2019scalable}, PAC (Probably Approximate Correct) methods based on learning from thermal states~\cite{anshu2020sampleefficient, haah2021optimal}, as well as PAC learning by querying the time evolution unitary $U = e^{-iHt}$ \cite{haah2021optimal, yu2022practical, gu2022practical, PhysRevLett.130.200403, caro2023learning, frança2022efficient, dutkiewicz2023advantage}. These latter techniques provide rigorous sample bounds on the evidence needed to learn the Hamiltonian within fixed error and probability of failure using thermal states or time evolution as the learning resource. However, in addition to potentially costly state preparation methods, approaches for Hamiltonian learning from thermal states are generally suitable for the high-temperature regime, not at very low temperatures.~\cite{anshu2020sampleefficient,haah2021optimal}.

Much of the recent progress made in the field has instead used the time evolution operator $U = e^{-iHt}$ as the learning resource. This approach naturally leads to applications for certifying quantum computers, where one would like to check whether a desired time evolution is actually performed in practice (this applies to scenarios where one would like to verify their own computations, as well as to scenarios where one party performs the computation, and a second party is interested in determining what went on under the hood). It also adds the possibility of using quantum control to improve the efficiency of learning techniques. Even so, the focus has typically been restricted to physically local Hamiltonians (and the slightly more general low-intersection class) rather than the more generic class of $k$-local Hamiltonians or beyond. In this work we address this by providing several new algorithms that allow us to efficiently learn Hamiltonians beyond the low-intersection regime, and provide explicit bounds on number of operations needed to produce these state both with and without coherent quantum control. Our method achieves this by introducing a new learning resource - in particular the ``pseudo-Choi state'' of the Hamiltonian - which we show can be used to efficiently learn the Hamiltonian. The main idea is that these states encode the Hamiltonian in a way that makes it easy to recover the Hamiltonian coefficients by using few copies of the state to estimate the expectation values for a set of ``decoding operators'', which are simple operators whose expectation values can be computed on a classical computer in polynomial time. Furthermore, we provide a method of generating these pseudo-Choi states efficiently via controlled queries to the time evolution unitary $U = e^{-iHt}$ and its inverse $U^\dagger$. This access model is similar to that in~\cite{haah2021optimal, yu2022practical, gu2022practical, PhysRevLett.130.200403, frança2022efficient, caro2023learning, dutkiewicz2023advantage}, though it is slightly more restrictive because it requires access to the backwards time-evolution $U^\dagger$ as well. 

While there have been a few recent approaches focused on learning Hamiltonians beyond the low-intersection class (\cite{yu2022practical, caro2023learning}), our approach improves the dependence of the query complexity on the error, locality of the Hamiltonian terms, and the number of Hamiltonian terms. In addition to this, some of our learning algorithms are robust in the sense that if there are unexpected Hamiltonian terms our algorithm will still correctly estimate the coefficients for terms we expected, and also signal that there were additional terms present. A more in-depth discussion about how our result compares to existing techniques for learning Hamiltonians is left for Section~\ref{section:Applications}. For now, we summarize by saying that there is no clear ``best'' approach, but rather several good approaches that make sense for different scenarios (i.e. the type of input model considered and the types of Hamiltonians considered). That said, under the admittedly more restrictive input model considered here, the techniques we develop are able to achieve the best scaling in terms of both the number of Hamiltonian terms, as well as the error in estimating the parameters, with the added benefit of robustness. Furthermore, depending on the choice of error metric, our approach is applicable to \textit{any} n-qubit Hamiltonian, including those with exponentially many terms (thus going far beyond even the k-local regime), so long as the norm of the Hamiltonian is bounded by a polynomial in the number of qubits. For these reasons, it remains of interest to study whether pseudo-Choi states can be produced in an efficient manner if one relaxes the input model to only require queries to $U$.

  The paper is laid out as follows.  Section~\ref{section:results} summarizes the key results of our work and gives specific statements of the learning problem we consider. In addition, we include a summary of the advantages and drawbacks of our approach when compared to existing Hamiltonian learning methods.
    Section~\ref{section:ChoiStateLearning} introduces the pseudo-Choi state of a Hamiltonian and discusses how it can be used as a resource to efficiently solve the Hamiltonian learning problem. A brief overview of the classical shadows tomography procedure introduced in~\cite{Huang_2020} is also included. Section~\ref{section:BlackboxLearning} examines the cost of learning the Hamiltonian by querying blackboxes that implement $U = e^{-iHt}$ and $U^\dagger$ using classical shadows in the learning procedure. An alternative method is introduced in Section~\ref{section: QME}, which uses a quantum algorithm introduced by \textit{Huggins et al.}~\cite{Huggins_nearly_optimal_estimating} for estimating expectation values. This approach quadratically improves how the query complexity scales with respect to the error and evolution time. Section~\ref{section:Applications} discusses the types of Hamiltonians our learning approach is well suited for and includes a comparison to prior work on Hamiltonian learning. Finally, we discuss the robustness of our procedure in Section~\ref{section: robustness}, before concluding in Section~\ref{section:conclusion}.

\section{Problem Statements and Summary of Results}\label{section:results}
The problem that we wish to tackle is a natural one.  Let us assume that we have a set of possible Hamiltonian terms whose weighted sum comprises a Hamiltonian.  Our task is to learn, within a fixed probability of failure, the weights for the Hamiltonian coefficeints within fixed error.  This requirement falls within the scope of the quantum PAC learning paradigm, which was first introduced for theoretical purposes~\cite{servedio2004equivalences,Aaronson_2004,arunachalam2017guest} but has recently become a practical tool used in characterization~\cite{Huang_2020,Huggins_nearly_optimal_estimating}.  The formal statement of this intuitive problem is given below.

        \begin{definition} [Hamiltonian Learning Problem]{ \label{def:HamLearningProblem}Let $\mathcal{H}_S$  be a Hilbert space of dimension $d=2^n$, and let $\{H_m \mid m=0,\ldots M-1\}$ be a set of Hermitian and unitary operators acting on the space such that $\{H_m\}$ forms an orthonormal operator basis with respect to the inner product $\langle H_j,H_k \rangle = d^{-1} {\rm Tr}(H_j H_k)$ and let the Hamiltonian $H$ be parameterized for ${\mathbf{c}} \in \mathbb{R}^M$ such that
        \begin{align*}
            H = \sum_{m=0}^{M-1} c_m H_m
        \end{align*}
        
        The Hamiltonian learning problem is to compute a vector $\hat{\textbf{c}}$ such that with probability at least $1-\delta$ we have that 
        $
            \norm{\hat{\textbf{c}} -  \textbf{c}}_2 \leq \epsilon
        $
        where $\norm{\cdot}_2$ is the vector $L^2$-norm.
        }
        \end{definition}
        
  In the most general case, there are up to $d^2$ unique Hamiltonian terms in the decomposition.  These dense Hamiltonians generically cannot be learned efficiently within the operator basis considered as they require that we learn exponentially many terms in the Hamiltonians. However, focusing on special classes of Hamiltonians can reduce this number. Several recent papers (\cite{haah2021optimal,PhysRevLett.130.200403,gu2022practical}) have focused on the \textbf{low-intersection} class of Hamiltonians, which was introduced in \cite{haah2021optimal}, and we describe in Definition~\ref{def:LowIntersection}. The low-intersection class contains all spatially local Hamiltonians, which are the focus of the previous work by Anshu et al. \cite{anshu2020sampleefficient}, and is therefore relevant for describing many physical systems.
        
        \begin{definition} [Low-Intersection Hamiltonian]{\label{def:LowIntersection}
        
        A low-intersection Hamiltonian on $n$ qubits is of the form $H=\sum_m c_m H_m$ and satisfies the following properties:
        
        \begin{enumerate}
            \item It is \textbf{k-local} : Each $H_m$ is a finite dimensional Hermitian operator that acts non-trivially (i.e. as a non-identity operator) on at most $k \in \mathcal{O}(1)$ qubits.
            
            \item For each Hamiltonian term $H_m$ there exist at most a constant number of distinct terms $H_l$ such that $H_m$ and $H_l$ act non-trivially on at least one common qubit.
        \end{enumerate} }
        \end{definition}
        The second property directly implies that each qubit is acted on non-trivially by only a constant number of Hamiltonian terms. As a result, the number of Hamiltonian terms in low-intersection Hamiltonians scales as $\mathcal{O}(n)$.
        
 In contrast, the class of k-local Hamiltonians that are not low-intersection is also of interest. This includes Hamiltonians that have pairwise interactions between every pair of qubits, as well as Hamiltonians that can be represented by star graphs (i.e. one qubit interacts with all others, but there are no/few interactions between the other qubits). Importantly, though this class is more general than the low-intersection class, k-local Hamiltonians contain at most $\mathcal{O}(n^k)$ Hamiltonian terms, and so the learning approach we introduce in this paper remains tractable. Physically interesting examples of k-local Hamiltonians include hardcore boson models, Heisenberg models on the complete graph, and spin glass models such as the Sherrington-Kirkpatrick and p-spin models. Furthermore, many quantum algorithms use Hamiltonian simulation as a subroutine, so for certification purposes it is useful to expand the scope of Hamiltonian learning beyond the low-intersection class and even to ``less physical'' Hamiltonians that may not be common in nature. 
        
We consider Hamiltonian learning using two different, yet related, learning resources. The first is a state which we dub the ``pseudo-Choi state'' of the Hamiltonian, which roughly speaking is the representation of the Hamiltonian as a quantum state in a larger Hilbert space: 
\begin{align}\label{eq: ChoiBrief}
    \ket{\psi_{c}} & := \frac{({H} \otimes \Id_A)\ket{\Phi_d}_{SA}\ket{0}_C + \ket{\Phi_d}_{SA}\ket{1}_C}{\alpha}
\end{align}
where $H$ is dimensionless, $\alpha$ is a normalization constant and $\ket{\Phi_d}_{SA}$ is a maximally entangled state over two sub-systems denoted $S$ and $A$. The name prescribed to this state is a consequence of the fact that the left-hand side of the state, $({H} \otimes \Id_A)\ket{\Phi_d}_{SA}$, is similar to a Choi state, which is defined for a unitary $U_S \in \mathcal{H}_S $ as $\ket{\psi_{Choi}} = (U_S \otimes \Id_A)\ket{\Phi_d}_{SA}$. However, it is not generally a Choi state because $H$ need not correspond to a CPTP map. The pseudo-Choi state is formally described in Definition~\ref{def:pseudoChoi}. While these states do not result from natural processes like in the case of thermal states, we will see that they allow for very efficient learning of the Hamiltonian coefficients, which makes it of interest to find ways of efficiently generating them.

The second learning resource under consideration is a controllable time-evolution black-box $U = e^{-iHt}$. Along with thermal states, the unitary dynamics of a system is a natural resource to consider for learning the system Hamiltonian. In addition to the cost of thermal state preparation, learning from thermal states is often not feasible, such as at very low temperatures.  To see this, consider two Hamiltonians $H_0 \succeq 0$ and $H_1 \succeq 0$ 
\begin{align}
H' = \ketbra{\psi}{\psi} + (I - \ketbra{\psi}{\psi})H_0(I - \ketbra{\psi}{\psi})\\
H'' = \ketbra{\psi}{\psi} + (I - \ketbra{\psi}{\psi})H_1(I - \ketbra{\psi}{\psi})
\end{align}
In both cases, the zero-temperature groundstate is $\ketbra{\psi}{\psi}$ and thus the success probability of successfully distinguishing between both Hamiltonians is $1/2$, which shows that for sufficiently low temperature thermal states it can become impractical to learn the Hamiltonian in question.
In such cases, it may be more convenient to use the unitary evolution of the system as the learning resource. In Section~\ref{section:BlackboxLearning} we demonstrate that the pseudo-Choi state can be generated efficiently using the time evolution and its inverse.

As the time evolution unitary is merely one of conceivably many ways of generating the pseudo-Choi state, Section~\ref{section:ChoiStateLearning} considers Hamiltonian learning in the more general scenario where copies of the pseudo-Choi state are provided to the user with no promises about how they were generated. We refer readers to Section~\ref{section: comparison} for more discussion about pseudo-Choi states, and in particular why they may be more powerful than regular Choi states of the unitary evolution. On the other hand, while the time evolution unitary is merely a specific example of how such states can be generated, it is also a more intuitive resource - especially in the context of certifying quantum simulators. Furthermore, it has been used as the learning resource in previous work on Hamiltonian learning (\cite{haah2021optimal, yu2022practical, gu2022practical, PhysRevLett.130.200403, caro2023learning}), which allows for a more direct comparison between the sample complexities of each approach. Therefore, we consider it as a separate case in Section~\ref{section:BlackboxLearning}.

\subsection{Main Results}

It should be noted that while our results below are stated for k-local Hamiltonians, our methods can be applied to more general Hamiltonians. In particular, as long as there are at most $\textbf{poly}(n)$ Hamiltonian terms, and we have information about the structure of the Hamiltonian (i.e. which Hamiltonian terms are present), the query complexity remains polynomial in the number of qubits for both the classical shadows and QME based approaches, as does the computational complexity. Furthermore, if we consider the infinity-norm in learning the Hamiltonian coefficients, our approach will be able to learn \textit{any} n-qubit Hamiltonian with $\norm{H} \in \textbf{poly}(n)$ using only $\textbf{poly}(n)$ queries to $U$ and $U^\dagger$. Remarkably, this includes Hamiltonians with an exponential number of terms. For these cases with exponentially many Hamiltonian coefficients, while our methods are query-efficient, the computational complexity will be exponential. However, since all the Hamiltonian coefficients can in theory be computed at the same time on a classical computer, this could be somewhat alleviated with access to a large classical memory.

    Our first main result is Theorem~\ref{thm:samplecomplexityIntro}, which provides an upper bound on the number of pseudo-Choi states needed to learn a Hamiltonian. 
    
    \begin{theorem}[Solving the pseudo-Choi Hamiltonian Learning Problem via Classical Shadows] { \label{thm:samplecomplexityIntro}

    The number of pseudo-Choi states required to solve the Hamiltonian learning problem of Definition~\ref{def:HamLearningProblem} within error $\epsilon$ and failure probability at most $\delta$ for a k-local Hamiltonian is at most
    \begin{align*}
        N \in \widetilde{\mathcal{O}}\left( \frac{\alpha^4  n^k}{\epsilon^2}  \right)
    \end{align*}
    where $\alpha\le 1+\|H\|$ is the normalization constant from equation~\eqref{eq: ChoiBrief}.
    }
    \end{theorem}
    Theorem~\ref{thm:samplecomplexityIntro} is presented in more detail as Theorem~\ref{thm:samplecomplexityUB} and is proven in Section~\ref{sec: SampleCompUB}. Note that in Theorem~\ref{thm:samplecomplexityIntro} we have replaced $M$ by the upper bound $M \in \mathcal{O}(n^k)$ for the number of terms in a k-local Hamiltonian, and the $\mathcal{\widetilde{O}}$ notation hides logarithmic factors.

    Our second main result, Theorem~\ref{thm:querycomplexityIntro}, gives an upper bound on the number of queries to the time evolution unitary and its inverse needed to learn a Hamiltonian using classical shadows. This result is related to the previous one, as our approach involves first generating copies of the pseudo-Choi state by querying the unitary dynamics of the Hamiltonian, and then proceeding in similar fashion to the approach used to prove Theorem~\ref{thm:samplecomplexityIntro}.
    
    \begin{theorem} [Solving the Unitary Hamiltonian Learning Problem via Classical Shadows]{\label{thm:querycomplexityIntro}

    The number of queries to the time evolution unitary $e^{-iHt}$ and its inverse, for $t \in \mathcal{O}( 1/\norm{H})$, required to solve the Hamiltonian learning problem of Definition~\ref{def:HamLearningProblem} within error $\epsilon$ and failure probability at most $\delta$ for a k-local Hamiltonian is at most
    \begin{align*}
        N & \in \widetilde{\mathcal{O}}\left( \frac{n^k}{t^2 \epsilon^2}\right)
    \end{align*}
    }
    \end{theorem}
    Theorem~\ref{thm:querycomplexityIntro} is stated in more detail as Theorem~\ref{thm:querycomplexityUB} and is proven in Section~\ref{section: QueryComplexity}. Again, we have replaced the number of Hamiltonian terms, $M$, by its upper bound for a k-local Hamiltonian. Importantly, the method we use for generating our resource states produces pseudo-Choi states of the normalized Hamiltonian, rather than the true Hamiltonian, which effectively results in a quadratic improvement of the query complexity in terms of the normalization constant. This is reflected in Theorem~\ref{thm:querycomplexityIntro} by the quadratic scaling in terms of $1/t$, versus the quartic scaling in terms $\alpha$ from Theorem~\ref{thm:samplecomplexityIntro}.

    We also demonstrate that our learning algorithm is robust to two types of errors and in fact provides a witness to the fact that an error has occurred in the learning problem.  The first type of error that we show robustness to is under specification of the Hamiltonian, wherein the Hamiltonian contains further terms that are not specified in the Hamiltonian Learning Problem.  In this case, we show that the learning problem will provide the coefficients for the Hamiltonian terms present in the system and the presence of uncharacterized terms can be inferred from the sum of the coefficients observed relative to the coefficient sum inferred from the pseudo-Choi state.  The second form of robustness that we consider is to error in the Choi state, in which case we demonstrate that polynomially small noise can be tolerated in the Hamiltonian learning problem which indicates that the protocol is robust. These results are discussed in Section~\ref{section: robustness}.
    
    Our final main result, Theorem~\ref{thm: gradEstFullQueryIntro}, gives an upper bound on the number of queries to the time evolution unitary and its inverse needed to learn a Hamiltonian using the mean-estimation results of Huggins et al.~\cite{Huggins_nearly_optimal_estimating}. This process requires more quantum computations than the classical shadows approach (in which all post-measurement calculations are classical), but improves the query complexity quadratically with respect to the evolution time and the error in learning the coefficients.
    \begin{theorem}[Solving the Unitary Hamiltonian Learning Problem via Quantum Mean Estimation]\label{thm: gradEstFullQueryIntro}
    The number of queries to the time evolution unitary $e^{-iHt}$ and its inverse, for $t \in \mathcal{O}( 1/\norm{H})$, required to solve the Hamiltonian learning problem of Definition~\ref{def:HamLearningProblem} within error $\epsilon$ and failure probability at most $\delta$ for a k-local Hamiltonian is at most
    \begin{align*}
        N &\in \widetilde{\mathcal{O}}\left( \frac{n^k}{\epsilon t} \right)
    \end{align*}
\end{theorem}

Theorem~\ref{thm: gradEstFullQueryIntro} is restated for a Hamiltonian with $M$ terms and proven as Theorem~\ref{thm: gradEstFullQuery} in Section~\ref{section:UnitaryConstruction}. Although our input model requires the backwards time evolution, this is the best query complexity we are aware of for learning the Hamiltonian coefficients to error $\epsilon$ in the 2-norm (see Section~\ref{section:Applications} for further discussion).

One caveat in that should be mentioned is that our approach requires short evolution times on the order of $1/\norm{H}$ in order to prepare pseudo-Choi states from $U$ and $U^\dagger$, preventing us from using long evolution times to reduce the query complexity. Furthermore, a good estimate of $\norm{H}$ is required in order to achieve the query complexity in Theorems~\ref{thm:querycomplexityIntro} and~\ref{thm: gradEstFullQueryIntro}. Without this knowledge, the norm can be overestimated by the number of Hamiltonian terms (leading to shorter-than-necessary evolution times) to ensure the algorithm doesn't break down. This, for example, would lead to an additional factor of $n^k$ in the query complexity in Theorem~\ref{thm: gradEstFullQueryIntro}.


\section{Learning the Hamiltonian from Pseudo-Choi States}\label{section:ChoiStateLearning}

This section is organized in the following way. We begin by explaining the first of the two learning resources (the pseudo-Choi state of the Hamiltonian, see Definition~\ref{def:pseudoChoi}) that will be considered in our model of Hamiltonian learning. Next, Section~\ref{subsection:shadows} summarizes the parts of the classical shadows procedure (introduced in~\cite{Huang_2020}) which are necessary to understand our approach. Finally, Section~\ref{subsection:LearnCoeffs} contains the two main results of this section. In particular: Algorithm~\ref{algorithm: FindCoeffClifford}, which describes how one would learn the Hamiltonian from multiple copies of the pseudo-Choi state through the use of classical shadows tomography, and Theorem~\ref{thm:samplecomplexityUB}, which describes the sample complexity of the approach.

Note that the use of classical shadows is not strictly necessary in the learning procedure; any shadow tomography protocol - that is, any protocol that allows one to estimate expectation values of operators for a given quantum state - can be used instead. The main benefits of classical shadows here are two-fold. First, for the operators we are interested in, the sample complexity has favourable scaling in terms of the number of Hamiltonian terms as a result of the shadow norm of these operators being upper bounded by a constant. Second, once all measurements are complete, the remainder of the algorithm is purely classical. While the main approach we focus on uses classical shadows, Section~\ref{section: QME} gives an example of an alternative approach.

Our input state, the ``pseudo-Choi state'' of the Hamiltonian, consists of three subsystems. It contains the subsystem of interest, an ancillary subsystem of the same size, and an additional single-qubit subsystem that is used to ensure that the norm of the Hamiltonian is not lost in the pseudo-Choi state.  The formal definition is given below. Note that while lowercase subscripts are used as indices throughout this paper, uppercase subscripts identify the system that corresponds to an operator or state. For example, the ket $\ket{\Phi_d}_{SA}$ resides in the Hilbert space $\mathcal{H}_{S}\otimes \mathcal{H}_A$.

\begin{definition}[Pseudo-Choi State]\label{def:pseudoChoi}
    Let $\mathcal{H}_S$ be the $d=2^n$ dimensional Hilbert space upon which the map containing a dimensionless Hamiltonian $H$ acts. Furthermore, let $\mathcal{H}_A$ be the Hilbert space of an ancilla of the same size, and let $\mathcal{H}_C$ be the 2 dimensional Hilbert space of another ancilla qubit.
    
    The pseudo-Choi state of the dimensionless Hamiltonian $H$ resides on the Hilbert space $\mathcal{H}_S \otimes \mathcal{H}_A \otimes \mathcal{H}_C$ and is defined as
    \begin{align}
        \ket{\psi_{c}} & := \frac{({H} \otimes \Id_A)\ket{\Phi_d}_{SA}\ket{0}_C + \ket{\Phi_d}_{SA}\ket{1}_C}{\alpha}\label{eq:choistate}
    \end{align}
    where $\alpha$ is a normalization constant of the form
    \begin{align}
        \alpha &:= \sqrt{\bra{\Phi_d}_{SA}(H^2\otimes \Id_A) \ket{\Phi_d}_{SA} + 1}
    \end{align}
    and
    \begin{align}
        \ket{\Phi_d}_{SA} = \frac{1}{\sqrt{d}} \sum_{i=0}^{d-1}\ket{i}_S\otimes\ket{i}_A
    \end{align}
    is the maximally entangled state between $\mathcal{H}_S$ and $\mathcal{H}_A$.
\end{definition}

Expanding $\alpha$ using the representation of the Hamiltonian from Definition~\ref{def:HamLearningProblem}, it is easily shown (see Appendix~\ref{section:randomsection2}) that
\begin{align}
    \alpha & = \sqrt{\norm{\textbf{c}}_2^2 + 1}\label{eq:alpha}
\end{align}

Note that the right-hand side of the pseudo-Choi state is important, since it acts as a reference and allows us to learn the overall sign and norm of the Hamiltonian (which would otherwise show up as a global phase if the reference term was not there).




\subsection{Shadow Tomography of the pseudo-Choi State \label{subsection:shadows}}
With the pseudo-Choi state in hand, we can now begin to learn the coefficients of the Hamiltonian through the use of shadow tomography. The idea of shadow tomography is to construct a model that can predict the expectation values of a particular set of observables with low error and high probability over the set.  Unlike conventional tomography, shadow tomography does not aim to provide a specific density operator, which means that the output expectation values need not be precisely consistent with a density operator, nor does the model need to provide a high fidelity estimate of the quantum state.  Instead, we require only that it predicts the outcomes of a set of observables.  Fortunately, this is all that we need to reconstruct the Hamiltonian and so this tool is precisely what we need to address our problem.  A formal statement of the shadow tomography problem is given below.

\begin{problem} [Shadow Tomography Problem]{\label{Prob:ShadTom}

Given an unknown quantum state, $\rho$, of dimension $d'=2^{\eta}$, and a set of $L$ observables, $\{E_i\}$, estimate $\text{Tr}(\rho E_i)$ to within error $\epsilon_s$ for all $1 \leq i \leq L$, with success probability $1-\delta_s$
using as few copies of $\rho$ as possible. \cite{aaronson2018shadow}
}
\end{problem}

This problem can be solved classically in in settings where the operators in question have efficient classical representations, such as Pauli operators. This approach, referred to as classical shadows tomography \cite{Huang_2020}, involves using a small number of copies of $\rho$ to create a classical representation of $\rho$. Section~\ref{subsection:LearnCoeffs} details how classical shadows can be used in conjunction with the pseudo-Choi state to learn the Hamiltonian coefficients, but to make it clear to the reader where some of the objects in later sections come from, we begin with a brief overview of classical shadows tomography.


The version of classical shadows we consider here involves applying Clifford operators from the group $\text{Cl}(2^\eta)$ to an $\eta$ qubit state $\rho$. The procedure consists of $N$ rounds in which a Clifford operator $U_i$ is sampled uniformly at random and applied to $\rho$, followed by a measurement of each qubit of the resulting state in the computational basis. After each measurement, a ``snapshot'' of the form $\sigma_i = U_i^\dagger \ket{b_i}\bra{b_i}U_i$ is stored in classical memory. Here $U_i$ is the randomly sampled Clifford unitary applied to $\rho$ during the $i^{th}$ round, and $\ket{b_i}$ is the bit string of length $\eta$ that encodes the measurement outcomes of the $i^{th}$ round. The $j^{th}$ bit of $\ket{b_i}$ is the outcome of measuring the $j^{th}$ qubit in the computational basis (i.e. if a bit in $\ket{b_i}$ is zero, then the outcome of measuring the corresponding qubit was $\ket{0}$). These snapshots are stabilizer states, meaning they can be obtained from $\ket{0}^{\otimes \eta}$ using only Clifford operations, and as a result $O\left(\eta^2\right)$ classical bits are needed to store each one in memory \cite{Aaronson_2004}.  In total, the procedure uses $N$ copies of $\rho$ to generate $N$ classical snapshots.

Note that although $\ket{b_i}$ is a classical bit string, we represent it using Dirac notation due to the simplicity of the notation, and to be consistent with the original notation used to describe classical shadows in \cite{Huang_2020}. The expectation value of these snapshots over all choices of Clifford unitaries and measurement outcomes can be viewed as a quantum channel. In particular, since the Clifford group forms a unitary 3-design, the process of randomly sampling Clifford operations corresponds to the following depolarizing channel~\cite{Huang_2020}:
\begin{align}
    D_{1/(2^{\eta} + 1)}(\rho) &= \frac{\rho}{2^{\eta} + 1} + \left(1 - \frac{1}{(2^{\eta} + 1)}\right)\frac{{\rm Tr}(\rho)\Id}{2^{\eta} } \nonumber \\
    &= \frac{\rho + {\rm Tr}(\rho)\Id}{2^{\eta} + 1}
\end{align}
This channel can be inverted in the following way:
\begin{align}
    D_{1/(2^{\eta} + 1)}^{-1}(\rho) &= D_{2^{\eta} + 1}(\rho) \nonumber \\
    &= (2^{\eta} + 1)\rho - {\rm Tr}(\rho)\Id \label{eq: inverseDepolarizingChannel}
\end{align}
where the first line uses the fact that $D_a(D_b(\rho)) = D_{ab}(\rho)$ (this is easily proven using the definition of a depolarizing channel), and the second line follows directly from the definition of a depolarizing channel.

The classical shadow of $\rho$ is then defined as the collection of $N$ snapshots after applying this inverse channel to each of them, as in Definition~\ref{Def:ClassicalShadowClifford}.
\begin{definition}[Classical Shadow (From Clifford Measurements)]\label{Def:ClassicalShadowClifford}
    Given $N$ copies of an $\eta$-qubit quantum state $\rho$, the classical shadows procedure based on random Clifford measurements returns a classical shadow of $\rho$ which is of the form
    \begin{align}
        \hat{\rho} &= \{\hat{\rho}_i | i \in \mathbb{Z}_N\},
    \end{align}
    where
    \begin{align}
        \hat{\rho}_i = (2^{\eta} + 1)(U_i^\dagger \ket{b_i}\bra{b_i} U_i) -\Id
    \end{align}
    is the result of applying the inverse of the depolarizing channel (equation~\eqref{eq: inverseDepolarizingChannel}) to the $i^{th}$ classical snapshot.
\end{definition}
As shown by Huang et al. \cite{Huang_2020}, the classical shadow can be used to efficiently predict expectation values of operators on the state $\rho$, as required by the Shadow Tomography Problem. 

Note that this procedure requires randomly sampling Clifford circuits, which is not as easy as sampling single-qubit operators such as Pauli operators. However, recent work has shown that this process can still be done efficiently. In particular, Bravyi and Maslov~\cite{Bravyi_2021} and van den Berg~\cite{vandenBerg_2021} have proposed algorithms for sampling Clifford circuits uniformly at random with time complexity $\mathcal{O}(n^2)$ and circuit depth $\mathcal{\widetilde{\mathcal{O}}}(n)$. The latter algorithm outputs a circuit directly and allows for parallelism, which effectively reduces the time complexity to $\mathcal{O}(n)$, but requires a quantum computer with fully connected topology to achieve the claimed circuit depth~\cite{vandenBerg_2021}. Meanwhile, the former approach achieves a two-qubit gate depth of $9n$ on the much simpler Linear Nearest Neighbour architecture~\cite{Bravyi_2021}.

Another variation of the classical shadows procedure replaces the random Clifford operators from the group $\text{Cl}(2^\eta)$ with tensor products of single-qubit Clifford operators (i.e. operators from the group $\text{Cl}(2)^{\otimes \eta}$). This process is referred to as the ``random Pauli measurement'' version of classical shadows, because the single-qubit Clifford group is generated by the Hadamard and phase gates, and these operations enact transformations between the Pauli X, Y, and Z bases. Therefore, allowing only single-qubit Clifford operations and measurement in the computational ($Z$) basis is equivalent to only allowing measurements in the Pauli X, Y, and Z bases. In the next section, we focus on the use of the Clifford measurement version of classical shadows, and leave discussion about the use of the Pauli version to Appendix~\ref{App: PauliShadows}.

\subsection{Efficient Estimation of the Hamiltonian Coefficients (Clifford Shadows)\label{subsection:LearnCoeffs}}

In this section we discuss using the Clifford measurement version of classical shadows to extract the Hamiltonian coefficients from the pseudo-Choi state. We refer to the inference of these coefficients as the Pseudo-Choi Hamiltonian Learning problem (PC HLP). Note that the Pauli measurement flavour of classical shadows can also be used, and this approach is included in Appendix~\ref{App: PauliShadows}. While the Pauli version only requires single-qubit operations and can reduce the computational complexity of the classical post-processing, the Clifford version has the better sample complexity for k-local systems. However, although the gap between the two sample complexity upper bounds grows exponentially with k, the benefits of the Pauli approach may outweigh the increased sample complexity for small values of k.

Since classical shadows tomography allows for efficient prediction of linear functions of operators, we can solve the PC HLP by choosing a specific set of operators whose expectation values correspond to the Hamiltonian coefficients. For Clifford-based shadows, these ``decoding operators'' are given in Definition~\ref{Def: CliffDecodingOperators}.

\begin{definition}[Decoding Operators] \label{Def: CliffDecodingOperators}
    Let the set of decoding operators be defined as
    \begin{align}
        \boldsymbol{O} & := \{ O_l | l \in \mathbb{Z}_{M} \} 
    \end{align}
    where
    \begin{align}
         O_l & := \left(H_l \otimes \Id\right) \ket{\Phi_d}_{SA} \bra{\Phi_d}_{SA} \otimes \ket{0}_C\bra{1}_C
    \end{align}
    and $H_l$ is one of the $M$ Hamiltonian terms. 
    
    Furthermore, we define
    \begin{align}
        O_{\alpha} & := \ket{\Phi_d}_{SA} \bra{\Phi_d}_{SA} \otimes \ket{1}_C\bra{1}_C
    \end{align}
\end{definition}

Recalling Definition~\ref{def:pseudoChoi}, the density matrix representing the pseudo-Choi state is given by
\begin{align}
    \rho_c  :=& \ket{\psi_c} \bra{\psi_c} \label{eq:ChoiDensityOp} \nonumber \\
    \begin{split}
         = & \frac{1}{\alpha^2}\Big((H \otimes \Id_{A}) \ket{\Phi_d}_{SA}\bra{\Phi_d}_{SA}(H \otimes \Id_{A}) \otimes \ket{0}_C\bra{0}_C \\
        & + (H \otimes \Id_{A}) \ket{\Phi_d}_{SA}\bra{\Phi_d}_{SA}\otimes \ket{0}_C\bra{1}_C  \\
        & + \ket{\Phi_d}_{SA}\bra{\Phi_d}_{SA} (H \otimes \Id_{A}) \otimes \ket{1}_C\bra{0}_C\\
        & +  \ket{\Phi_d}_{SA}\bra{\Phi_d}_{SA}\otimes \ket{1}_C\bra{1}_C \Big)
    \end{split}
\end{align}

We show below in Proposition~\ref{proposition: ComputeHamCoeffs} that the expectation values of the decoding operators acting on the pseudo-Choi state give the Hamiltonian coefficients up to a constant factor. Note that the Hamiltonian terms must be known in order to construct the decoding operators. As this work focuses on k-local Hamiltonians, this is not a problem, since there are at most $\mathcal{O}(n^k)$ possible terms.

\begin{proposition}[Computing the Hamiltonian Coefficients\label{proposition: ComputeHamCoeffs}]
If $\rho_c$ is the pseudo-Choi state \eqref{eq:ChoiDensityOp}, and $O_l$ is a decoding operator as defined in Definition~\ref{Def: CliffDecodingOperators} then 
\begin{align}
    {\rm Tr}(\rho_c O_{l}) = \frac{c_l}{\alpha^2},
\end{align}
where $c_l$ is the $l^{th}$ Hamiltonian coefficient, and $\alpha$ is the normalization constant of the pseudo-Choi state \eqref{eq:alpha}. Furthermore,
\begin{align}
    {\rm Tr}(\rho_c O_{\alpha}) = \frac{1}{\alpha^2}
\end{align}

\end{proposition}
\begin{proof}

From the definition of the decoding operators and the pseudo-Choi state, it is easy to see that
\begin{align}
    {\rm Tr}(\rho_c O_l)  &= \frac{\bra{\Phi_d}_{SA}(H \otimes \Id) \left(H_l \otimes \Id\right) \ket{\Phi_d}_{SA} }{\alpha^2} \label{eq:randomlabel5}
\end{align}

Expanding the above using the representation of the Hamiltonian from Definition~\ref{def:HamLearningProblem} gives
\begin{align}
    {\rm Tr}(\rho_c O_l) & = \frac{\bra{\Phi_d}_{SA} \left(\sum_{m} c_m H_m \otimes \Id\right) \left(H_l \otimes \Id \right) \ket{\Phi_d}_{SA}}{\alpha^2} \nonumber \\
    & = \frac{1}{\alpha^2}  \sum_{m} c_m \bra{\Phi_d}_{SA} \left( H_m \otimes \Id\right) \left(H_l \otimes \Id \right) \ket{\Phi_d}_{SA}\label{eq:randomlabel6}
\end{align}

We can now use the definition of $\ket{\Phi_d}_{SA}$ to expand the term inside the summation:
\begin{align}
    \bra{\Phi_d}_{SA} \left(H_m \otimes \Id \right)  \left(H_l \otimes \Id\right) \ket{\Phi_d}_{SA} &= \left(\frac{1}{\sqrt{d}}\sum_i \bra{i}\bra{i}\right)(H_m \otimes \Id)(H_l \otimes \Id)\left(\frac{1}{\sqrt{d}}\sum_j \ket{j}\ket{j}\right) \nonumber \\
    & = \frac{1}{d}\sum_i \sum_j \bra{i}H_m H_l\ket{j} \bra{i}\ket{j} \nonumber \\
    & = \frac{1}{d}\sum_i \bra{i}H_m H_l\ket{i} \nonumber \\
    & = \frac{\text{Tr}(H_m H_l)}{d}\label{eq:randomlabel7}
\end{align}

Finally, we substitute \eqref{eq:randomlabel7} back in to \eqref{eq:randomlabel6} to get
\begin{align}
    {\rm Tr}(\rho_c O_l) & = \frac{1}{\alpha^2}\sum_{m} c_m  \frac{\text{Tr}(H_m H_l)}{d} \nonumber \\
    & = \frac{1}{\alpha^2}\sum_{m} c_m  \delta_{ma} \nonumber \\
    & = \frac{c_a}{\alpha^2}
\end{align}

The second claim (${\rm Tr}(\rho_c O_{\alpha}) = \frac{1}{\alpha^2}$) follows immediately from the definition of the pseudo-Choi state \eqref{eq:ChoiDensityOp} and the definition of $O_{\alpha}$. 

\end{proof}


We seek to approximate the expectation values of Proposition~\ref{proposition: ComputeHamCoeffs} using the classical shadows approach introduced in \cite{Huang_2020}. However, the classical shadows procedure requires Hermitian operators, which $O_l$ are not. Instead, we can use the following pair of Hermitian operators:
\begin{align}
    O_l^{+} &:= O_l + \left(O_l\right)^\dagger\\
    O_l^{-} &:= iO_l - i\left(O_l\right)^\dagger
\end{align}

These two operators are Hermitian, and can be inverted to find that
\begin{align}
    \frac{O_l^{+} - iO_l^{-}}{2} & = O_l
\end{align}

Therefore, the expectation value of $O_l$ can be predicted by using classical shadows to estimate the expectation values of $O_l^{+}$ and $O_l^{-}$:
\begin{align}
    \avg{O_l}_s &= \frac{\avg{O_l^{+}}_s - \avg{iO_l^{-}}_s}{2} \label{eq:O_l expectation}
\end{align}
where the subscript $s$ denotes that these are not the true expectation values, but estimates generated by the classical shadows procedure.


We now consider using the Clifford measurement flavor of classical shadows for the pseudo-Choi state. Since the pseudo-Choi state is on $2n+1$ qubits, its classical shadow is made up of $N$ components of the form
\begin{align}
    \hat{\rho}_i &= \left( 2^{2n+1} + 1\right) U^\dagger_i \ket{b_i}\bra{b_i} U_i - \Id
\end{align}

However, it is not necessary to store the shadow in this explicit form. Instead, it is helpful to use the form of the classical shadow specified in Definition~\ref{Def: compressedShadowCliff}, and introduce the $\Id$ term and the factor of $2^{2n+1} + 1$ as needed.
\begin{definition}\label{Def: compressedShadowCliff}
    If the random Clifford measurement version of the classical shadows procedure is performed, the resulting classical shadow of size N can be stored in the following way:
    \begin{align}
        \hat{\rho} &= \{\ket{\hat{\rho}_i} | i \in \mathbb{Z}_N\},
    \end{align}
    where
    \begin{align}
        \ket{\hat{\rho}_i} &= U_i^\dagger \ket{b_i}
    \end{align}
\end{definition}

The full procedure for estimating the vector of Hamiltonian coefficients, \textbf{c}, is given in Algorithm~\ref{algorithm: FindCoeffClifford}.  The algorithm also invokes a further subroutine (Algorithm~\ref{algorithm: ComputeExpectation}) for determining the expectations using classical shadows. Algorithms~\ref{algorithm: FindCoeffClifford} and~\ref{algorithm: ComputeExpectation} do not use the explicit form of the classical shadow given by Definition~\ref{Def:ClassicalShadowClifford}, but instead use the ``compressed'' version given by Definition~\ref{Def: compressedShadowCliff}.

Previous work by Aaronson and Gottesman provides an algorithm for simulating Clifford circuits classically in polynomial time, as well as an algorithm for calculating the inner product between two stabilizer states in $\mathcal{O}\left(n^3\right)$ time \cite{Aaronson_2004}. For future reference, we shall refer to this second algorithm (see the end of Section III in \cite{Aaronson_2004}) as the Stabilizer Inner Product (SIP) Algorithm. Algorithm~\ref{algorithm: FindCoeffClifford} invokes the former of these algorithms, while Algorithm~\ref{algorithm: ComputeExpectation} leverages the SIP Algorithm to compute expectation values in polynomial time.

\begin{algorithm}[t!]
\caption{$\mathtt{FindCoeffClifford}(\rho^{\otimes N}_c, \boldsymbol{H}, n, N)$: Determine the Vector of Hamiltonian Coefficients from $\rho_c$ using random Clifford measurement based classical shadows}\label{algorithm: FindCoeffClifford}
\textbf{Input:} 

$\rho_c^{\otimes N} $: Collection of $N$ pseudo-Choi states on $2n + 1$ qubits, each on Hilbert space $\mathcal{H}_S \otimes \mathcal{H}_A \otimes \mathcal{H}_C$

$\boldsymbol{H}$: Set of M k-local Hamiltonian terms $H_l$ whose coefficients $c_m$ are desired

\Comment{Note: $\rho^{\otimes N}_c$ is a quantum state, whereas $\boldsymbol{H}$ is a set of classical operators}

$n$: Number of qubits in the system on Hilbert space $\mathcal{H}_S$

$N$: Number of pseudo-Choi states

\begin{enumerate}
    \item \textbf{Initialize the array of Hamiltonian coefficients}
    
    Let $\hat{\textbf{c}} \gets \begin{bmatrix} 0, 0, ..., 0 \end{bmatrix}_{1\times M}$ and let $\hat{\textbf{c}}_l$ denote the $l^{th}$ element of $\hat{\textbf{c}}$
    
    \item \textbf{Generate and store the classical shadow}
    
    Use $N$ copies of $\rho_c$ to generate $\hat{\rho}$, a classical shadow of size $N$ of $\rho_c$, using the classical shadows procedure based on random Clifford measurements from \cite{Huang_2020}.
    
    Store $\hat{\rho}$ in classical memory (as per Definition~\ref{Def: compressedShadowCliff}) using the tableau algorithm introduced in \cite{Aaronson_2004} to compute and store each stabilizer state $\ket{\rho_i} = U^\dagger_i \ket{b_i}$
    
    \item \textbf{Generate an estimate of ${\rm Tr}(\rho O_{l}) = \frac{c_l}{\alpha^2}$ for each corresponding Hamiltonian term}
 
    $\mathtt{For}$ $l$ $\mathtt{from}$ $0$ $\mathtt{to}$ $M-1$:

        \qquad $\hat{\textbf{c}}_l \gets \mathtt{ComputeExpectation}(\hat{\rho}, H_l, N, l, n)$ \Comment{Estimate of $\frac{c_l}{\alpha^2}$}

    \item \textbf{Generate an estimate of ${\rm Tr}(\rho O_{\alpha}) = \frac{1}{\alpha^2}$}
    
    $\hat{o}_{\alpha} \gets \mathtt{ComputeExpectation}(\hat{\rho}, O_{\alpha}, N, -1, n)$ \Comment{Estimate of $\frac{1}{\alpha^2}$}
    
    \item \textbf{Remove the factor of $\frac{1}{\alpha^2}$ from the Hamiltonian coefficients}
    
    Let $\hat{\textbf{c}} \gets \frac{\hat{\textbf{c}}}{\hat{o}_{\alpha}}$ 
    
    \item \textbf{Output the vector of Hamiltonian coefficients}
    
    Return $\hat{\textbf{c}} = 
    \begin{bmatrix} \hat{c}_0, \hat{c}_1, ..., \hat{c}_{M-1} \end{bmatrix}$
\end{enumerate}
\end{algorithm}

Note that in order for the output of Algorithm~\ref{algorithm: FindCoeffClifford} to solve the PC HLP, the number of samples of $\rho_c$ must be at large enough for the error and success probability to satisfy the requirements of the Hamiltonian Learning Problem~\ref{def:HamLearningProblem}. This sufficient value is given later as Theorem~\ref{thm:samplecomplexityUB}.

To go along with Algorithm~\ref{algorithm: ComputeExpectation}, it is also helpful to understand how the expectation values involved are explicitly computed using a classical shadow. Proposition~\ref{Prop: EfficientTrace} gives expressions for the three expectation values of interest. These expressions contain several inner products which can be efficiently computed using the SIP algorithm of \cite{Aaronson_2004}.
\begin{proposition}[] \label{Prop: EfficientTrace}

Let $\hat{\rho}_i = \left( 2^{2n+1} + 1\right) U^\dagger_i \ket{b_i}\bra{b_i} U_i - \Id$ be a classical shadow generated from the random Clifford measurement version of the classical shadows procedure, and consider the operators
\begin{align*}
    O_l^{+} &= O_l + \left(O_l\right)^\dagger\\
    O_l^{-} &= iO_l - i\left(O_l\right)^\dagger
\end{align*}
where $O_l$ is a decoding operator as in Definition~\ref{Def: CliffDecodingOperators}. 

We then have that

\begin{align}
	\begin{split}
		{\rm Tr}\left(\hat{\rho_i} O_l^+\right)=&  \left( 2^{2n+1} + 1 \right) \Big(\bra{b_i} U_i  \left(H_l \otimes \Id_A \right) \ket{\Phi_d}_{SA} \ket{0}_C\Big) \Big( \bra{\Phi_d}_{SA}  \bra{1}_C    U_i^\dagger \ket{b_i}\Big) \\
		&+  \left( 2^{2n+1} + 1 \right)\Big( \bra{b_i} U_i  \ket{\Phi_d}_{SA} \ket{1}_C\Big) \Big( \bra{\Phi_d}_{SA}\left(H_l \otimes \Id_A \right)  \bra{0}_C    U_i^\dagger \ket{b_i}\Big),
	\end{split}\label{eq:EfficientTrace1}
\end{align}
and similarly,
\begin{align}
    \begin{split}
		{\rm Tr}\left(\hat{\rho_i} O_l^-\right)=&  i\left( 2^{2n+1} + 1 \right) \Big(\bra{b_i} U_i  \left(H_l \otimes \Id_A \right) \ket{\Phi_d}_{SA} \ket{0}_C\Big) \Big( \bra{\Phi_d}_{SA}  \bra{1}_C    U_i^\dagger \ket{b_i}\Big) \\
		&-i  \left( 2^{2n+1} + 1 \right)\Big( \bra{b_i} U_i  \ket{\Phi_d}_{SA} \ket{1}_C\Big) \Big( \bra{\Phi_d}_{SA}\left(H_l \otimes \Id_A \right)   \bra{0}_C    U_i^\dagger \ket{b_i}\Big) \\
	\end{split}\label{eq:EfficientTrace2}
\end{align}
Furthermore, 
\begin{align}
    {\rm Tr}( \hat{\rho}_i O_\alpha) & = \left( 2^{2n + 1} + 1 \right) \abs{\bra{b_i}U_i \ket{\phi_d}_{SA}\ket{1}_C}^2 - 1 \label{eq:EfficientTrace3}
\end{align}

Note that $U_i^\dagger \ket{b_i}$ is on the Hilbert space $\mathcal{H}_S \otimes\mathcal{H}_A \otimes\mathcal{H}_C$, so all the inner products above give scalar values.
\end{proposition}

The proof of Proposition~\ref{Prop: EfficientTrace} can be found in Appendix~\ref{section:EfficientTrace}. As previously mentioned, $U_i \ket{b_i}$ is a stabilizer state. Furthermore, $H_l \ket{\Phi_d}_{SA}$ is also a stabilizer state for all $l \in M$, since the Hamiltonian terms are Pauli operators, and therefore also Clifford operators. From this, we see that calculating the expectation values in Proposition~\ref{Prop: EfficientTrace} comes down to computing a few inner products between stabilizer states, which we can do on a classical computer. Algorithm~\ref{algorithm: ComputeExpectation} computes these inner products using the SIP Algorithm of \cite{Aaronson_2004}, and then performs a median of means (as prescribed by the classical shadows procedure \cite{Huang_2020}) to reduce the variance of the estimate.

\begin{algorithm}
\caption{$\mathtt{ComputeExpectation}(\hat{\rho}, O, N, l, n)$: Compute Expectation Values with a Classical Shadow}\label{algorithm: ComputeExpectation}

\textbf{Input:} 

$\hat{\rho} $: Classical shadow of size $N$ of the resource state $\rho$, in the compressed form of Definition~\ref{Def: compressedShadowCliff}

$H_l$: Hermitian operator from the set $\boldsymbol{H}$

$N$: The size of the classical shadow $\hat{\rho}$

$l$: Integer indicating which of the $M$ Hamiltonian terms $O$ is. $l=-1$ indicates that the algorithm should use the formula for estimating the normalization constant $\alpha$

$n$: Number of qubits in the system

\begin{enumerate}
    \item \textbf{Iterate over the components ($\ket{\hat{\rho}_i} = U_i^\dagger \ket{b_i}$) of $\hat{\rho}$ and compute ${\rm Tr}\left(\hat{\rho_i} O\right)$ for each one}
        Let $\hat{\textbf{o}} \gets \begin{bmatrix} 0, 0, ..., 0 \end{bmatrix}_{1\times N}$ and let $\hat{\textbf{o}}_i$ denote the $i^{th}$ element of $\hat{\textbf{o}}$
        
        $\mathtt{if}$ $l==-1 \mathtt{:}$
        
        \qquad $\mathtt{for}$ $i$ $\mathtt{from}$ $0$ to $N-1$:
        
        \qquad \qquad Compute ${\rm Tr}\left(\hat{\rho_i} O_\alpha\right)$ using the SIP Algorithm given in \cite{Aaronson_2004} to compute the 
        
        \qquad \qquad inner product in equation \eqref{eq:EfficientTrace3}
        
        $\mathtt{else:}$
        
        \qquad $\mathtt{for}$ $i$ $\mathtt{from}$ $0$ to $N-1$:
        
        \qquad \qquad Compute ${\rm Tr}\left(\hat{\rho_i} O_l^+\right)$ using the SIP Algorithm given in \cite{Aaronson_2004} to compute the four  
        
        \qquad \qquad inner products in equation \eqref{eq:EfficientTrace1}
        
        \qquad \qquad Compute ${\rm Tr}\left(\hat{\rho_i} O_l^-\right)$ using the SIP Algorithm given in \cite{Aaronson_2004} to compute the four  
        
        \qquad \qquad inner products in equation \eqref{eq:EfficientTrace2}
        
        \qquad \qquad $\hat{o}_i \gets \frac{{\rm Tr}\left(\hat{\rho_i} O_l^+\right) - i{\rm Tr}\left(\hat{\rho_i} O_l^-\right) }{2}$
        \Comment{As per Equation \eqref{eq:O_l expectation}}

    \item \textbf{Perform median of means protocol and output the estimate of ${\rm Tr}(\rho_c O_{l}) = \frac{c_l}{\alpha^2}$}
    
    Return $\mathtt{MedianOfMeans}(\hat{\textbf{o}})$

\end{enumerate}
\end{algorithm}

\begin{algorithm}
\caption{$\mathtt{MedianOfMeans}(\hat{\textbf{o}})$:}\label{algorithm: MoM}

\textbf{Input:} 

$\hat{\textbf{o}} $: Vector of size $N$ 

\bigbreak

\textbf{Perform median of means protocol:}

Separate the elements of $\hat{\textbf{o}}$, $\hat{o}_i$, into $K$ groups, $G_k$, of equal size

For $k$ from $0$ to $K-1$:

\qquad $\hat{o}_k \gets \textbf{mean}\left( G_k \right)$

$\hat{o}_l \gets \textbf{median}\left(\hat{o}_0, \hat{o}_1, ..., \hat{o}_{K-1} \right)$

Return $\hat{o}_l$

\end{algorithm}

The SIP Algorithm given in \cite{Aaronson_2004} requires $\mathcal{O}(n^3)$ time to compute inner products between stabilizer states, and Algorithm~\ref{algorithm: ComputeExpectation} calls it $\mathcal{O}(1)$ times for each of the $N$ snapshots (see Theorem~\ref{thm:samplecomplexityUB} for the sufficient value of $N$). In turn, Algorithm~\ref{algorithm: FindCoeffClifford} calls Algorithm~\ref{algorithm: ComputeExpectation} for each of the $M+1$ operators, giving a total time complexity of $\mathcal{O}(MNn^3)$ for computing all the inner products between stabilizer states. However, all these calculations are independent and can be parallelized over both the operators and the classical snapshots. Thus, with a large classical memory, the computation time required can be greatly reduced through massive parallelism.


\subsubsection{Sample Complexity Upper Bound} \label{sec: SampleCompUB}
The main result in this section is Theorem~\ref{thm:samplecomplexityUB}, a proof of the sample complexity of solving the Hamiltonian learning problem using pseudo-Choi states as resources.  Our result depends on the concept of a shadow norm which we define below for reference.
\begin{definition} [Shadow Norm for Clifford-Based Shadows]{
    Let $O$ be a finite-dimensional operator. The shadow norm of the operator $O$ is defined as
    \begin{align*}
        \norm{O}_{shadow} = \max_{\sigma}\left( \mathbb{E}_{U \in \mathcal{C}} \sum_{b \in \{0,1\}^n} \bra{b} U\sigma U^\dagger \ket{b} \bra{b} U \mathcal{M}^{-1}(O) U^\dagger\ket{b}^2 \right)^{1/2}
    \end{align*}
    where $\mathcal{C}$ is the Clifford group and the maximization is over all quantum states.
}
\end{definition}

Note that the definition of the shadow norm above assumes the classical shadows process is randomly sampling Clifford unitaries, and not unitaries from some other ensemble.

In order to state our proof of Theorem~\ref{thm:samplecomplexityUB} we need to first state Lemma~\ref{lem:shadownormUB} which bounds the shadow norm, as well as Proposition~\ref{prop:errbd} which gives a value of $\epsilon_s$ that ensures the total error of the learning process is at most $\epsilon$, as required by the Hamiltonian Learning Problem (Definition~\ref{def:HamLearningProblem}). 
\begin{lemma} [Shadow Norm Upper Bounds\label{lem:shadownormUB}]{
    Let the set \{$O_{\alpha}$\} $\cup$  $\{O_l^{+}$, $O_l^{-} \mid l \in \mathbb{Z}_{M}\}$ be denoted by $\textbf{O} \equiv \{O_i \mid i \in \mathbb{Z}_{2M}\}$. If using the Clifford measurement version of the classical shadows procedure to predict the expectation values of the operators in the set, the shadow norm of each operator is at most 
    \begin{align*}
        \norm{O_{i}}^2_{shadow} \leq 3{\rm Tr}\left(\left(O_{i}\right)^2\right) & \leq 6
    \end{align*}
}
\end{lemma}
The proof of Lemma~\ref{lem:shadownormUB} can be found in Appendix~\ref{App:shadownormUB}.

\begin{proposition}\label{prop:errbd}
To solve the Hamiltonian learning problem with an error of at most $\epsilon$, it serves to let the error in the classical shadows procedure be
\begin{align*}
    \epsilon_s & = \frac{\epsilon}{\alpha^2 \sqrt{c^2_{\rm max} +1} \sqrt{M}}, 
\end{align*}
where $c_{\rm max}$ denotes the Hamiltonian coefficient with the largest absolute value. 
\label{proposition: esUB}
\end{proposition}
The proof of Proposition~\ref{prop:errbd} can be found in Appendix~\ref{App:esUB}.

With these results in hand, we can prove the following result which provides the scaling of the number of resource states that are needed to solve our learning problem.
\begin{theorem} [Sample Complexity Bound for PC HLP\label{thm:samplecomplexityUB}]{

The number of copies of the pseudo-Choi state  \eqref{eq:choistate} required by Algorithm~\ref{algorithm: FindCoeffClifford} to solve the Hamiltonian Learning Problem given by Definition~\ref{def:HamLearningProblem} is 
\begin{align}\label{eq:Nbd}
    N \in \mathcal{O}\left( \frac{\alpha^4 (c_{\rm max} + 1) M\log\left(M/\delta\right) }{\epsilon^2}  \right)
\end{align}

}
\end{theorem}

\begin{proof}[Proof of Theorem~\ref{thm:samplecomplexityUB}]

The sample complexity of the shadow tomography procedure depends on the number of expectation values that are to be predicted. Recall that $M$ is the number of terms in the Hamiltonian. Algorithm~\ref{algorithm: FindCoeffClifford} requires estimating $M$ expectation values $\avg{O_{l}}$ (and each $\avg{O_{l}}$ requires predicting the expectation values of $2$ operators), as well as $\avg{O_{\alpha}}$. Therefore, the total number of operators whose expectation values are to be estimated via shadow tomography is $2M+1$. Let the set of these operators (\{$O_{\alpha}$\} $\cup$  $\{O_l^{+}$, $O_l^{-} \mid l \in \mathbb{Z}_{M}\}$ ) be denoted by $\textbf{O} \equiv \{O_i \mid i \in \mathbb{Z}_{2M}\}$.

Using classical shadows, this gives a sample complexity of 
\begin{align}
    N_s & = \mathcal{O}\left( \frac{\log\left(M/\delta_s\right)}{\epsilon_s^2} \max_i \norm{O_i}^2_{shadow} \right)
\end{align}
to estimate all the expectation values, each within sampling error $\epsilon_s$ and with a total probability of failure at most $\delta$. \cite{Huang_2020}

Since we use the random Clifford measurement version of the classical shadows procedure, Lemma~\ref{lem:shadownormUB} implies that
\begin{align}
    \max_i \norm{O_i}^2_{shadow} & \leq 6
\end{align}
This result implies that the sample complexity for predicting the expectation values in Algorithm~\ref{algorithm: FindCoeffClifford} within an error  $\epsilon_s$ with probability $\delta_s$ is
\begin{align}
    N_s & = \mathcal{O}\left( \frac{\log\left(M/\delta_s\right)}{\epsilon_s^2}  \right)\label{eq:NsUB}.
\end{align}

We would like to determine an upper bound on the number of samples such that with probability $1- \delta$, the $l$-$2$ error in the vector of Hamiltonian coefficients is at most $\epsilon$, as per the Hamiltonian learning problem~\ref{def:HamLearningProblem}. The total probability of failure of Algorithm~\ref{algorithm: FindCoeffClifford} is that of the classical shadows procedure, so $\delta = \delta_s$. Additionally, to minimize the number of samples needed while still ensuring that the total error in the Hamiltonian learning process is at most $\epsilon$, it serves to choose the value of $\epsilon_s$ given as per Proposition~\ref{prop:errbd}. Combining this value of $\epsilon_s$ with Equation \eqref{eq:NsUB}, we arrive at the claimed result.
\end{proof}
It is worth noting that if coherent quantum queries were made to the quantum state preparation routine then we can learn the coefficients with quadratically better scaling with $\epsilon$ than the above procedure using results by Huggins et al.~\cite{Huggins_nearly_optimal_estimating}. However, since doing so necessitates the introduction of a more powerful oracle that allows coherent queries we ignore such optimizations here, and instead revisit this approach in Section~\ref{section: QME}.


\section{Learning the Hamiltonian from a Unitary Time Evolution Blackbox}\label{section:BlackboxLearning}

While the pseudo-Choi state might seem like a strange choice to use as the learning resource, the time evolution blackbox of the form
\begin{align}
    U & = e^{-iHt}
\end{align}
is perhaps a more intuitive choice of learning resource, and has been recently studied in the context of Hamiltonian learning \cite{haah2021optimal,yu2022practical, PhysRevLett.130.200403, gu2022practical}. The goal of this section is to demonstrate that by querying this time evolution unitary and its inverse for $\abs{t} \leq \frac{\pi}{2 \norm{H}}$, we can produce states similar to the pseudo-Choi states of the previous section, and as a result learn the Hamiltonian using a procedure very similar to the one in Algorithm~\ref{algorithm: FindCoeffClifford}. This leads to Theorem~\ref{thm:querycomplexityUB}, which is the main result of this section and gives an upper bound on the number of queries to $U$ and $U^\dagger$ needed to solve the Hamiltonian learning problem.

The process of generating the pseudo-Choi state contains three main steps. First, the time evolution black-box is used to generate a unitary matrix that encodes the Hamiltonian (referred to as a block-encoding of the Hamiltonian). Second, this block-encoding is applied in a controlled fashion to a maximally entangled state ($\ket{\phi_d}_{SA}$) and an ancilla qubit. This generates a relative phase between a term that applies the block-encoding to the input state and a term that acts as the identity on it. Finally, the pseudo-Choi state is produced probabilistically by measuring the ancilla qubit in the computational basis. The outcome of this measurement indicates whether or not the state was successfully generated.

An important issue that arises here is that we need to assume that the time evolution can be implemented controllably.  This is often the case when we are considering learning a Hamiltonian applied within a simulator, but may not necessarily be straight forward to implement when learning a physical Hamiltonian that cannot be so manipulated.  In such cases, we can still perform the control using controlled swap operations if a $+1$ eigenstate of the unitary is known.  Further, it is worth noting that in practice we also need the inverse of the evolution using our approach.  If these assumptions are met, however, we are able to show that we can efficiently generate the learning resource. Additionally, there exist classes of Hamiltonians (``time-reversible Hamiltonians'') for which the backward time evolution can be performed efficiently even if one only has access to a black-box that performs the forward evolution. For example, certain Heisenberg models that can be represented as bipartite 
- the simplest being a 1-dimensional Ising model - can be easily time-reversed given access to the time evolution black-box. Therefore, if a Hamiltonian can be efficiently broken up into relatively few such Hamiltonians, our approach remains feasible as one can learn the time-reversible Hamiltonians one by one until all the coefficients of the original Hamiltonian are recovered.


\subsection{Generating pseudo-Choi States}\label{section: GenerateChoiState}

As mentioned above, the process of producing the pseudo-Choi state begins by using the time evolution unitary to create a block-encoding of the Hamiltonian. Recall that the pseudo-Choi state consists of three subsystems. The map containing the Hamiltonian acts on the subsystem $S$ (denoted by the subscript $S$), which resides in a $d$ dimensional Hilbert space $\mathcal{H}_S$. Subsystem $A$ is an ancilla of the same size as $S$. Lastly, subsystem $C$ is a single qubit used to retain information about the overall sign and norm of the Hamiltonian. In addition to these three, we will now require an additional subsystem, $B$, which is a single qubit used for block-encoding purposes, and will be referred to as the ``block-encoding qubit''.  We formalize the notion of block-encoding below.



\begin{definition}[Block-Encoding]
Let $U_{block}$ be a unitary operator acting on the Hilbert space $\mathcal{H}_B\otimes \mathcal{H}_S$. We then say that this unitary is a block-encoding (or more specifically a $(\Delta, {\rm dim}(\mathcal{H}_B),0)$-block-encoding as defined in Section 4 in~\cite{Gily_n_2019}) of a Hamiltonian $H$ if
\begin{align}
                        \frac{H}{\Delta}=( \bra{0}_B \otimes \Id_S) U_{block} (\ket{0}_B \otimes \Id_S)\label{eq:blockencoding}
\end{align}
where 
$\Delta$ is a normalization factor.
\end{definition}

In our context, this unitary block-encoding plays the same role as a unitary that prepares a purified Gibbs state in other approaches to quantum Hamiltonian learning. Rather than a Gibbs state, $U_{block}$ is used to prepare a pseudo-Choi state, from which we can then recover the Hamiltonian coefficients as in Algorithms~\ref{algorithm: FindCoeffClifford} and~\ref{algorithm: ComputeExpectation} in Section~\ref{subsection:LearnCoeffs}.

\textit{Gilyén et al.} \cite{Gily_n_2019} developed a method for generating a block-encoding of the logarithm of an operator, which is precisely what we seek to do with the time evolution operator. Their results, modified to fit our problem, are summarized in Lemma~\ref{lemma:ProduceBlockEncoding} below.

\begin{lemma} [Producing a Block-Encoded Hamiltonian from a Time Evolution Unitary~
\cite{Gily_n_2019} \label{lemma:ProduceBlockEncoding}]{
Let $U=e^{-iHt}$ where the H is Hermitian and $\norm{Ht} \leq \frac{1}{2}$. Let $\epsilon_b \in (0, \frac{1}{2}]$ be the block-encoding error. Using 
\begin{align*}
    \mathcal{O}\left(\log(1/\epsilon_b)\right)
\end{align*}
controlled queries to $U$ and $U^{-1}$, as well as $\mathcal{O}\left(\log(1/\epsilon_b)\right)$ 2-qubit gates and one ancilla qubit, a block-encoding of the Hamiltonian of the form
\begin{align}
    U_{block} & = \begin{bmatrix}
                        \frac{2\widetilde{H}t}{\pi} & I\\
                        J & K
                        \end{bmatrix}
\end{align}

can be produced such that
\begin{align}
    \norm{Ht - \widetilde{H}t} &\leq \epsilon_b
\end{align}
}
\end{lemma}

The proof of Lemma~\ref{lemma:ProduceBlockEncoding} can be found in \cite{Gily_n_2019} (Lemma~\ref{lemma:ProduceBlockEncoding} is identical to Corollary 71 in \cite{Gily_n_2019}, only with $H$ replaced by $\widetilde{H}t$), but the general idea is that $\sin(Ht) = \frac{e^{iHt}- e^{-iHt}}{2i}$ can be implemented via Linear Combination of Unitaries (LCU) techniques, resulting in a block-encoding of $\sin(Ht)$. The Quantum Singular Value Transform (QSVT) can then be used to construct a polynomial that approximates the $\arcsin$ of this quantity, which results in an approximate block-encoding of $\frac{2Ht}{\pi}$.
This can be achieved using either quantum singular value transformation methods or linear combinations of unitaries~\cite{Gily_n_2019,childs2012hamiltonian,van2020quantum}



Using the block-encoding of the Hamiltonian from Lemma~\ref{lemma:ProduceBlockEncoding}, we can now generate pseudo-Choi states, as per Lemma~\ref{lemma: ProduceChoiState}. Because the Hamiltonian in the block-encoding contains a factor of $\frac{2t}{\pi}$, the pseudo-Choi state produced will be slightly different than the one in Section~\ref{section:ChoiStateLearning} (Equation \eqref{eq:choistate}). That is, rather than a pseudo-Choi state of $\widetilde{H}$, the state produced is a pseudo-Choi state of $\frac{2\widetilde{H}t}{\pi}$:
\begin{align}
    \ket{\psi_c'} &=  \frac{\ket{0}_C(\frac{\widetilde{H}}{\Delta} \otimes \Id_A) \ket{\Phi_d}_{SA} + \ket{1}_C\ket{\Phi_d}_{SA}}{\gamma}\label{eq:ChoiState2} ,
\end{align}
where
\begin{align}
    \gamma &= \sqrt{\frac{\norm{\widetilde{\textbf{c}}}_2^2}{\Delta^2} + 1}\label{eq:gamma},
\end{align}
and
\begin{align}
    \Delta = \frac{\pi}{2t} \label{eq: Delta}.
\end{align}

The vector $\widetilde{\textbf{c}}$ above is the vector of Hamiltonian coefficients that correspond to the block-encoded Hamiltonian $\widetilde{H}$ (see Definition~\ref{def:Hamdefs} in the following section).

\begin{figure}[tb]
\centering
    \begin{quantikz}
        \lstick{$\ket{0}_B$}                  & \qw          & \qw & \gate[wires=2]{U_{block}} & \qw & \meter{} \rstick{$\rm{Pr}(0) \geq  \frac{1}{2}$}\\
        \lstick[wires=2]{$\ket{\phi_d}_{SA}$} & \qwbundle{n} & \qw &                           & \qw & \qw \rstick[wires=3]{$\ket{\psi_c'} = \frac{\ket{0}_C                                                                                                                            (\frac{\widetilde{H}}{\Delta} \otimes \Id_A)\ket{\Phi_d}_{SA} +                                                                                                                         \ket{1}_C\ket{\Phi_d}_{SA}}{\gamma}$}\\
                                              & \qwbundle{n} & \qw & \qw                       & \qw & \qw \\
        \lstick{$\ket{0}_C$}                  &\gate{H}      & \qw & \octrl{-2}                & \qw & \qw
    \end{quantikz}
    \caption{A quantum circuit that produces the pseudo-Choi state in Lemma~\ref{lemma: ProduceChoiState}. If the qubit in register $B$ is found to be in the state $\ket{0}$ after measuring it in the computational basis, the output of the remaining registers is the pseudo-Choi state $\ket{\psi_c'}$. This outcome occurs with probability equal to $\frac{\norm{\textbf{c}}^2_2}{2\Delta^2} + \frac{1}{2}$, where $\norm{\textbf{c}}_2$ is the 2-norm of the vector of Hamiltonian coefficients, and $\Delta = \frac{\pi}{2t}$.\label{figure: ProduceChoiState}}
\end{figure}

\begin{lemma}[Generating pseudo-Choi States from the Hamiltonian Block-Encoding\label{lemma: ProduceChoiState}]
    Let $t$ be a positive number in $(0, 1/2\|H\|]$.  There exists an algorithm that yields a pseudo-Choi state of the form $\ket{\psi_c'} = \frac{\ket{0}_C (\frac{\widetilde{H}}{\Delta} \otimes \Id_A)\ket{\Phi_d}_{SA} + \ket{1}_C \ket{\Phi_d}_{SA}}{\gamma}$, where $\Delta = \frac{\pi}{2t}$ and $\gamma = \sqrt{\frac{\norm{\widetilde{\textbf{c}}}_2^2}{\Delta^2} + 1}$, using a single controlled application of $U_{block}$ whose success is heralded by a measurement that occurs with probability strictly greater than $\frac{1}{2}$. 
\end{lemma}
Figure~\ref{figure: ProduceChoiState} gives an explicit circuit for generating the pseudo-Choi state from the Hamiltonian block-encoding $U_{block}$, and the details of the proof of Lemma~\ref{lemma: ProduceChoiState} can be found in Appendix~\ref{App:ProduceChoiState}.  Note that the restriction on the evolution time considered above is inherited from Lemma~\ref{lemma:ProduceBlockEncoding}.

\subsection{Learning via Shadow Tomography}


Since the block-encoding $U_{block}$ is used to produce our resource states, the learning procedure will generate an approximation of the Hamiltonian within the block-encoding. This means that there will be additional error introduced due to the fact that the Hamiltonian in the block-encoding will not be exactly equal to the true Hamiltonian.

We now rephrase the Hamiltonian learning problem in terms of the true Hamiltonian, its block-encoded approximation, and the Hamiltonian recovered by the learning procedure, as well as their corresponding coefficient vectors, as follows:

\begin{definition}{\label{def:Hamdefs}
Let us take the following three definitions for the true Hamiltonian $H$, the block-encoded Hamiltonian $\tilde{H}$, and the approximation $\hat{H}$ that results from the learning process:
\begin{enumerate}
\item Let $H := \sum_{m=0}^{M-1} c_m H_m$ be the Hamiltonian we seek to learn, and let $\textbf{c}$ be the vector of its Hamiltonian coefficients (that is, the coefficients of its representation in some orthonormal basis $\{H_m \mid m \in \mathbb{Z}_{M}\}$).

\item Let $\widetilde{H} := \sum_{m=0}^{M-1} \widetilde{c}_m H_m$ be the approximation of $H$ that results from the block-encoding procedure, and let $\widetilde{\textbf{c}}$ be the vector of its Hamiltonian coefficients.

\item Let $\hat{H} := \sum_{m=0}^{M-1} \hat{c}_m H_m$ be the approximation of $\widetilde{H}$ that results from the learning procedure, and let $\hat{\textbf{c}}$ be the vector of its Hamiltonian coefficients.
\end{enumerate}
}
\end{definition}

These above definitions naturally lead us to the following definition for the Hamiltonian learning problem.

\begin{definition}[Unitary Hamiltonian Learning Problem]
\label{def:HamLearningProblemShort}
The Hamiltonian learning problem is to compute, for any $\epsilon>0$ and $\delta>0$, a vector $\hat{\textbf{c}}$ such that with probability at least $1-\delta$,
\begin{align}
    \norm{\hat{\textbf{c}} -  \textbf{c}}_2 \leq \epsilon,
\end{align}
by querying $U = e^{-iHt}$ and $U^\dagger$.
\end{definition}

We also define the task of learning the approximate Hamiltonian $\widetilde{H}$ which appears in the block-encoding as follows:
\begin{definition} [Block-Encoded Hamiltonian Learning Problem\label{def:BELearningProblem}]{

The block-encoded Hamiltonian learning problem is to compute a vector $\hat{\textbf{c}}$ such that with probability at least $1-\delta$,
\begin{align}
    \norm{\hat{\textbf{c}} -  \widetilde{\textbf{c}}}_2 \leq \epsilon_c,
\end{align}
by querying $U_{block}$

}
\end{definition}

Recall that Lemma~\ref{lemma:ProduceBlockEncoding} also gives us $\epsilon_b$, the upper bound on the error of the block-encoded Hamiltonian.


From the pseudo-Choi state \eqref{eq:ChoiState2}, we can learn the Hamiltonian coefficients in similar fashion to the PC HLP considered above.  Specifically, recall that
        \begin{align}
            O_l &= \ket{0}_C\bra{1}_C \otimes \left(H_l \otimes \Id\right) \ket{\Phi_d}_{SA} \bra{\Phi_d}_{SA} 
        \end{align}
        where $H_l$ is a Hamiltonian term, and let 
        \begin{align}
            O_{\gamma} &= \ket{1}_C\bra{1}_C \otimes \ket{\Phi_d}_{SA} \bra{\Phi_d}_{SA}
        \end{align}
 Define the Pseudo-Choi state for the density operator $\rho_c'$ according to the following
        \begin{align}
            \rho_c' := \ket{\psi_c'}\bra{\psi_c'},
        \end{align}
        so that 
        \begin{align}
            {\rm Tr}(\rho_c' O_{l}) = \frac{c_l}{\Delta \gamma^2},\label{eq:TrrhoOl}
        \end{align}
        where $c_l$ is the $l^{th}$ Hamiltonian coefficient and
        \begin{align}
            {\rm Tr}(\rho_c' O_{\gamma}) = \frac{1}{\gamma^2}\label{eq:gamma2}
        \end{align}
        where $\gamma$ is defined in~\eqref{eq:gamma}
The entire learning process, outlined in Algorithm~\ref{algorithm: FindCoeffUnitary}, therefore consists of first generating a sufficient amount of pseudo-Choi states \eqref{eq:ChoiState2} from the Hamiltonian, running Algorithm~\ref{algorithm: FindCoeffClifford}, and then scaling the result by $\Delta$.

\begin{algorithm}[t]
\caption{$\mathtt{FindCoeffUnitary}(U, U^\dagger, \Delta, \boldsymbol{H}, n, N_s)$: Learn the Hamiltonian Coefficients by Querying $U = e^{-iHt}$ and $U^\dagger$}\label{algorithm: FindCoeffUnitary}
\textbf{Input:} 

$U$: Time evolution operator $e^{-iHt}$ acting on Hilbert space $\mathcal{H}_S$

$\Delta:$ Normalization constant equal to $\frac{\pi}{2t}$ where $t \in \left(0, \frac{1}{2 \norm{H}}\right]$ is the evolution time used in $U$, which dictates the probability of success for projection onto the pseudo-Choi state measurement.

$\boldsymbol{H}$: Set of M k-local Hamiltonian terms $H_l$ whose coefficients $c_m$ are desired

\Comment{Note: $\boldsymbol{H}$ is a set of classical operators}

$n$: Number of qubits in the system on Hilbert space $\mathcal{H}_S$ 

$N_s$: Number pseudo-Choi states desired

\begin{enumerate}
    \item \textbf{Generate pseudo-Choi states}
    
    Run the circuit from Figure~\ref{figure: ProduceChoiState} on $N \in \mathcal{\widetilde{\mathcal{O}}}\left(N_s\right)$ independent input states to produce $N_s$ pseudo-Choi states of the form $\rho_c' = \ket{\psi_c'}\bra{\psi_c'}$, where $\ket{\psi_c'} = \frac{\ket{0}_C(\frac{\widetilde{H}}{\Delta} \otimes \Id_A)\ket{\Phi_d}_{SA} + \ket{1}_C\ket{\Phi_d}_{SA}}{\gamma}$ with high probability.
    
    \Comment{An explicit upper bound on $N$ is given in Theorem~\ref{thm:querycomplexityUB}}

    \Comment{Recall that $U_{block}$ can be produced by querying $U$ and $U^\dagger$ as per Lemma~\ref{lemma:ProduceBlockEncoding}, using the methods in \cite{Gily_n_2019}}

    \item \textbf{Run Algorithm~\ref{algorithm: FindCoeffClifford} to estimate $\frac{\hat{\mathbf{c}}}{\Delta}$}
    
    $\mathtt{Set}$ $\hat{\mathbf{c}} \gets \mathtt{FindCoeffClifford}(\rho_c'^{\otimes N_s}, \boldsymbol{H}, n, N_s)$ 
    
    \Comment{Algorithm~\ref{algorithm: FindCoeffPauli} can be used as well}
    
    \item \textbf{Rescale the result by $\Delta$ as per equation \eqref{eq:TrrhoOl}}
    
    $\mathtt{Return}$ $\Delta\hat{\mathbf{c}}$
\end{enumerate}
\end{algorithm}


\subsection{Query Complexity for Hamiltonian Learning from Time Evolution Unitaries}\label{section: blackboxQueryComplexity}
Rather than sample complexity, it is more meaningful to look at the query complexity for this model - that is, the number of (controlled) queries to $U = e^{-iHt}$. The main result of this section is Theorem~\ref{thm:querycomplexityUB}, which gives an upper bound on the number of times the time-evolution unitary must be queried in order to solve the Unitary HLP~\ref{def:HamLearningProblemShort}. We begin by finding a similar upper bound for an easier task, in particular solving the Block-Encoded HLP~\ref{def:BELearningProblem}.

\begin{theorem} [Query Complexity Upper Bound for Solving the Block-Encoded HLP \label{thm:samplecomplexityUB2}]{

The number of queries to the block-encoding \eqref{eq:blockencoding} required to solve the Block-Encoded Hamiltonian Learning Problem~\ref{def:BELearningProblem} for a Hamiltonian whose $M$ Hamiltonian terms are known, is at most
\begin{align}
    \widetilde{N} & \in \mathcal{O}\left( \frac{M \gamma^2 ({\widetilde{c}_{max}}^2 + \frac{1}{t^2}) \log{(M/\delta)}}{\epsilon_c^2} \right)
\end{align}

}
\end{theorem}

\begin{proof}[Proof of Theorem~\ref{thm:samplecomplexityUB2}]
        
    The sample complexity of predicting the expectation values of the $2M$ operators in Algorithm~\ref{algorithm: FindCoeffUnitary} using classical shadows is
    \begin{align}
        N_s & \in \mathcal{O}\left( \frac{\log(M/\delta_s)}{\epsilon_s^2} \right) \label{eq:basicShadowSampleComplexity}
    \end{align}
    This is the number of pseudo-Choi states \eqref{eq:ChoiState2} required to predict $\frac{1}{\gamma^2}$ and $\hat{o}_m = \frac{\widetilde{c}_m}{\gamma^2 \Delta}$  to within error $\epsilon_s$ for all $m \in \mathbb{Z}_M$ with probability at least $1-\delta_s$. Multiplying each $\hat{o}_m$ by $\gamma^2 \Delta$ gives the coefficients, $\widetilde{c}_m$, of the block-encoded Hamiltonian within error
    \begin{align}
        \epsilon_m \equiv \gamma^2 \epsilon_s \sqrt{{\widetilde{c}_m}^2 + \Delta^2}.
    \end{align}
    Recall that learning these coefficients is the goal of the Block-Encoded HLP in Definition~\ref{def:BELearningProblem}.
    The $l_2$ error between the vector of Hamiltonian coefficients returned by Algorithm~\ref{algorithm: FindCoeffUnitary} and that of the block-encoded Hamiltonian is
    \begin{align}
        \norm{\widetilde{\textbf{c}} - \hat{\textbf{c}}}_2 &\leq \sqrt{M} \max_{m} \epsilon_m \nonumber \\
        & = \sqrt{M}\gamma^2 \epsilon_s \sqrt{{\widetilde{c}_{max}}^2 + \Delta^2},
    \end{align}
    where $\widetilde{c}_{max}$ is the coefficient with the largest absolute value. To satisfy the Block-Encoded HLP~\ref{def:BELearningProblem}, we let the upper bound on this error be equal to $\epsilon_c$. It follows that the following value of $\epsilon_s$ is sufficient to solve the Block-Encoded HLP:
    \begin{align}
        \epsilon_s &= \frac{\epsilon_c}{\sqrt{M}\gamma^2 \sqrt{{\widetilde{c}_{max}}^2 + \Delta^2}}.\label{eq:epsc}
    \end{align}

    Substituting this value for $\epsilon_s$ into Equation~\eqref{eq:basicShadowSampleComplexity} means that the total number of pseudo-Choi states we need to solve the Hamiltonian Learning Problem is
    \begin{align}
        N_s & \in \mathcal{O}\left( \frac{M \gamma^4 ({\widetilde{c}_{max}}^2 + \Delta^2) \log{(M/\delta_s)}}{\epsilon_c^2} \right)
    \end{align}

    Now that we have an upper bound on the number of pseudo-Choi states \eqref{eq:ChoiState2} required to solve the Block-Encoded HLP~\ref{def:BELearningProblem} with probability at least $1 - \delta_s$, we would like to know how many queries to $U_{block}$ are needed to solve the Block-Encoded HLP with probability at least $1-\delta$. We must therefore consider the failure probability of the learning algorithm.

    To overestimate the failure probability of the entire learning procedure, we assume that it will always fail if less than $N_s$ pseudo-Choi states are produced, and also that the success probability of classical shadows is the same as long as the number of pseudo-Choi states generated is at least $N_s$. The probability of our learning algorithm failing is then
    \begin{align}
        \textbf{P}_{fail} & \leq \delta_s + \delta_{N_s} \nonumber \\
        & \equiv \delta,
    \end{align}
    where $\delta_s$ is the probability of classical shadows failing if we have $N_s$ pseudo-Choi states, $\delta_{N_s}$ is the probability of producing less than $N_s$ pseudo-Choi states after $\widetilde{N}$ queries to $U_{block}$, and the third line comes from the requirement of the Block-Encoded HLP~\ref{def:BELearningProblem} which says that the failure probability must be at most $\delta$. We then let $\delta_s \equiv \delta_{N_s} \equiv \frac{\delta}{2}$. This choice of $\delta_s$ gives
    \begin{align}
        N_s & \in \mathcal{O}\left( \frac{M \gamma^4 ({\widetilde{c}_{max}}^2 + \Delta^2) \log{(M/\delta)}}{\epsilon_c^2}\right), 
    \end{align}
    
    Next, it can be shown via Chernoff bound that
    \begin{align}
        \widetilde{N} & \in \mathcal{O}\left(\frac{N_s}{\gamma^2}\right) \label{eq: chernoff}
    \end{align}
    queries to the the block-encoding $U_{block}$ are needed to ensure $N_s$ pseudo-Choi states are produced with probability at least $1-\delta$. Note that the asymptotic notation may be misleading here since $\gamma^2$ is larger than 1, but $\widetilde{N}$ is in fact larger than $N_s$ (see Appendix~\ref{App: Hoeffding}). Combining this with the upper bound on $N_s$ above, and using $\Delta = \frac{\pi}{2t}$, we get
    \begin{align}
        \widetilde{N} & \in \mathcal{O}\left( \frac{M \gamma^2 ({\widetilde{c}_{max}}^2 + \frac{1}{t^2}) \log{(M/\delta)}}{\epsilon_c^2}\right),
    \end{align}
    proving the result in Theorem~\ref{thm:samplecomplexityUB2}.
    
\end{proof}

\subsubsection{Choosing Sufficient Values for the Block-Encoding Error and Block-Encoding Learning Error}

Theorem~\ref{thm:samplecomplexityUB2} gives an upper bound on the number of Hamiltonian block-encodings needed to solve the Block-Encoded HLP~\ref{def:BELearningProblem}. Note that solving the Block-Encoded HLP is just a sub-task of solving the Unitary HLP~\ref{def:HamLearningProblemShort}, with the missing step being the generation of the Hamiltonian block-encodings from the time evolution unitary $e^{-iHt}$. This means that to determine the total query complexity, we have to factor in the number of queries required to produce the desired number of block-encodings of the Hamiltonian. In addition, since the block-encoding generated by the procedure described in Lemma~\ref{lemma:ProduceBlockEncoding} has an error of $\epsilon_b$, the learning algorithm will not approximate the coefficients of the Hamiltonian terms of $H$, but those of the approximation of $H$ in the block-encoding. Thus, the block-encoding error $\epsilon_b$ will have a further impact on the total query complexity.

We now seek values for the block-encoding error $\epsilon_b$ and the block-encoding learning error $\epsilon_c$, such that the total error in learning the Hamiltonian is at most $\epsilon$, as required by the Unitary HLP~\ref{def:HamLearningProblemShort}. We begin with the following proposition:
\begin{proposition}\label{prop: coeffErrUB}
    The error in the vector, $\mathbf{\widetilde{c}}$, of Hamiltonian coefficients of the block-encoded Hamiltonian $\widetilde{H}$ is upper bounded by 
    \begin{align*}
        \norm{\widetilde{\textbf{c}} - \textbf{c}}_2 & \leq \frac{\sqrt{M} \epsilon_b}{t},
    \end{align*}
    where $M$ is the number of Hamiltonian terms, $t$ is the evolution time used to generate the block-encoding, and $\epsilon$ is the block-encoding error, as in Lemma~\ref{lemma:ProduceBlockEncoding}.
\end{proposition}
\begin{proof}[Proof of Proposition~\ref{prop: coeffErrUB}]
\begin{align}
    \abs{\widetilde{c}_m - c_m} &= \abs{\frac{1}{d}\text{Tr}(\widetilde{H} H_m) - \frac{1}{d}\text{Tr}(H H_m)} \nonumber \\
    & = \frac{1}{d}\abs{\text{Tr}((\widetilde{H}-H) H_m)} \nonumber \\
    & \leq \frac{1}{d}\sum_{i=1}^k\sigma_i(\widetilde{H}-H)\sigma_i(H_m) \nonumber \\
    & \leq \frac{1}{d}\sigma_{max}(\widetilde{H}-H)\sum_{i=1}^k\sigma_i(H_m) \nonumber \\
    & = \frac{1}{d}\sigma_{max}(\widetilde{H}-H) \sum_{i=1}^k\abs{\lambda_i(H_m)} \nonumber  \\
    & = \frac{1}{d}\abs{\sigma_{max}(\widetilde{H}-H)}\text{Tr}\left(\sqrt{H_m^\dagger H_m}\right) \nonumber \\
    & = \abs{\sigma_{max}(\widetilde{H}-H)} \nonumber \\
    & = \norm{\widetilde{H}-H}_2
\end{align}
In the first inequality we have used the Von Neumann trace inequality, where $\sigma_i(A)$ is the $i^{th}$ of the $k$ singular values of A when sorted in descending order. Likewise, $\sigma_{max}(\widetilde{H}-H)$ is the singular value of $\widetilde{H}-H$ with the largest magnitude. 

Noting that this result does not depend on $m$, we use it to finish the proof as follows:
\begin{align}
    \norm{\widetilde{\textbf{c}} - \textbf{c}}_2 & \leq \sqrt{M} \norm{\widetilde{\textbf{c}} - \textbf{c}}_\infty \nonumber  \\
    &= \sqrt{M} \max_m \abs{\widetilde{c}_m - c_m} \nonumber \\
    & \leq \sqrt{M} \norm{\widetilde{H} - H}_2 \nonumber \\
    &\leq \frac{\sqrt{M} \epsilon_b}{t}
\end{align}
where we use the result of Lemma~\ref{lemma:ProduceBlockEncoding} in the last line.
\end{proof}

The next step in choosing sufficient values for $\epsilon_b$ and $\epsilon_c$ is to upper bound the total learning error from solving the Unitary HLP in terms of $\epsilon_b$ and $\epsilon_c$, as in Lemma~\ref{lem:LearningErrorUB}.

\begin{lemma} [Upper Bound on Error for the Unitary HLP\label{lem:LearningErrorUB}]{
The total error between the vector of the true Hamiltonian coefficients and the vector recovered by using Algorithm~\ref{algorithm: FindCoeffUnitary} to solve the Unitary HLP is upper bounded by
\begin{align*}
    \norm{\hat{\textbf{c}} - \textbf{c}}_2 & \leq \epsilon_c + \frac{\sqrt{M} \epsilon_b}{t}\label{totalerrbound2}
\end{align*}
where $\epsilon_c$ is the error in learning the 
block-encoded Hamiltonian (see Definition~\ref{def:BELearningProblem}) and $\epsilon_b$ is the error in the block-encoding itself (see Lemma~\ref{lemma:ProduceBlockEncoding}).

}
\end{lemma}

\begin{proof}[Proof of Lemma~\ref{lem:LearningErrorUB}]
\begin{align}
    \norm{\hat{\textbf{c}} - \textbf{c}}_2 & = \norm{\hat{\textbf{c}} - \widetilde{\textbf{c}} + \widetilde{\textbf{c}} - \textbf{c}}_2 \nonumber \\
    & \leq \norm{\hat{\textbf{c}} - \widetilde{\textbf{c}}}_2 +  \norm{\widetilde{\textbf{c}} - \textbf{c}}_2 \nonumber \\
    & \leq \epsilon_c +  \norm{\widetilde{\textbf{c}} - \textbf{c}}_2
\end{align}
Note that the second inequality is guaranteed by querying the Hamiltonian block-encoding $U_{block}$ the sufficient amount of times to solve the Block-Encoded HLP \eqref{def:BELearningProblem}. Theorem~\ref{thm:samplecomplexityUB2} gives an upper bound on this number of queries.

Combining this with the result of Proposition~\ref{prop: coeffErrUB} gives the result in Lemma~\ref{lem:LearningErrorUB}.
\end{proof}

Finally, to achieve the desired upper bound of $\norm{\hat{\textbf{c}} - \textbf{c}}_2 \leq \epsilon$ for the total error of the Hamiltonian learning problem \eqref{def:HamLearningProblemShort}, we simply choose
$\epsilon_b$ and $\epsilon_c$ such that the right side of the inequality in Lemma~\ref{lem:LearningErrorUB} is less than $\epsilon$. Since the learning process is much more expensive than the block-encoding process, it makes sense to allocate most of the error budget to $\epsilon_c$, as will be discussed in the following section.

\subsubsection{Query Complexity Upper Bound}\label{section: QueryComplexity}
With the previous results in place, we can now state our bounds on the query complexity for the Unitary HLP~\ref{def:HamLearningProblemShort}. Unlike in the Block-Encoded HLP, and the PC HLP, we now consider unitary preparation of the pseudo-Choi states via the process described in Section~\ref{section: GenerateChoiState}.  Note that in the following we do not make use of amplitude amplification as the success probability for this learning problem is at least $1/2$, which implies that amplitude amplification will not lead to asymptotic advantages.



\begin{theorem} [Query Complexity Upper Bound for Solving the Unitary HLP \label{thm:querycomplexityUB}]{

The number of queries to the time evolution unitary $e^{-iHt}$ and its inverse required to solve the Unitary Hamiltonian Learning Problem~\ref{def:HamLearningProblemShort} for a Hamiltonian with $M$ Hamiltonian terms, is at most
\begin{align}
    N & \in \mathcal{O}\left( \frac{M \log{(M/\delta)}}{t^2 \epsilon^2} \log{\left( \frac{M}{t \epsilon} \right)}\right)
\end{align}
where $t \in \mathcal{O}(1/\norm{H})$
}
\end{theorem}
\begin{proof}[Proof of Theorem~\ref{thm:querycomplexityUB}]
From Theorem~\ref{thm:samplecomplexityUB2} we have that the number of block-encodings needed to solve the Block-Encoded HLP~\ref{def:BELearningProblem} is
\begin{align}
    \widetilde{N} & \in \mathcal{O}\left( \frac{M \gamma^2 ({\widetilde{c}_{max}}^2 + \frac{1}{t^2}) \log{(M/\delta)}}{\epsilon_c^2} \right)
\end{align}
Lemma~\ref{lemma:ProduceBlockEncoding} states that $\mathcal{O}\left( \log\left(\frac{1}{\epsilon_b}\right)  \right)$ queries to the time evolution blackbox $U = e^{-iHt}$ and its inverse are required to generate each block-encoding. Since we apply Lemma~\ref{lemma:ProduceBlockEncoding} to produce the block-encoding, it is important to note the preconditions on $\epsilon$ and $\epsilon_b \le 1/2$.  The latter of which is guaranteed by our assumption on $\epsilon$ if $\epsilon_b \le \epsilon$. Next, recall that $\frac{\gamma^2}{2}$ is the probability of projecting onto a pseudo-Choi state upon measuring the block-encoding qubit.  Lemma~\ref{lemma: ProduceChoiState} then implies that $\gamma^2 \in \mathcal{O}(1) $ under the assumption that $t\in \mathcal{O}(1/\norm{H})$. Furthermore, we make the assumption that each Hamiltonian coefficient satisfies $\abs{c_m} \leq 1$.  As the norm of the Hamiltonian coefficients can be rescaled by choosing the evolution time to be larger or smaller, we can make such a choice for our state learning without restricting the class of Hamiltonians that is learnable within our approach.

Therefore, the upper bound on the total query complexity for the Hamiltonian learning problem \eqref{def:HamLearningProblemShort} is
\begin{align}
    N &\in \mathcal{O}\left(\frac{M\log\left(M/\delta\right) }{t^2 \epsilon_c^2}   \log\left(\frac{1}{\epsilon_b}\right)  \right)
\end{align}

To satisfy Problem~\ref{def:HamLearningProblemShort}, the total error must be $\norm{\hat{\textbf{c}} - \textbf{c}}_2  \leq \epsilon$. From Lemma~\ref{lem:LearningErrorUB} we see that this can be accomplished by choosing $\epsilon_c$ and $\epsilon_b$ such that $\epsilon_c + \frac{M\epsilon_b}{t} \leq \epsilon$. If we assign an error budget of $\frac{\epsilon}{2}$ to each of these terms, that is choose
\begin{align}
    \epsilon_c & := \epsilon/2, \label{eq:epsiloncValue}
\end{align}
and
\begin{align}
    \epsilon_b := \frac{\epsilon t}{2 M},
\end{align}
we arrive at the query complexity given in Theorem~\ref{thm:querycomplexityUB}.
\end{proof}

Note that the choice of error budget is somewhat arbitrary for the purposes of determining the asymptotic scaling, as it only affects the query complexity by a constant factor. However, since the query complexity is quadratic in $\epsilon_c^{-1}$ but only logarithmic in $\epsilon_b^{-1}$, it is logical to assign most of the error budget to $\epsilon_c$. For example, choosing
\begin{align}
    \epsilon_c \leq \left( 1 - \frac{1}{M} \right)\epsilon \label{eq:ecUB}
\end{align}    
and
\begin{align}
    \epsilon_b \leq \frac{t \epsilon }{M^2} \label{eq:ebUB}
\end{align}
keeps the total error to at most $\epsilon$ and limits the increase in query complexity to a constant factor of approximately 2 for large values of M.  Our assumption that $\epsilon_b \le 1/2$ then immediately follows from $\epsilon \le M^2/(2t)$.



\section{Hamiltonian Learning via Quantum Mean-Value Estimation}\label{section: QME}

The approach considered previously has great strengths.  Notably, the quantum operations required are simple Pauli or Clifford operations (depending on the flavour of classical shadows used), and after measurements are complete, all expectation values can be calculated efficiently on a classical computer. It should be noted that it also allows us to learn the Hamiltonian without having an explicit dependence on the number of terms in the Hamiltonian if we look at the $l_\infty$ error in the Hamiltonian coefficients (though we have focused on the $l_2$ error metric throughout this paper, as it better highlights the differences between the query complexities of our approach and previous ones).  However, the most significant drawback of this approach is that the method scales poorly with $1/\epsilon$.  This can be ameliorated by dropping the use of shadows and using a quantum computer to directly estimate all of the terms directly from the PC state by computing the expectation values of the Hamiltonian using mean-value estimation algorithms.

In \cite{Huggins_nearly_optimal_estimating} an algorithm is proposed for estimating the expectation value of a vector of observables.  The central idea behind this approach is to take the problem of mean estimation and reduce it to a gradient estimation procedure on the exponential of the sum of Hamiltonian terms. The gradient estimation procedure can then be executed optimally using the technique of~\cite{gilyen2019optimizing}, which then allows the vector of mean values to be estimated using a number of queries that also scales optimally with the parameters involved.  The key result is that for some state $\psi$, the mean values of a set of $M$ operators with norm at most $\nu_{max}$ can be computed within $l_\infty$ error $\epsilon'$ and probability of failure at most $\delta$ using $\widetilde{\mathcal{O}}( \nu_{max}\sqrt{M} / \epsilon')$ queries to the unitary that prepares $\psi$ as well as to oracles that perform the exponentials of each of the operators. In Lemma~\ref{lem: gradEstQuery} we show that the number of queries of the same form required to learn a Hamiltonian using this technique is $\widetilde{\mathcal{O}}( \alpha^2 M / \epsilon)$, where $\epsilon$ is the $l_2$ error in the Hamiltonian coefficients and $\alpha^2$ is the normalization factor in the pseudo-Choi state.


The main result of this section is given later in Theorem~\ref{thm: gradEstFullQuery}, which gives an upper bound on the number of queries to the time-evolution unitary that are needed to solve the Hamiltonian learning problem. However, we begin in Section~\ref{section:OracularQME} by proving Lemma~\ref{lem: gradEstQuery}, which solves a simpler problem. In particular, it assumes we have a unitary oracle that prepares our resource state, and gives an upper bound on the number of queries to the oracle needed to solve the Hamiltonian learning problem. Next, in Section~\ref{section:UnitaryConstruction} we show how to generate the resource state in unitary fashion by querying the time-evolution operator and its inverse and then prove Theorem~\ref{thm: gradEstFullQuery}.

\subsection{Learning the Hamiltonian Coefficients from a Unitary Preparation of the Resource State} \label{section:OracularQME}

As in the case of Pauli-based classical shadows, we shall start with the pseudo-Choi state and perform a partial trace over the ancillary system $A$, resulting in the following resource state:
\begin{align}
    \rho = & \frac{1}{d\alpha^2}\Big(H^2 \otimes \ket{0}_C\bra{0}_C  + H \otimes \ket{0}_C\bra{1}_C   +  H \otimes \ket{1}_C\bra{0}_C +  \Id \otimes \ket{1}_C\bra{1}_C \Big)\label{eq:PartialChoiDensityOp2},
\end{align}
where 
\begin{align}
    \alpha &= \sqrt{\norm{\textbf{c}}^2_2 + 1}
\end{align}
and $\textbf{c}$ is the vector of Hamiltonian coefficients.

Next, we define the decoding operators we will use to extract the Hamiltonian coefficients from our resource state.
\begin{definition}[Decoding Operators] \label{Def: QMEDecodingOperators}
    Let the set of decoding operators be defined as
    \begin{align}
        \boldsymbol{O} & \equiv \{ O_l | l \in \mathbb{Z}_{M} \} 
    \end{align}
    where
    \begin{align}
        O_l &= \frac{H_l  \otimes X_C}{2}
    \end{align}
    and $H_l$ is one of the $M$ Hamiltonian terms. 
    
    Furthermore, we define
    \begin{align}
        O_{\alpha} & \equiv \Id \otimes \ket{1}_C\bra{1}_C
    \end{align}
\end{definition}
 
Proposition~\ref{proposition: ComputeHamCoeffsPauli} shows that for our resource state $\rho$, the expectation values of the decoding operators can be used to recover the Hamiltonian coefficients, as well as the normalization factor of the resource state.

\begin{proposition}[Computing the Hamiltonian Coefficients\label{proposition: ComputeHamCoeffsPauli}]
If $\rho$ is the resource state \eqref{eq:PartialChoiDensityOp2}, then $\forall$ $l \in \mathbb{Z}_M$ we have
\begin{align}
    {\rm Tr}(\rho O_{l}) = \frac{c_l}{\alpha^2},
\end{align}
where $c_l$ is the $l^{th}$ Hamiltonian coefficient, and $\alpha$ is the normalization constant of the pseudo-Choi state \eqref{eq:alpha}. Furthermore,
\begin{align}
    {\rm Tr}(\rho O_{\alpha}) = \frac{1}{\alpha^2}
\end{align}

\end{proposition}
\begin{proof}[Proof of Proposition~\ref{proposition: ComputeHamCoeffsPauli}]

From the definitions of $O_l$ and the resource state $\rho$, it is easy to see that
\begin{align}
    {\rm Tr}(\rho O_l) & = \frac{{\rm Tr}(H H_l)}{d \alpha^2}.
\end{align}
Expanding the above using the representation of the Hamiltonian from Definition~\ref{def:HamLearningProblem} gives
\begin{align}
    {\rm Tr}(\rho O_l) & = \frac{\sum_{m} c_m{\rm Tr}( H_m  H_l)  }{d \alpha^2} \nonumber \\
    & = \frac{1}{d \alpha^2}  \sum_{m} c_m  \delta_{ml}d \nonumber \\
    & = \frac{c_l}{\alpha^2}
\end{align}

The second claim (${\rm Tr}(\rho O_{\alpha}) = \frac{1}{\alpha^2}$) follows immediately from Equation \eqref{eq:PartialChoiDensityOp} and the definition of $O_{\alpha}$.
\end{proof}

Next, we show that the Hamiltonian learning problem can be efficiently solved by querying a unitary oracle that prepares the resource state.

\begin{lemma}\label{lem: gradEstQuery}
    Let $U_\rho$ be a unitary oracle that prepares the state $\rho$ from equation~\eqref{eq:PartialChoiDensityOp2}. The Hamiltonian learning problem in Definition~\ref{def:HamLearningProblem} can be solved using
    \begin{align*}
        \widetilde{\mathcal{O}}\left( \frac{\alpha^2 M}{\epsilon} \right)
    \end{align*}
    queries to $U_\rho$ and $U_\rho^\dagger$, as well as to an oracle that implements controlled-$e^{-i\theta O_l}$ for each decoding operator $O_l$.
\end{lemma}

\begin{proof}
    By Proposition~\ref{proposition: ComputeHamCoeffsPauli}, we have
    \begin{align}
        {\rm Tr}(\rho O_{l}) = \frac{c_l}{\alpha^2},
    \end{align}
    From this we can see that the Hamiltonian can be estimated by estimating a mean-value of a set of Hamiltonian operators.  Specifically, consider the vector $\mathbf{d}$
    \begin{equation}
        \mathbf{d} := \left[ {\rm Tr}(\rho O_0),\ldots,  {\rm Tr}(\rho O_{M-1})\right]
    \end{equation}
    from which the vector of Hamiltonian coefficients can be computed via
    \begin{equation}
        c_l = \alpha^2 d_l \label{eq: alphad}
    \end{equation}
    Since the Hamiltonian learning problem considers Hamiltonian terms to have a norm that is at most a constant (so $\nu_{max} \in \mathcal{O}(1)$), the mean value estimation algorithm of~\cite{Huggins_nearly_optimal_estimating} can be used to estimate $\mathbf{d}$ within infinity norm $\epsilon'$ using 
    \begin{align}
        N &\in \widetilde{\mathcal{O}}(\sqrt{M}/\epsilon')
    \end{align}
    queries to $U_\rho$ and $U_\rho^\dagger$, as well as to controlled-$e^{-i\theta O_l}$ $\forall l \in [0, M-1]$. All that is left now is to determine how small $\epsilon'$ should be to guarantee that the $l_2$ error in the estimate of the vector of Hamiltonian coefficients is at most $\epsilon$.
    
    If we define $\hat{c}_l$ to be our estimate of $c_l$, then
    \begin{align}
        |\hat{c}_l - c_l| & \leq \alpha^2 \epsilon' \sqrt{c_l^2 + 1} \nonumber \\
        & \in O\left(\alpha^2 \epsilon' \right),
    \end{align}
    since we need to scale the result of the mean estimation by $\alpha^2$ (as in equation~\eqref{eq: alphad}). 
    Therefore, if we want the $l_{\infty}$ error in estimating $c_l$ to be at most $\epsilon_{\infty}$, we require
    \begin{align}
        \epsilon' &\in \mathcal{O}(\epsilon_\infty /\alpha^2)
    \end{align}
    
    Finally, for the 2-norm of the error in the estimate of the vector of Hamiltonian coefficients to be at most $\epsilon$, we choose $\epsilon_\infty \in \Theta(\epsilon/\sqrt{M})$ (because $\|\cdot\|_2 \le \sqrt{M}\|\cdot\|_\infty$ from Cauchy Schwarz). This leads to a query complexity of 
    \begin{align}
        N & \in \widetilde{\mathcal{O}}(\alpha^2 M/\epsilon).
    \end{align}
\end{proof}

We now wish to use the result of Lemma~\ref{lem: gradEstQuery} to determine how many queries to the time evolution blackbox will be necessary to solve the Unitary HLP.

\subsection{Constructing a Unitary Preparation of the Resource State}\label{section:UnitaryConstruction}
In Lemma~\ref{lem: gradEstQuery} we showed that the Hamiltonian Learning Problem could be solved using $\mathcal{\widetilde{O}}\left( \frac{\alpha^2 M}{\epsilon} \right)$ queries to a unitary operation that prepares the state $\rho$ from equation~\eqref{eq:PartialChoiDensityOp2}. We would now like to show that this state can be prepared in unitary fashion by querying the time-evolution operator $U=e^{-iHt}$ and its inverse. The idea is that if we know the number of queries to $U$ and $U^\dagger$ it takes to prepare this state, we can use Lemma~\ref{lem: gradEstQuery} to find the number of queries to $U$ and $U^\dagger$ needed to solve the Unitary HLP. We begin by considering the circuit from Section~\ref{section: GenerateChoiState} (reproduced here in Figure~\ref{figure: ProduceChoiState2}) that was used to generate the pseudo-Choi state. Recall that $U_{block}$ contains a block-encoding, $\widetilde{H}$, of the Hamiltonian $H$, and Lemma~\ref{lemma:ProduceBlockEncoding} shows that $U_{block}$ can be efficiently generated by querying the time-evolution unitary $U=e^{-iHt}$ and its inverse. A full analysis of the circuit can be found in Appendix~\ref{App:ProduceChoiState}.

\begin{figure}[tb]
\centering
    \begin{quantikz}
        \lstick{$\ket{0}_B$}                  & \qw          & \qw & \gate[wires=2]{U_{block}} & \qw & \meter{} \rstick{$\rm{Pr}(0) \geq  \frac{1}{2}$}\\
        \lstick[wires=2]{$\ket{\phi_d}_{SA}$} & \qwbundle{n} & \qw &                           & \qw & \qw \rstick[wires=3]{$\ket{\psi_c'} = \frac{\ket{0}_C                                                                                                                            (\frac{\widetilde{H}}{\Delta} \otimes \Id_A)\ket{\Phi_d}_{SA} +                                                                                                                         \ket{1}_C\ket{\Phi_d}_{SA}}{\gamma}$}\\
                                              & \qwbundle{n} & \qw & \qw                       & \qw & \qw \\
        \lstick{$\ket{0}_C$}                  &\gate{H}      & \qw & \octrl{-2}                & \qw & \qw
    \end{quantikz}
    \caption{A quantum circuit that produces the pseudo-Choi state of $\widetilde{H}/\Delta$. If the qubit in register $B$ is found to be in the state $\ket{0}$ after measuring it in the computational basis, the output of the remaining registers is the pseudo-Choi state. This outcome occurs with probability equal to $\frac{\norm{\textbf{c}}^2_2}{2\Delta^2} + \frac{1}{2}$, where $\norm{\textbf{c}}_2$ is the 2-norm of the vector of Hamiltonian coefficients, and $\Delta = \frac{\pi}{2t}$.\label{figure: ProduceChoiState2}}
\end{figure}

The first step in producing the state in Equation~\ref{eq:PartialChoiDensityOp2} is to produce the pseudo-Choi state, in similar fashion to Figure~\ref{figure: ProduceChoiState2}. However, since we would like to use the result of Lemma~\ref{lem: gradEstQuery}, we must make note of two things. First, the pseudo-Choi state we get from querying the time-evolution operator is not of the Hamiltonian $H$, but of $\widetilde{H}/\Delta$ (where $\widetilde{H}$ is the block-encoding of the Hamiltonian and $\Delta = \frac{\pi}{2t}$), meaning that we will end up with a state that is slightly different from the one in equation~\ref{eq:PartialChoiDensityOp2}. We will address this later when looking at the query complexity. The second thing to take note of is that Lemma~\ref{lem: gradEstQuery} requires a unitary preparation of the resource state, but the circuit in Figure~\ref{figure: ProduceChoiState2} involves measuring the block-encoding qubit. We can get around this measurement by using fixed-point amplitude amplification to approximate the pseudo-Choi state as follows.

The fixed-point amplitude amplification (FPAA) process of Yoder et al.~\cite{yoder2014fixed} considers a starting state $\ket{s}$, a circuit $A$ that prepares the starting state, a target state $\ket{T}$, and an oracle $U_{AA}$ that flips an ancilla qubit when given the target state. It then applies a circuit $S_L$ (which makes $\mathcal{O}(L)$ queries $U$, $A$, and $A^\dagger$) to the initial state such that $\norm{\bra{T}S_L \ket{s}}^2 \geq 1 - \delta_{AA}$ for a given $\delta_{AA} \in [0,1]$. In our case, the preparation circuit, $A$, is simply the circuit from Figure~\ref{figure: ProduceChoiState2}, excluding the measurement of register $B$. Recall that this circuit queries the time evolution unitary, which is our learning resource, and prepares the state
\begin{align}
    \ket{s} & \equiv \frac{1}{\sqrt{2}}\Bigg(  \ket{0}_C   \ket{0}_B  \left(\frac{\widetilde{H}}{\Delta} \otimes \Id_A\right) \ket{\Phi_d}_{SA}  + \ket{0}_C   \ket{1}_B  ( J_S \otimes \Id_A) \ket{\Phi_d}_{SA}  + \ket{1}_C \ket{0}_B \ket{\Phi_d}_{SA} \Bigg), \label{eq: initialState}
\end{align}
which serves as the initial state for the FPAA circuit. Our target state is the pseudo-Choi state
\begin{align}
    \ket{T} & \equiv  \frac{\ket{0}_C \ket{0}_B (\frac{\widetilde{H}}{\Delta} \otimes \Id_A) \ket{\Phi_d}_{SA} + \ket{1}_C \ket{0}_B \ket{\Phi_d}_{SA}}{ \gamma }, \label{eq: targetState}
\end{align}
where $\gamma = \sqrt{\frac{\norm{\widetilde{\textbf{c}}}_2^2}{\Delta^2} + 1}$.

To perform FPAA, we need a unitary that flips an ancilla qubit when it acts on the target state. While we do not know the pseudo-Choi state (seeing as it contains all the unknown Hamiltonian coefficients), we know that if the block-encoding qubit is in the state $\ket{0}_B$, the remaining registers contain the pseudo-Choi state. Therefore, the unitary we desire is simply a zero-controlled $CNOT$ operation using register $B$ as the control. With this, we can apply FPAA to $\ket{s}$ to get the state
\begin{align}
    \ket{T_{approx}} & \equiv S_L \ket{s},
\end{align}
which is very close to the pseudo-Choi state, $\ket{T}$, of the block-encoded Hamiltonian $\widetilde{H}/\Delta$. Following this, we simply ignore all qubits in register $A$, resulting in
\begin{align}
    \chi & \equiv {\rm Tr}_A(\ketbra{T_{approx}}{T_{approx}}),
\end{align}
which is an approximation of the state
\begin{align}
    \widetilde{\rho} & \equiv {\rm Tr}_A(\ketbra{T}{T}) \nonumber \\
    & =\frac{1}{d\gamma^2}\Big(\frac{\widetilde{H}^\dagger \widetilde{H}}{\Delta^2} \otimes \ket{0}_C\bra{0}_C  + \frac{\widetilde{H}}{\Delta} \otimes \ket{0}_C\bra{1}_C   +  \frac{\widetilde{H}^\dagger}{\Delta} \otimes \ket{1}_C\bra{0}_C +  \Id \otimes \ket{1}_C\bra{1}_C \Big) \otimes \ketbra{0}{0}_B \label{eq:PartialApprox}.
\end{align}
In turn, $\widetilde{\rho}$ is essentially an approximation of the state $\rho$ from equation~\eqref{eq:PartialChoiDensityOp2}, with the block-encoding procedure being the source of error. We must also highlight the distinction that the Hamiltonian is normalized by $\Delta$ in this case, whereas the Hamiltonian in $\rho$ was not.

\begin{figure}[tb]\label{figure: ProducePartialChoiState}
\centering
    \begin{quantikz}
        \lstick{$\ket{0}_B$}                  & \qw           & \qw & \gate[wires=2]{U_{block}} & \qw & \gate[wires=5]{F.P.A.A.} & \qw \rstick[wires=4]{$\ket{T_{approx}} \equiv S_L \ket{s}$} \\
        \lstick[wires=2]{$\ket{\phi_d}_{SA}$} & \qwbundle{n}  & \qw &                           &  \qw   &                       & \qw   \\
                                              & \qwbundle{n}  & \qw & \qw                       &  \qw   &                       & \qw \\
        \lstick{$\ket{0}_C$}                  &\gate{H}       & \qw & \octrl{-2}                &  \qw   &                       & \qw \\
        \lstick{$\ket{ancilla}$}              & \qw  & \qw &\qw                        &  \qw   &                       & \qw  
    \end{quantikz}
    \caption{A quantum circuit that produces an approximation of the pseudo-Choi state of $\widetilde{H}/\Delta$, where $F.P.A.A.$ refers to the circuit performing fixed-point amplitude amplification due to Yoder et al.~\protect\cite{yoder2014fixed} . Performing a partial trace over subsystem $A$ results in an approximation of the state in equation~\ref{eq:PartialChoiDensityOp2}. Since this process is unitary, Lemma~\ref{lem: gradEstQuery} gives an upper bound on the number of queries to this circuit needed to solve the Unitary HLP.}
\end{figure}
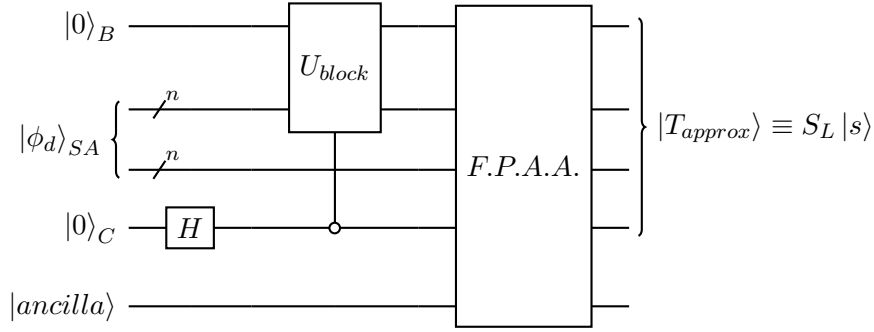

Now that we have a unitary preparation of this state, we need to determine how many iterations of FPAA are needed in order to ensure that after using the results of Lemma~\ref{lem: gradEstQuery}, the total error in predicting the Hamiltonian coefficients is small enough to satisfy the Unitary Hamiltonian Learning Problem from Definition~\ref{def:HamLearningProblemShort}. This will tell us the total number of queries to the time evolution operator $U$ and it's inverse, which is the total query complexity of the algorithm. An upper bound on this value is given in Theorem~\ref{thm: gradEstFullQuery}, which we prove at the end of this section. We begin by establishing a few lemmas that capture the query complexity of specific parts of the total algorithm.

Yoder et al.~\cite{yoder2014fixed} describe the query complexity of their fixed-point amplitude amplification algorithm in terms of the success probability of obtaining the target state after the FPAA process, which we summarize in Lemma~\ref{lem: FPAA}.
\begin{lemma}[Fixed Point Amplitude Amplification]\label{lem: FPAA}
    Let $\ket{s}$ be some initial state, $\ket{T}$ be the desired target state, with their overlap given by $\braket{T}{s} = \sqrt{\lambda}e^{i\xi}$. Let $A$ be a circuit that prepares $\ket{s}$ from the zero state, and $U_{AA}$ be a unitary oracle that flips an ancilla qubit when the input is $\ket{T}$. Additionally, let $S_L$ be the circuit that performs the fixed-point amplitude amplification process using $\mathcal{O}(L)$ queries to $A$, $A^\dagger$, and $U_{AA}$. The success probability of obtaining the target state is given by $P_L = \abs{\bra{T} S_L \ket{s}}^2$. For a given $\delta_{AA}$, the value of $L$ needed to extract the target state with probability $P_L \geq 1-\delta_{AA}$ is 
    \begin{align*}
        L &\in \mathcal{O}\left( \frac{\log (1/\delta_{AA})}{\sqrt{\lambda}}\right)
    \end{align*}
\end{lemma}

\begin{proof}
    See Yoder et al.~\cite{yoder2014fixed}
\end{proof}
Since our preparation of the pseudo-Choi state should be unitary, and thus does not involve measurement, we restate this query complexity in Lemma~\ref{lem: FPAA2} in terms of the trace distance between the target state $\ket{T}$ and the output state $\ket{T_{approx}} = S_L \ket{s}$, rather than in terms of a success probability. Note that since our ultimate goal is to determine the number of queries to the time evolution unitary, we only care about the queries to the preparation circuit $A$, which queries the time evolution unitary, and not the number of queries to $U_{AA}$, which as we discussed previously is just a simple 2-qubit operation.
\begin{lemma}[Query Complexity of Fixed Point Amplitude Amplification]\label{lem: FPAA2}
    Let $\ket{s}$ be the initial state as defined in equation~\eqref{eq: initialState}, $\ket{T}$ be the target state as defined in equation~\eqref{eq: targetState}, and $A$ be a circuit that prepares $\ket{s}$ from the zero state. Additionally, let $S_L$ be the circuit that performs the fixed-point amplitude amplification process using $\mathcal{O}(L)$ queries to $A$ and $A^\dagger$. In order for the trace distance between the target state $\ket{T}$ and the output state $\ket{T_{approx}} = S_L \ket{s}$ to be at most $D$, the number of queries to $A$ and $A^\dagger$ needed is 
    \begin{align*}
        L &\in \mathcal{O}\left( \log (1/D)\right)
    \end{align*}
\end{lemma}

\begin{proof}
    Lemma~\ref{lem: FPAA} guarantees that $\mathcal{O}\left( \frac{\log (1/\delta_{AA})}{\sqrt{\lambda}}\right)$ queries to $A$ and $A^\dagger$ are enough to ensure that $\abs{\bra{T} S_L \ket{s}}^2 \geq 1-\delta_{AA}^2$, which implies that
    \begin{align}
        \frac{1}{\delta_{AA}} & \leq \frac{1}{\sqrt{1 - \abs{\bra{T} S_L \ket{s}}^2}}.
    \end{align}

    Therefore, we can express the number of queries as
    \begin{align}
        L & \in \mathcal{O}\left(\frac{1}{\sqrt{\lambda}} \log \left(\frac{1}{\sqrt{1 - \abs{\bra{T} S_L \ket{s}}^2}} \right)\right)\label{eq: randomlabel8}
    \end{align}

    Using the definitions of $\ket{s}$ and $\ket{T}$ from equations~\eqref{eq: initialState} and~\eqref{eq: targetState}, it can be shown that $\braket{T}{s} \in [1/2, 1]$ (the calculation is identical to that of calculating the success probability of measuring the block-encoding qubit in the desired state, which can be found in Appendix~\ref{App:ProduceChoiState}). Therefore, since $\braket{T}{s} = \lambda e^{i \xi}$, we need not consider $\lambda$ in the query complexity, as in the worst case it only affects the number of queries by a constant factor.

    Since $S_L \ket{s}$ and $\ket{T}$ are pure states, the trace distance between them is given by $D = \frac{1}{2}\sqrt{1 - \abs{\bra{T} S_L \ket{s}}^2}$. With this, we can express equation~\eqref{eq: randomlabel8} as
    \begin{align}
        L & \in \mathcal{O}\left( \log \left(\frac{1}{D} \right)\right)
    \end{align}  
\end{proof}

Next, we prove Lemma~\ref{lem: QMEQueryComplexity}, which is simply a reformulation of Lemma~\ref{lem: gradEstQuery} using the state $\widetilde{\rho}$ from equation~\eqref{eq:PartialApprox} rather than the state $\rho$ from equation~\eqref{eq:PartialChoiDensityOp2}.

\begin{lemma}[Query Complexity of Quantum Mean Estimation]\label{lem: QMEQueryComplexity}
     Let $U_{\widetilde{\rho}}$ be a unitary oracle that prepares the state $\widetilde{\rho}$ from equation~\eqref{eq:PartialApprox}. The Hamiltonian learning problem in Definition~\ref{def:HamLearningProblem} can be solved using
    \begin{align*}
        \widetilde{\mathcal{O}}\left( \frac{\Delta M}{\epsilon} \right)
    \end{align*}  
    queries to $U_{\widetilde{\rho}}$ and $U_{\widetilde{\rho}}^\dagger$, as well as to an oracle that implements controlled-$e^{-i\theta O_l}$ for each decoding operator of the form $O_l = \frac{H_l  \otimes X_C}{2}$.
\end{lemma}

\begin{proof}
    Using the definition the block-encoded Hamiltonian $\widetilde{H}$ from Definition~\ref{def:Hamdefs}, a short calculation shows that
    \begin{align}
        {\rm Tr}(\widetilde{\rho} O_l) &= \frac{{\rm Re}[\widetilde{c}_l]}{\Delta \gamma^2}
    \end{align}
    Note that while the coefficients $c_m$ of the Hamiltonian $H$ are real, the coefficients $\widetilde{c}_m$ of the approximation in the block-encoding $\widetilde{H}$ may be complex, though we only learn their real component here. Fortunately, this is of no consequence; in Section~\ref{section: blackboxQueryComplexity} we showed that $\abs{\widetilde{c}_l - c_l} \leq \frac{\epsilon_b}{t}$, where $t$ is the evolution time and $\epsilon_b$ is the block-encoding error, and writing $\widetilde{c}_l$ as the sum of its real and imaginary parts it is easy to see that $\abs{{\rm Re}[\widetilde{c}_l] - c_l} \leq \abs{\widetilde{c}_l - c_l}$, which means that the upper bound still holds. We can now continue in similar fashion to Lemma~\ref{lem: gradEstQuery}, using the quantum mean estimation protocol introduced by Huggins et al.~\cite{Huggins_nearly_optimal_estimating} to estimate the vector
    \begin{align}
        \mathbf{d} & := [{\rm Tr}(\widetilde{\rho} O_0), ..., {\rm Tr}(\widetilde{\rho} O_{M-1})],
    \end{align}
    from which we approximate the Hamiltonian coefficients by
    \begin{align}
        \widetilde{c}_l &:= \Delta \gamma^2 d_l,
    \end{align}
    where, again, we ignore any possible imaginary component. By the same argument as in Lemma~\ref{lem: gradEstQuery}, we require
    \begin{align}
        \widetilde{\mathcal{O}}(\Delta \gamma^2 M/\epsilon)
    \end{align}
    queries to $U_{\widetilde{\rho}}$ and $U_{\widetilde{\rho}}^\dagger$ to predict the vector of coefficients $\mathbf{\widetilde{c}} = [\widetilde{c}_0, ..., \widetilde{c}_{M-1}]$ to within $l_2$ error $\epsilon$.

    To complete the proof, we refer to Appendix~\ref{App:ProduceChoiState}, where we showed that $\gamma \in [1, \sqrt{2}] \in \mathcal{O}(1)$, and therefore we state the query complexity as
    \begin{align}
        \widetilde{\mathcal{O}}(\Delta M/\epsilon)
    \end{align}
    
\end{proof}

We now have the sufficient resources to prove Theorem~\ref{thm: gradEstFullQuery}.


\begin{theorem}[Solving the Hamiltonian Learning Problem via Quantum Mean Estimation]\label{thm: gradEstFullQuery}
    Let $U=e^{-iHt}$ be the unitary operator that performs time evolution under the Hamiltonian $H$ for time $t$. The number of queries to $U$ and $U^\dagger$ needed to solve the Unitary Hamiltonian Learning Problem for the Hamiltonian $H$ is at most

    \begin{align*}
        N &\in \widetilde{\mathcal{O}}\left( \frac{M}{\epsilon t} \right)
    \end{align*}
\end{theorem}
\begin{proof}
    The total number of queries to the time evolution operator and its inverse can be calculated by looking at the procedure as 4 nested queries.
    
    By Lemma~\ref{lem: QMEQueryComplexity}, the quantum mean estimation process makes
    \begin{align}
        \widetilde{\mathcal{O}}(\Delta M/\epsilon')
    \end{align}
    queries to the unitary $U_{\widetilde{\rho}}$ that prepares the state $\widetilde{\rho}$ (equation~\eqref{eq:PartialApprox}), where $\Delta = \frac{\pi}{2t}$. $U_{\widetilde{\rho}}$ makes use of the fixed-point amplitude amplification circuit, which makes
    \begin{align}
        \mathcal{O}\left( \log (1/D)\right)
    \end{align}
    queries the circuit $A$ that prepares state $\ket{s}$ (equation~\eqref{eq: initialState}), as well as to $A^\dagger$, where $D$ is the trace distance between the target state and the approximation of the target state that results from the FPAA process. Finally, the circuit $A$ makes a single query to the block-encoding of the Hamiltonian, which is prepared using
    \begin{align}
        \mathcal{O}\left(\log(1/\epsilon_b)\right)
    \end{align}
    queries to the time evolution operator $U=e^{-iHt}$ and its inverse, where $\epsilon_b$ is the block-encoding error. Altogether, this results in a query complexity of
    \begin{align}
        \mathcal{O}\left(\frac{M \log (1/D) \log(1/\epsilon_b)}{\epsilon' t}\right).
    \end{align}

    Since each step of the process involves approximation, we must choose sufficiently small values for the learning error $\epsilon'$, the block-encoding error $\epsilon_b$, and the FPAA error $\epsilon_{AA} := D$ such that the total $l_2$ error in the estimate of the Hamiltonian coefficients is at most $\epsilon$, as required by the Hamiltonian learning problem.

    Recall from Definition~\ref{def:Hamdefs} that $\mathbf{c}$ is the vector of Hamiltonian coefficients we wish to learn, $\mathbf{\widetilde{c}}$ is the vector of coefficients for the block-encoded Hamiltonian, and $\mathbf{\hat{c}}$ is the vector of coefficients that is estimated by the learning procedure. The total error in learning the coefficients is then
    \begin{align}
        \norm{\mathbf{\hat{c}} - \mathbf{c}}_2 & = \norm{\mathbf{\hat{c}} - \mathbf{\widetilde{c}} + \mathbf{\widetilde{c}} - \mathbf{c}}_2 \nonumber \\
        & \leq \norm{\mathbf{\hat{c}} - \mathbf{\widetilde{c}} }_2 + \norm{\mathbf{\widetilde{c}} - \mathbf{c}}_2 \nonumber \\
        & \leq \norm{\mathbf{\hat{c}} - \mathbf{\widetilde{c}} }_2 + \frac{M\epsilon_b}{t},
    \end{align}
    where we have used the result of Proposition~\ref{prop: coeffErrUB} in the last line. Conceptually, what we have done is separate the total error into the sum of the error in learning the block-encoded coefficients and the error in the block-encoded coefficients themselves. Next, we separate $\norm{\mathbf{\hat{c}} - \mathbf{\widetilde{c}} }_2$ into the error that comes from the quantum mean estimation, and the error resulting from using the state $\chi$ to approximate the QME resource state $\widetilde{\rho}$:
    \begin{align}
        \norm{\mathbf{\hat{c}} - \mathbf{\widetilde{c}} }_2 & \leq \norm{\mathbf{\hat{c}} - \mathbf{c'} }_2 + \norm{\mathbf{c'} - \mathbf{\widetilde{c}} }_2,
    \end{align}
    where $\mathbf{c'}$ is the vector of Hamiltonian coefficients that would be recovered if the the state $\widetilde{\rho}$ was used in the QME process. We upper bound $\norm{\mathbf{\hat{c}} - \mathbf{c'} }_2$ by considering the error in estimating a single coefficient:
    \begin{align}
        \abs{\hat{c}_l - c'_l} &= \Delta \gamma^2 \abs{{\rm Tr}(\chi O_l) - {\rm Tr}(\widetilde{\rho} O_l)} \nonumber \\
        & = \Delta \gamma^2 \abs{{\rm Tr}((\chi -\widetilde{\rho}) O_l)} \nonumber \\
        & \leq \Delta \gamma^2 \sum_{i=1}^k \sigma_i(\chi -\widetilde{\rho}) \sigma_i(O_l) \nonumber \\
        & \leq \Delta \gamma^2 \sigma_{max}(O_l) \sum_{i=1}^k \sigma_i(\chi -\widetilde{\rho})  \nonumber \\
        & = \Delta \gamma^2 \norm{O_l}_2 \sum_{i=1}^k \abs{\lambda_i(\chi -\widetilde{\rho})}  \nonumber \\
        & = \frac{\Delta \gamma^2}{2}  {\rm Tr}\left( \sqrt{(\chi -\widetilde{\rho})^\dagger (\chi -\widetilde{\rho})} \right)  \nonumber \\
        & = \Delta \gamma^2 T(\chi, \widetilde{\rho})
    \end{align}
    where $T(\chi, \widetilde{\rho})$ is the trace distance between $\chi$ and $\widetilde{\rho}$. We have made use of Von Neumann's trace inequality in the third line, where $\sigma_i(\cdot)$ represents the $i^{th}$ singular value of the corresponding operator. In the fifth line we used the fact that $\norm{O_l}_2 = 1/2$ ($O_l$ is simply a Pauli operator multiplied by a factor of 1/2), as well as expressed the singular values of $\chi -\widetilde{\rho}$ as the absolute values of its eigenvalues (this is valid since $\chi -\widetilde{\rho}$ is Hermitian, and therefore normal).

    By the Cauchy-Schwarz inequality, we then have
    \begin{align}
        \norm{\mathbf{\hat{c}} - \mathbf{c'} }_2 & \leq \sqrt{M} \max_l \abs{\hat{c}_l - c'_l} \nonumber \\
        & \leq \sqrt{M}\Delta \gamma^2 T(\chi, \widetilde{\rho})
    \end{align}

    Recall that $\chi$ is the state we get after tracing out system $A$ for the state $\ket{T_{approx}}$ that results from the FPAA process, and similarly $\widetilde{\rho}$ is the state resulting from tracing out system $A$ for the target state $\ket{T}$ of the FPAA. Therefore, the trace distance between $\chi$ and $\widetilde{\rho}$ is less than or equal to the trace distance between $\ket{T_{approx}}$ and $\ket{T}$~\cite{Rastegin_2012}.

    As per Lemma~\ref{lem: FPAA2}, performing $\mathcal{O}(\log(1/D))$ iterations of FPAA guarantees that the trace distance between $\ket{T_{approx}}$ and $\ket{T}$ is at most $D$. Therefore, we arrive at the following expression for the total error in predicting the Hamiltonian coefficients:
    \begin{align}
        \norm{\mathbf{\hat{c}} - \mathbf{c}}_2 & \leq \sqrt{M} \Delta \gamma^2 D + \epsilon' + \frac{M \epsilon_b}{t}\nonumber \\
        & = \frac{\pi \gamma^2 \sqrt{M} D}{2t} + \epsilon' + \frac{M \epsilon_b}{t}.
    \end{align}
    Assigning an error budget of $\epsilon/3$ to each of the three terms (i.e. choosing $D := \frac{2 \epsilon t}{3 \pi \gamma^2 \sqrt{M}}$, $\epsilon' := \epsilon/3 $, and $\epsilon_b := \frac{\epsilon t}{3 M}$) ensures the total error is at most $\epsilon$, and gives a final query complexity of
    \begin{align}
        N & \in \mathcal{O}\left( \frac{M  \log^2\left(\frac{M}{\epsilon t}\right)}{\epsilon t} \right).
    \end{align}
    Using $\widetilde{\mathcal{O}}$ which suppresses subdominant logarithmic terms, we get
    \begin{align}
        N & \in \widetilde{\mathcal{O}}\left( \frac{M}{\epsilon t} \right),
    \end{align}
    proving the result in Theorem~\ref{thm: gradEstFullQuery}.  

    Note that the failure probability does not lead to an increase in the number of queries (in $\widetilde{O}$ notation) because unlike in the classical shadows approach, we replace the measurement in the preparation circuit with fixed-point amplitude amplification so that instead of getting our desired resource state with some probability, we are guaranteed to have an approximation of the resource state. This means the only failure probability comes from the quantum mean estimation process, which guarantees the failure probability is smaller than $1/3$~\cite{Huggins_nearly_optimal_estimating}. By Chernoff arguments this probability can be made less than $\epsilon$  using a logarithmic number of extra queries, thus the $\widetilde{O}$ query complexity is unaffected.

\end{proof}

\section{Applications}\label{section:Applications}

Since in general a Hamiltonian can have exponentially many terms, one of the focuses of research on Hamiltonian learning has been to find suitable classes of Hamiltonians for which learning is more tractable. One way to approach this problem is to limit the scope to Hamiltonians with some known structure. Along this vein, several recent approaches to Hamiltonian learning have considered low-intersection Hamiltonians whose Hamiltonian terms have limited overlap in their support \cite{haah2021optimal, PhysRevLett.130.200403, gu2022practical}. This class of Hamiltonians contains all spatially local Hamiltonians, and as such contains many physically relevant Hamiltonians. 

While the low-intersection class of Hamiltonians is undoubtedly rich, it does not capture all physically relevant Hamiltonians. Our approach is well suited to more general Hamiltonians, such as those that are k-local but not low-intersection (for example, if you have pairwise interactions between all pairs of qubits). This is the first regime where our methods outperform previous approaches, under the assumption that we can query $U = e^{-iHt}$ as well as $U^\dagger$. However, the learning algorithms we develop are suitable even for Hamiltonians beyond the k-local class. Furthermore, in addition to physically notable Hamiltonians like the examples given in this section, several quantum algorithms use Hamiltonian simulation as a subroutine. Notable examples include algorithms for solving the quantum linear systems problem~\cite{HHL, FasterLinSys}, phase estimation~\cite{KitaevPhaseEstimation}, semi-definite programming~\cite{brando_SDP}, linear PDEs~\cite{Childs_PDEs}, and recent work on simulating systems of coupled oscillators~\cite{Nathan_Oscillators}. Therefore, it is of interest to develop learning techniques for different and more general classes of Hamiltonians, and not just the low-intersection class.

In this section we present a few examples of Hamiltonians that are $k$-local, but not low-intersection. We also include a discussion about the settings in which previous Hamiltonian techniques work well, and in which settings our approach outperforms them - in several cases reducing the query complexity polynomially, or even exponentially.

\subsection{Fermionic and Hard-core Boson Models}
A particularly relevant application of Hamiltonian learning is the learning problem for systems of particles.  These models are naturally represented in terms of creation and annihilation operators.  The nature of these operators vary, however the common thread between them is that the number operator (which measures the number of particles located in a particular mode) can be expressed as 
\begin{align}
n=a^\dagger a
\end{align}
where $a$ is an annihilation operator and $a^\dagger$ is a creation operator.  The form of these operators depends on the symmetries of the underlying particles; however, popular examples include bosonic case where $(a^\dagger)^{p}\ket{0} = \sqrt{p!} \ket{p}$ where $\ket{p}$ is a state that has $p$ particles in it and the fermionic or hard-core bosonic cases where $(a^\dagger)^p \ket{0}= \delta_{p\le 1} \ket{p}$.  The bosonic cases have commuting operators; whereas the fermionic case uses anti-commuting creation and annihilation operators.

As a particular example, let us consider the following model which describes the hopping of particles on a graph $G=(V,E)$ with interactions between nearest neighbors in the graph.  For simplicity, let us assume that the creation and annihilation operators are either fermionic or hard-core bosonic.  Under these assumptions, the Hamiltonian can be expressed as
\begin{equation}
H = \sum_{j\in V} c_j a^\dagger_j a_j + \sum_{(i,j) \in E} d_{i,j} (a^\dagger_i a_j + a^\dagger_j a_i) + e_{i,j} a^\dagger_i a_i a^\dagger_j a_j
\end{equation}
These Hamiltonian terms can be simplified by realizing that $\hat{n} = a^\dagger a = (I +Z)/2$
\begin{equation}
H = \sum_{j\in V} c_j' Z_j + \sum_{(i,j) \in E} d_{i,j} (a^\dagger_i a_j + a^\dagger_j a_i) + e_{i,j} Z_i Z_j/4
\end{equation}
In this case the Hamiltonian satisfies the above assumptions and further a specific qubit representation of the creation and annihilation operators are not needed for this case provided ${\rm Tr}( a^\dagger_i a_j a^\dagger_k a_\ell) =0$ for distinct $i,j$ and $k,\ell$ unless $i=\ell$ and $j=k$.  To see this, note that $a_i$ is an off-diagonal matrix as $a_i \ket{1} = \ket{0}$.  Thus the only way we can build a diagonal operator is by pairing $a_i$ with $a_i^\dagger$.  This can only occur, under the assumption of distinctness if $i=\ell$ and similarly $j=k$.  Thus this conforms the the Hamiltonian Learning Problem from Definition~\ref{def:HamLearningProblem} while not needing a specific Pauli decomposition of the creation and annihilation operators.  In contrast, many existing approaches to Hamiltonian learning require a low-intersection Hamiltonian and such a Hamiltonian need not be low intersection.  Specifically, if one considers the case of a fermionic Hamiltonian then the Jordan-Wigner representation does not yield a low-intersection Hamiltonian except for the case of one-dimensional graph $G$.  Thus our method differs from many approaches because such Hamiltonians can be learned directly rather than going first through a representation of for their creation and annihilation operators.

\subsection{Spin Glass Models}
Other models that lend themselves well to our approach to Hamiltonian learning are spin glasses. One particular example on the complete graph is the Sherrington-Kirkpatrick (SK) model \cite{SKModel}, an Ising-like model in which any pair of spins on a lattice can be coupled, regardless of the distance between them. The SK model was introduced to study the magnetic properties of spin glasses, and is also of interest in combinatorial optimization and the study of neural networks \cite{sherrington2012physics}.

In the presence of a transverse field, the Sherrington-Kirkpatrick Hamiltonian takes the form 
\begin{align}
    H_{SK} &= \frac{1}{\sqrt{n}} \sum_{1 \leq i < j \leq n} J_{ij} Z_i Z_j + g \sum_{i=1}^n X_i,
\end{align}
where the $J_{ij}$ are i.i.d. random variables sampled from the standard Gaussian distribution, and $g>0$ is the strength of the transverse field.

Since couplings are between pairs of spins, the Hamiltonian is 2-local and thus contains $M \in \mathcal{O}(n^2)$ terms. However, the Hamiltonian is very much not low-intersection, as each term has overlapping support with $\mathcal{O}(M)$ other terms. Note that while Heisenberg models on the complete graph contain only pairwise interactions, our method applies to more general spin glass models which allow for interactions between more than two spins. These models, often referred to as p-spin models, are p-local Hamiltonians with infinite range p-body interactions and thus the query complexity of the learning process increases, as there are $M \in \mathcal{O}(n^p)$ Hamiltonian coefficients to learn. Additionally, while these spin glass models only contain Pauli $Z$ interactions, our methods also apply to models with more general interactions such as Heisenberg models on dense graphs.

\subsection{Prior Work}\label{section: comparison}

In this section we consider several recent works on Hamiltonian learning. Note that the varying cost and error metrics used by different authors makes comparing approaches a difficult endeavour, so our goal here is mainly to highlight the scenarios where each approach may be particularly useful and mention how our approach fares in each case. Regarding error metrics, we would like to note that the infinity-norm, where one would like to simultaneously estimate every individual Hamiltonian coefficient to within some error ($\max_i \abs{\hat{c}_i - c_i} \leq \epsilon$), is popular as a metric. However, we feel that the 2-norm ($\norm{\hat{\textbf{c}} - \textbf{c}  }_2  \leq \epsilon$) better highlights the cost of learning, as the infinity-error obscures some of the dependence on the number of Hamiltonian terms. Thus, unless otherwise stated, references to query complexity will be for learning the Hamiltonian coefficients to error $\epsilon$ in the 2-norm. While all the approaches we will discuss use the same basic learning resource (black-box access to $U=e^{iHt}$), we separate them into two classes: those that require quantum control, and those that do not. The distinguishing factor between these two cases is that experiments without quantum control use only a single application of $U$.

First, we consider approaches that have the best scaling with respect to the error $\epsilon$ in learning the Hamiltonian coefficients. Namely, the results by \textit{Huang et al.}~\cite{PhysRevLett.130.200403} and \textit{Dutkiewicz et al.}~\cite{dutkiewicz2023advantage}, which are able to achieve Heisenberg-limited scaling ($\mathcal{O}(\epsilon^{-1})$) in the evolution time. As shown in~\cite{dutkiewicz2023advantage}, this is possible only through the use of quantum control, which makes their approach well suited for learning low-intersection Hamiltonians when accuracy is important. Our QME-based approach also makes use of quantum control to achieve Heisenberg-limited scaling, and has the added advantage that we are able to learn more general Hamiltonians, such as k-local Hamiltonians. Our shadows-based approach does not achieve this error scaling, but is robust to unexpected Hamiltonian terms, as we will discuss in the following section. A direct comparison between our approach and the one of~\cite{PhysRevLett.130.200403} is difficult because of differences in error metrics, complexity metrics, and input models. Importantly, the input models differ in that we require backwards time evolution, but only require one controlled application of $U$ and $U^\dagger$ for a short time. Meanwhile, in~\cite{PhysRevLett.130.200403} they require the ability to rapidly interleave many short-time evolutions. One case where our approach excels is when the Hamiltonian contains high-weight Pauli terms (i.e. terms that act non-trivially on $k$ qubits, for large $k$). Here, the techniques of~\cite{PhysRevLett.130.200403} require exponential (in $k$) queries and total evolution time. This is particularly bad if $k \in \mathcal{O}(n)$. On the other hand, our approach is able to handle this case with $\mathbf{poly}(n)$ queries to $U$ and $U^\dagger$.

There are also several approaches that do not beat the standard quantum limit(~\cite{haah2021optimal, gu2022practical, frança2022efficient, yu2022practical, caro2023learning}) but have other notable advantages - perhaps the most significant being that aside from~\cite{yu2022practical} they do not require quantum control. Several of these approaches also aim to look beyond low-intersection Hamiltonians. The results of \textit{Franca et al.}~\cite{frança2022efficient} scale particularly well with then number of Hamiltonian terms, requiring $\widetilde{\mathcal{O}}\left(\frac{M}{\epsilon^2} \right)$ queries to $U$ to learn a Hamiltonian with $M$ terms. However, their constant factor is exponentially large in $k$, and so their approach is unfavourable if the Hamiltonian contains high-weight Pauli terms. Conversely, the approach of \textit{Gu et al.}~\cite{gu2022practical} does not scale as favourably for Hamiltonians with many terms (requiring $\mathcal{O}(M^5)$ queries), but is able to handle high-weight Pauli terms without incurring an exponential (in $k$) overhead. These approaches are very well-suited for leaning low-intersection Hamiltonians, and even some (but not all) more general types of Hamiltonians, but are not ideal when there are high-weight terms or the number of terms is large.  
    
The techniques of \textit{Yu et al.}~\cite{yu2022practical} and \textit{Caro}~\cite{caro2023learning} seek to remedy this and consider even broader classes of Hamiltonians. The approach of~\cite{caro2023learning}, allows for any Hamiltonian to be learned (in the 2-norm) using $\mathbf{poly}(n)$ queries to $U$, so long as the Hamiltonian contains at most $\mathbf{poly}(n)$ non-zero terms, has some known structure (for example, k-locality), and $\norm{H} \in \mathbf{poly}(n)$. In~\cite{yu2022practical}, the authors dropped the need for an underlying structure to be known a priori, and instead only require that the Hamiltonians are sparse. The main drawback is that the resource complexity of these learning procedures does not scale well with the error in learning the coefficients, as both require $\mathcal{O}(1/\epsilon^4)$ queries to $U$. In addition, the effects of the evolution time on the query complexity of~\cite{yu2022practical} are not clear.

The result due to \textit{Caro}~\cite{caro2023learning} is particularly relevant to our work. Firstly, the procedure makes use of Choi states of the time evolution unitary, whereas we use pseudo-Choi states of the Hamiltonian itself. Our improved error scaling appears to indicate that pseudo-Choi states are more powerful in some ways than regular Choi states, and we believe this comes from the use of quantum control in preparing pseudo-Choi states. As mentioned above, ~\cite{dutkiewicz2023advantage} demonstrated the power of quantum control in achieving better error scaling, so it is perhaps not unexpected that our error scaling is not matched by an approach that uses only time-evolved states. Interestingly, both approaches also share similar bounds on evolution times. In addition to this, our results and those of~\cite{caro2023learning} have a key advantage that to the best of our knowledge no other approach shares, which is that our shadows-based approach, along with that of~\cite{caro2023learning}, can be used to learn \textit{any} Hamiltonian query efficiently if we consider the error in the infinity norm. As mentioned earlier, this error metric perhaps blurs the dependence on the number of Hamiltonian terms, but in situations where one really only cares about this metric, these two approaches can be used to learn Hamiltonians with even an exponential number of terms using a polynomial number of queries as long as the norm of the Hamiltonian is polynomially bounded in the number of qubits (though our algorithm remains more efficient). Of course, the computational complexity would be exponential as we would need to output estimates for an exponential number of coefficients, but the coefficients could be computed in parallel on classical computers, which could drastically reduce the wall time needed.


\section{Robustness}\label{section: robustness}
A last point that we wish to raise about our protocol involves the question of robustness of the learning protocol.  Specifically, when learning the Hamiltonian there is always a possibility that the Hamiltonian that we wish to learn contains terms that are absent in the model we wish to describe it with.  For example, we could assume that the Hamiltonian is of the form $H = x_1 P_1 + x_2 P_2$ but the actual Hamiltonian used to generate the pseudo-Choi state is of the form $H_{\rm true} = x_1P_1 + x_2 P_2 + x_3 P_3$ for orthogonal Pauli operators $P_1,P_2,P_3$.  In this case, it is clear that our protocols will never be able to learn the true Hamiltonian model for the pseudo-Choi state without including the $P_3$ term in our model.  Nonetheless, we could ask ourselves if the protocol is robust to small errors in the Hamiltonian and if such errors are evident in the reconstruction.  Interestingly, we will see that our protocols are both robust and also error evident meaning that we can tell when terms are present in the true Hamiltonian that are not reflected in the model.

\begin{definition}[Under Specified Hamiltonian Learning Problem]
Let $H = \sum_{j=1}^M c_j H_j +\chi E$ where $E$ is an unknown Hermitian operator such that $\max_j |{\rm Tr}(E H_j)|=0$ and  $\frac{1}{d} {\rm Tr}(E^2)=1$ with $\chi$ an unknown multiplicative constant.  Our aim is to learn $\hat{{\bf c}}$ such that $\|{\bf c - \hat{c}}\|_2 \le \epsilon$ and an estimate $\hat{\chi}$ such that $|\chi - \hat{\chi}| \le \epsilon_\chi$ with probability greater than $1-\delta$
\end{definition}

Corollary~\ref{corollary: underspecified} considers the error in learning the coefficient $\chi$ of the extra terms $E$. In particular, it shows that if we allow the error in learning $\chi$ to be larger than some lower bound, the Under Specified HLP can be solved using a number of queries on the same order as for the Unitary HLP. Thus, it is possible to check if any extra terms are present in the Hamiltonian by estimating the value of $\chi$. Note that although the value of $\epsilon$ remains the same as for the Unitary HLP, the evolution time $t$ must be small enough to normalize the Hamiltonian with the extra terms in order for the protocol to function.

\begin{corollary}\label{corollary: underspecified}
There exists an algorithm that can solve the under specified Hamiltonian learning problem using $\mathcal{O}\left( \frac{M \log{(M/\delta)}}{t^2 \epsilon^2} \log{\left( \frac{M}{t \epsilon} \right)}\right)$ controlled applications of $e^{-iHt}$ and its inverse under the assumptions of~Theorem~\ref{thm:querycomplexityUB} if $\epsilon_\chi = \Omega(\frac{\epsilon \norm{\mathbf{\widetilde{c}}}_1}{|\chi|})$ and $\epsilon_s \gamma^2 \le 1$.
\end{corollary}

\begin{proof}   
The underspecified Hamiltonian learning problem can be solved in exactly the same manner as the fully specified learning problem.  Specifically, let us assume that we apply the algorithm $\mathtt{FindCoeffUnitary}(U, U^\dagger, \Delta, \boldsymbol{H}, n, N_s)$ where $U$ here is $e^{-iHt}$ and $\Delta$ is chosen to be $\pi/2t$ for $t\in \left( 0, \frac{1}{2\|H\|}\right)$.  Note that here we need to assume that this constant is chosen to be sufficiently large so that the erroneous Hamiltonian is also normalized given that $\chi\ne 0$.

Let us assume for the moment that we knew the identity of $E$ and performed the full learning protocol on the state. The only differences between the existing protocol and this one are that $M+1$ terms are present, and $t$ must be small enough to normalize the true Hamiltonian. It follows that the total number of queries to the unitary dynamics needed to learn the Hamiltonian scales from Theorem~\ref{thm:querycomplexityUB} as
\begin{align}
    N & \in \mathcal{O}\left( \frac{M \log{(M/\delta)}}{t^2 \epsilon^2} \log{\left( \frac{M}{t \epsilon} \right)}\right)
\end{align}
where the constants here are taken to be the true constants of the Hamiltonian with $E$ present.

Next assume that we construct the same Clifford shadow of the state as we would for solving the Unitary HLP. We then wish to reconstruct only $\mathbf{c}$ from the data.  We can achieve this by following the procedure $\mathtt{FindCoeffClifford}$ after noting that the values of the coefficients of $\mathbf{c}$ are computed independently of the coefficient of $E$.  Thus we can reconstruct $\mathbf{c}$ without reconstructing $E$ simply by skipping the $E$ reconstruction in the state (which wouldn't be possible anyways since the identity of $E$ is not known). Then from Theorem~\ref{thm:querycomplexityUB}, we are guaranteed that the error in the reconstructed $\hat{\mathbf{c}}$ is at most $\epsilon$ in this case and we learn $1/\gamma^2$ within precision $\epsilon_s$.  Thus we can learn the vector $\bf c$ within the required error from the exact same procedure discussed above.

The coefficient of $E$ can the be inferred from the result given that
\begin{equation}
    \gamma^2 = \frac{\|\tilde{c}\|_2^2 + \chi^2}{\Delta^2} +1.
\end{equation}
Thus the value of $\chi^2$ can be inferred from the residual value of the normalization after the vector $\hat{\bf c}$ is subtracted from it
\begin{equation}
    \hat{\chi}^2 = (\hat{\gamma}^2 -1)\Delta^2 - \|\tilde{c}\|_2^2
\end{equation}
given that we compute the value of $\gamma^{-2}$ within error $\epsilon_s$ this implies that
\begin{equation}
    |(\gamma^{-2} \pm \epsilon_s)^{-1} - \gamma^2| \le \frac{\gamma^2}{1 - \epsilon_s \gamma^2}-\gamma^2 = \epsilon_s \gamma^4 + O(\epsilon_s^2 \gamma^6).
\end{equation}
Thus if $\epsilon_s \gamma^2 \le 1$ then the remainder is $O(\epsilon_s \gamma^4)$. Recall that $\gamma^2 \in \mathcal{O}(1)$, so this assumption is not restrictive. Under the assumption that the magnitudes of the Hamiltonian coefficients are bounded above by $1$, we find by substitution of~\eqref{eq:epsc} that the error in our estimate of $\chi^2$ is at most 
\begin{align}
    \epsilon_{\chi^2} & = \sqrt{\left( \frac{\Delta^4 \gamma^8 \epsilon^2}{M (1 + \Delta^2)} \right) + 2\norm{\mathbf{\widetilde{c}}}_2^2 \epsilon^2} \nonumber \\
    & \in \mathcal{O}\left( \epsilon \sqrt{\left( \frac{\Delta^2 }{M} \right) + \norm{\mathbf{\widetilde{c}}}_2^2 } \right)
\end{align}
where we use the fact that $\gamma \in \mathcal{O}(1)$  and assume $\Delta \gg 1$ (which is reasonable, as $\Delta$ needs to normalize the Hamiltonian).

In order to understand how the error in our estimate of $\chi$ scales, we need to consider how an estimate of $\chi$ is found from $\chi^2$.  The simplest way of performing this is to simply take the square root of our estimate of $\chi^2$, which under the assumption that $ \abs{\epsilon_{\chi^2}} \le \chi^2/2$ leads us to the conclusion that 
\begin{align}
    \epsilon_\chi & \in O\left(\frac{\epsilon_{\chi^2}}{\abs{\chi}} \right). 
\end{align}
This results in 
\begin{align}
    \epsilon_\chi \in \mathcal{O}\left( \frac{\epsilon \sqrt{\left( \frac{\Delta^2 }{M} \right) + \norm{\mathbf{\widetilde{c}}}_2^2 }}{\abs{\chi}} \right)
\end{align}
Recall that $\Delta \in \mathcal{O}(\norm{H}_2)$, and so $\Delta \in \mathcal{O}(\norm{\widetilde{\mathbf{{c}}}}_1)$ if we assume the extra terms are smaller in magnitude than the terms that were expected in our model (that is, $\norm{\widetilde{\mathbf{e}}}_1 \leq \norm{\widetilde{\mathbf{{c}}}}_1$). Therefore, we have
\begin{align}
    \epsilon_\chi & \in \mathcal{O}\left( \frac{\epsilon \norm{\mathbf{\widetilde{c}}}_1}{\abs{\chi}}  \right).
\end{align}

Note that in the calculations above, we have assumed there was no error in the block-encoding of the additional Hamiltonian terms. In reality, the steps above correspond to the error in learning $\widetilde{\chi}$, the magnitude of the extra Hamiltonian terms in the block-encoding. The difference between $\chi$ and $\widetilde{\chi}$ can be upper bounded using Proposition~\ref{prop: coeffErrUB}, which leads to
\begin{align}
    \abs{\chi - \hat{\chi}} & \leq \abs{\chi - \widetilde{\chi}} + \abs{\widetilde{\chi} - \hat{\chi}} \nonumber \\
    & \leq   \frac{\epsilon_b}{t} + \epsilon_\chi
\end{align}

Recalling that in Theorem~\ref{thm:querycomplexityUB} the value of $\epsilon_b$ was specifically chosen such that $\frac{\sqrt{M} \epsilon_b }{t}$ was at most $\epsilon/2$, we see that
\begin{align}
    \abs{\chi - \hat{\chi}} & \in \mathcal{O}\left(   \frac{\epsilon}{2\sqrt{M}} + \frac{\epsilon \norm{\mathbf{\widetilde{c}}}_1}{|\chi|}  \right)    ,
\end{align}
and assuming the magnitude of the extra terms is smaller, or at least not much larger, than the terms that were expected (enough so that $|\chi|$  is smaller than something on the order of $\sqrt{M}\norm{\widetilde{\textbf{c}}}_1$), we have
\begin{align}
    \abs{\chi - \hat{\chi}} & \in  \mathcal{O}\left(  \frac{\epsilon \norm{\mathbf{\widetilde{c}}}_1}{|\chi|}    \right).
\end{align}

Thus if we insist that the error tolerance in the estimation of $\chi$ is in $\Omega\left(\frac{\epsilon \norm{\mathbf{\widetilde{c}}}_1}{|\chi|} \right)$ then we see that our accuracy criteria is automatically satisfied, and we can solve the under specified Hamiltonian learning problem.
\end{proof}

This shows that the protocol is robust to errors in the underlying Hamiltonian.  Specifically, it shows that if there are unknown terms in the Hamiltonian then shadow tomography will learn the correct coefficients of the known terms but also the coefficient of the (normalized) sum of the unknown terms will also be visible from the results.  This allows us to identify the presence of such terms, but our protocol provides no direct guidance other than the fact that the terms in question must be orthogonal to the known Hamiltonian terms.

In theory, since all additional terms can be represented as a sum of Pauli operators, one could systematically check the coefficients for the (k+1)-local terms, the (k+2)-local terms, etc. until all the non-zero ones are found (that is, when the magnitude of the additional coefficients is consistent with the magnitude of $\chi$). Furthermore, this could all be done without performing any additional measurements, as this is all performed on a classical computer using the classical shadow. The downside is that if one has no knowledge about the locality (in the Pauli basis) of the extra terms, then in the worst case this would require $\mathcal{O}(4^{n})$ time to check all possible Hamiltonian terms (though with a large enough classical memory, theoretically all these checks could be done in parallel), and with a polynomial number of queries, the resulting exponentially large vector of coefficients would only be epsilon accurate in the infinity norm. However, if one has some intuition about the potential identities of the extra terms, it is straightforward to check if they exist and learn their coefficients. In these cases, Corollary~\ref{corollary: underspecified} gives the following upper bound on the error in learning all the coefficients, including the ones that were not initially considered:
\begin{align}
    \norm{\mathbf{c} + \mathbf{e} - (\hat{\mathbf{c}} + \hat{\mathbf{e}})}_2 & \leq \epsilon + \epsilon_\chi,
\end{align}
where $\epsilon$ is the $l_2$ error we wanted to achieve in to satisfy the original Hamiltonian learning problem, $\mathbf{e}$ is the vector of coefficients of the extra Hamiltonian terms, and $\hat{\mathbf{e}}$ is our estimate of it. Again, note that we still require $t$ to be small enough to account for the extra terms, so it may be useful to choose a value of $t$ that is slightly smaller than needed for a k-local Hamiltonian, in case there are additional Hamiltonian terms. 

The remaining question that we wish to discuss is how accurate the preparation of our pseudo-Choi state needs to be in order to ensure that our estimate of the classical shadow is correct.  Let us assume for this discussion that for parameter $\omega$ and state operator $\rho^\perp$
\begin{align}
\tilde{\rho}_{c} := (1-\omega) \ketbra{\psi_c'}{\psi_c'} + \omega \rho^\perp.
\end{align}
This setting can cover not only cases where an inexact pseudo-Choi state is prepared but also the case where an unknown number of qubits are acted on by the Hamiltonian.  Although at first glance an explicit analysis similar to the previous corollary would seem appropriate for this case, conceptual hurdles arise in the way in which we define the pseudo-Choi state in this case if we do not know the underlying qubits.  The more natural setting is one in which the Hamiltonian terms are coupled to an unknown reservoir for the system which naturally leads to open system dynamics and in turn to a density matrix similar to that described above after a partial trace over the environmental degrees of freedom.  For this reason, we focus on this model for the system.

The shadow tomography reconstruction will then use such states to prepare the distributions $\tilde{\hat{\rho}}_i$ based on the measurement statistics that we glean from measuring the state in the computational basis.  Specifically we have from~\eqref{eq:TrrhoOl} and~\eqref{eq:gamma2} that the error in reconstructing each of the Hamiltonian coefficients, under the assumption that the spectral norm of each operator obeys $\|O_l\|\le 1$, is
\begin{align}
|{\rm Tr}(\tilde{\rho_c} O_l) - {\rm Tr}({\rho_c}' O_l)|\le \|\tilde{\rho_c} - \rho_c'\|_{\rm Tr}\le 2\omega
\end{align}
Similarly, our estimate of $1/\gamma^2$ will err by at most $2\omega$. Note that in addition to this error, which is caused by our inaccurate preparation of the pseudo-Choi state, we also have an additional error, $\epsilon_s$, due to the use of classical shadows in estimating the expectation values above. 

Let $c_l$ be the true coefficient of the $l^{th}$ Hamiltonian term, and let $c_l'$ be the estimate we would get if the pseudo-Choi state was prepared accurately (that is, $\omega = 0$). Finally, let $\widetilde{c}_l$ be the estimate we get from an inaccurate preparation of the pseudo-Choi state. The error in estimating this coefficient is then
\begin{align}
    \abs{c_l - \widetilde{c}_l} & \leq \abs{c_l - c_l'} + \abs{c_l' - \widetilde{c}_l} \nonumber \\
    & \leq 2 \gamma^2 \Delta (\epsilon_s + 2 \omega),
\end{align}
where the factor of $2 \gamma^2 \Delta$ on each term results from having to multiply the expectation values by $\gamma^2 \Delta$, and we have used the fact that the magnitudes of the Hamiltonian coefficients are upper bounded by 1, as well as the assumption that $\Delta \geq 2$ (which is reasonable since $\Delta \in \mathcal{O}(\norm{H}_2)$).

The error due to the inaccurate preparation of the pseudo-Choi state will not dominate the error from shadow tomography if this error is subdominant to $\epsilon_s$, which from~\eqref{eq:epsc} will occur if 
\begin{equation}
\omega \in O\left(\frac{\epsilon_c }{\sqrt{M}\gamma^2 \sqrt{\tilde{c}_{\max}^2 +\Delta^2}} \right).
\end{equation}
From this we see that the result is robust in the sense that the value of $\omega$ required shrinks at most polynomially with the number of terms in the Hamiltonian.  However, as the number of terms increases polynomially with the number of qubits for many systems of interest, this means that in the worst case scenario a very accurate preparation of the pseudo-Choi state will be needed in order to achieve the requisite accuracy for the entire Hamiltonian.

\section{Conclusion and Outlook}\label{section:conclusion}
In this work we have provided a new approach to Hamiltonian learning that draws inspiration from the state to channel isomorphism to yield a state from which the Hamiltonian can be directly learned and can be prepared using a simple quantum circuit. We are able to show that such approaches can easily learn $k$-local Hamiltonians, which is a broader class of Hamiltonians than the low-intersection case considered in most other Hamiltonian learning works, and in some cases can learn even much more general Hamiltonians without requiring exponentially many queries.  We further consider the learning problem in a scenario where amplitude estimation can be used to further accelerate it using mean-value estimation algorithms and achieve the optimal $1/\epsilon$ scaling with the error tolerance.

There are a number of interesting problems revealed by this work.  While our approach extends the scope of recent work on Hamiltonian learning (\cite{haah2021optimal, gu2022practical, PhysRevLett.130.200403}) to include more general Hamiltonians, a caveat is that we require backwards time evolution to generate the pseudo-Choi state. One area of interest for future work would be to find methods of generating pseudo-Choi-like states efficiently without the need for backwards time evolution, making our techniques very appealing for Hamiltonian learning under a much wider range of scenarios. Additionally, recent work by \textit{Huang et al.} \cite{PhysRevLett.130.200403} introduced an interesting method for reshaping Hamiltonians into non-interacting patches, allowing different parts of the Hamiltonian to be learned in parallel. If these methods could be applied to our approach, it may be possible to create pseudo-Choi states for separate parts of the Hamiltonian instead of the entire Hamiltonian, which could allow for longer evolution times, potentially reducing the total number of queries.

More broadly, however, this form of a block-encoding of the Hamiltonian within a quantum state allows much more flexible learning than existing approaches permit. While preparing pseudo-Choi states requires a trusted quantum computer, making them less appealing for near-term experimental settings, they open up the potential of tackling many different learning problems by similarly encoding information about the system within states. Such problems include Lindbladian learning for open systems, time-dependent Hamiltonian learning, and even applications of Hamiltonian learning methods to learn the structure of quantum walks.  It is our belief that these results represent a step towards a broader understanding of learning dynamical models for quantum systems and in turn a characterization of the physical models that can or cannot be efficiently learned from experiment.

\subsection*{Acknowledgements}
JC acknowledges funding from NSERC and CIFAR.  NW would like to acknowledge funding for this work from Google Inc. This material is based upon work supported by the U.S. Department of Energy, Office of Science, National Quantum Information Science Research Centers, Co-design Center for Quantum Advantage (C2QA) under contract number DE-SC0012704 (PNNL FWP 76274). 

\section{References}

\printbibliography

\appendix
 
\section{Estimating the Hamiltonian Coefficients (Pauli Shadows)}\label{App: PauliShadows}

Previously we considered the use of Clifford-based shadows procedure for estimating the Hamiltonian coefficients, but it is also possible to use the Pauli version of classical shadows. Since the operators involved in the analysis of the Pauli-based shadows are all single-qubit operators, we introduce some notation to indicate the specific 2 dimensional subspace an operator acts on, as well as the qubit that a measurement result corresponds to. First, since the subscript $i$ already denotes the round of measurement each bit string $\ket{b_i}$ came from, we use the notation $\ket{b_i[j]}$ to refer to specific bits of $\ket{b_i}$. More specifically, 
\begin{align}
    \ket{b_i} &= \bigotimes_{j=0}^{\eta-1} \ket{b_i[j]}
\end{align}
Once again, recall that $\ket{b_i}$ is a classical bit string, and Dirac notation is only used for presentation purposes. Second, the random unitaries, $U_i$, applied during each round of measurement are tensor products of single-qubit Clifford operators. Therefore, we write them as
\begin{align}
    U_i = \bigotimes_{j=0}^{\eta-1} U_{(i,j)},
\end{align}
where $U_{(i,j)}$ is a single-qubit operator acting on the $j^{th}$ qubit. Similarly, we can write the Hamiltonian terms $H_l$ as
\begin{align}
    H_l &= \bigotimes_{j=0}^{\eta-1} H_{(l,j)}
\end{align}

We again consider the process (viewed as a channel) of taking the expectation value over all measurement outcomes and unitaries in the group $\text{Cl}(2)^{\otimes \eta}$. Rather than a single depolarizing channel (as in the case where Clifford operations from $\text{Cl}(2^\eta)$ were randomly sampled), this can be viewed as $\eta$ depolarizing channels, each acting on a single qubit \cite{Huang_2020}. Inverting these channels as before gives the classical shadow specified in Definition~\ref{Def:ClassicalShadowPauli}.
    
    \begin{definition}[Classical Shadow (From Pauli Measurements)]\label{Def:ClassicalShadowPauli}
        Given $N$ copies of an $\eta$-qubit quantum state $\rho$, the classical shadows procedure based on random Pauli measurements returns a classical shadow of $\rho$ which is of the form
        \begin{align}
            \hat{\rho} &= \{\hat{\rho}_i | i \in \mathbb{Z}_N\},
        \end{align}
        where
        \begin{align}
            \hat{\rho}_i = \bigotimes^{\eta-1}_{j=0} \Big(3U_{(i,j)}^\dagger \ket{b_i[j]}\bra{b_i[j]} U_{(i,j)} -\Id \Big)
        \end{align}
        
        Recall that $i$ denotes the round of measurement and $j$ denotes the qubit; the unitary applied to $\rho$ during the $i^{th}$ round of the classical shadows is $U_i = \bigotimes_{j=0}^{\eta-1} U_{(i,j)}$, and $\ket{b_i[j]}$ refers to the outcome of measuring the $j^{th}$ qubit in the computational basis during the $i^{th}$ round of the classical shadows procedure.
        
    \end{definition}

Just as in the case where the Clifford flavour of classical shadows was used, the goal of this section is to extract the Hamiltonian coefficients from the pseudo-Choi state - that is, to solve the Pseudo-Choi Hamiltonian Learning problem (PC HLP). However, instead of generating a classical shadow of the pseudo-Choi state, we begin by performing a partial trace over the ancillary system $A$ to simplify the learning process. This results in the following state which will serve as our learning resource:
\begin{align}
    \rho = & \frac{1}{d\alpha^2}\Big(H^2 \otimes \ket{0}_C\bra{0}_C  + H \otimes \ket{0}_C\bra{1}_C   +  H \otimes \ket{1}_C\bra{0}_C +  \Id \otimes \ket{1}_C\bra{1}_C \Big)\label{eq:PartialChoiDensityOp},
\end{align}
where $d=2^n$ is the dimension of the Hilbert space the Hamiltonian acts on, and $\alpha = \sqrt{\norm{\mathbf{c}}_2^2 + 1}$. The set of decoding operators whose expectation values correspond to the Hamiltonian coefficients is also modified as per Definition~\ref{Def: PauliDecodingOperators}:

\begin{definition}[Decoding Operators] \label{Def: PauliDecodingOperators}
    Let the set of decoding operators be defined as
    \begin{align}
        \boldsymbol{O} & \equiv \{ O_l | l \in \mathbb{Z}_{M} \} 
    \end{align}
    where
    \begin{align}
        O_l &= \frac{H_l  \otimes X_C}{2}
    \end{align}
    and $H_l$ is one of the $M$ Hamiltonian terms. 
    
    Furthermore, we define
    \begin{align}
        O_{\alpha} & \equiv \Id \otimes \ket{1}_C\bra{1}_C
    \end{align}
\end{definition}

By Proposition~\ref{proposition: ComputeHamCoeffsPauli}, we have that
\begin{align}
    {\rm Tr}(\rho O_{l}) = \frac{c_l}{\alpha^2},
\end{align}
where $c_l$ is the $l^{th}$ Hamiltonian coefficient, and $\alpha$ is the normalization constant of the pseudo-Choi state \eqref{eq:alpha}. Furthermore,
\begin{align}
    {\rm Tr}(\rho O_{\alpha}) = \frac{1}{\alpha^2}
\end{align}

Now that we have a set of operators whose expectation values correspond to the Hamiltonian coefficients when acting on the resource state $\rho$ \eqref{eq:PartialChoiDensityOp}, we seek to approximate these expectation values using a classical shadow of $\rho$. Recall that to obtain $\rho$ we traced out the ancillary system $A$ from the pseudo-Choi state \eqref{eq:ChoiDensityOp}. To generate a classical shadow of $\rho$, starting from the pseudo-Choi state, we must therefore ignore the ancilla system $A$ during the classical shadows procedure, meaning that all random (single-qubit) Clifford operations and measurements are only on $n + 1$ qubits ($n$ qubits for the system on Hilbert space $\mathcal{H}_S$ and one qubit for ancilla on Hilbert space $\mathcal{H}_C$). In terms of the definition of the classical shadow, we let $\eta = n+1$ so that a classical shadow is made up of $N$ components of the form
\begin{align}
    \hat{\rho}_i & = \bigotimes_{j=0}^{n} \Big( 3U_{(i,j)}^\dagger \ket{b_i[j]}\bra{b_i[j]}U_{(i,j)} - \Id\Big)
\end{align}
With this, and the definition of the decoding operators given in Definition~\ref{Def: PauliDecodingOperators}, we can write the expectation values of interest in the explicit forms given in Claim~\ref{claim: PauliExpectations}.
\begin{claim}\label{claim: PauliExpectations}
Let $X_C$ denote the single-qubit Pauli X operation acting on $\mathcal{H}_C$, and let $H_{l,j}$ be the single qubit Pauli operator from $H_l$ that acts on the $j^{th}$ qubit (recall that each $H_l$ is an $n$-qubit Pauli operator on $\mathcal{H}_S$).
\begin{align}
	{\rm Tr}(\hat{\rho}_i O_{l}) &= \frac{3}{2} \bra{b_i[n]} U_{(i,n)} X_C U_{(i,n)}^\dagger \ket{b_i[n]}\prod_{j=0}^{n-1} \Big(3 \bra{b_i[j]} U_{(i,j)} H_{(l,j)} U_{(i,j)}^\dagger \ket{b_i[j]} - {\rm Tr}(H_{(l,j)}) \Big)  \label{eq:EfficientTrace4},
\end{align}

Similarly,
\begin{align}
    {\rm Tr}(\hat{\rho}_i O_\alpha) & = 3 \abs{\bra{b_i[n]} U_{(i,n)}\ket{1}_C }^2 - 1 \label{eq:EfficientTrace5}
\end{align}
\end{claim}

\begin{proof}[Proof of Claim~\ref{claim: PauliExpectations}]
    Recall that $\hat{\rho}_i$ is on $n+1$ qubits, $H_l$ is on $n$ qubits, and $X_C$ is the single qubit Pauli X operator. Then,
    \begin{align}
        \begin{split}
        	{\rm Tr}[\hat{\rho}_i O_{l}] =& {\rm Tr}\left[ \left(\bigotimes_{j=0}^{n} \Big( 3U_{(i,j)}^\dagger \ket{b_i[j]}\bra{b_i[j]}U_{(i,j)} - \Id\Big)\right) \left(\frac{H_l\otimes X_C}{2}\right)   \right]\\
        	=& \frac{1}{2}{\rm Tr}\left[ \left(\bigotimes_{j=0}^{n} \Big( 3U_{(i,j)}^\dagger \ket{b_i[j]}\bra{b_i[j]}U_{(i,j)} - \Id\Big)\right) \left(\left(\bigotimes_{j=0}^{n-1} H_{(l,j)}\right)\otimes X_C \right)  \right]\\
        	= & \frac{1}{2}{\rm Tr}\Bigg[ \Bigg(\bigotimes_{j=0}^{n-1} \Big( 3U_{(i,j)}^\dagger \ket{b_i[j]}\bra{b_i[j]}U_{(i,j)} H_{(l,j)} - H_{(l,j)} \Big)\Bigg) \\
            & \otimes \Bigg(3U_{(i,n)}^\dagger\ket{b_i[n]}\bra{b_i[n]}U_{(i,n)} X_C - X_C\Bigg) \Bigg]\\
        	=& \frac{3}{2}{\rm Tr}\Bigg[ U_{(i,n)}^\dagger\ket{b_i[n]}\bra{b_i[n]}U_{(i,n)} X_C  \Bigg] \prod_{j=0}^{n-1} {\rm Tr}\Bigg[ 3U_{(i,j)}^\dagger \ket{b_i[j]}\bra{b_i[j]}U_{(i,j)} H_{(l,j)} - H_{(l,j)} \Bigg]\\
        	=& \frac{3}{2} \bra{b_i[n]}U_{(i,n)} X_C U_{(i,n)}^\dagger\ket{b_i[n]}  \prod_{j=0}^{n-1}  \left(3 \bra{b_i[j]}U_{(i,j)} H_{(l,j)}U_{(i,j)}^\dagger \ket{b_i[j]} - {\rm Tr}\left[H_{(l,j)} \right]\right),
    	\end{split}
    \end{align}
    which gives the result of equation \eqref{eq:EfficientTrace4}. 
    
    Likewise,
    \begin{align}
        \begin{split}
            {\rm Tr}[\hat{\rho}_i O_\alpha]  =& {\rm Tr}\left[ \Bigg( \bigotimes_{j=0}^{n} \Big( 3U_{(i,j)}^\dagger \ket{b_i[j]}\bra{b_i[j]}U_{(i,j)} - \Id\Big)\Bigg)   \Bigg( \Id_S \otimes \ket{1}_C \bra{1}_C\Bigg) \right]\\
            =& {\rm Tr}\Bigg[ \Bigg( \bigotimes_{j=0}^{n-1} \Big( 3U_{(i,j)}^\dagger \ket{b_i[j]}\bra{b_i[j]}U_{(i,j)} - \Id\Big)\Bigg)   \\
            & \otimes\Bigg( \Big( 3U_{(i,n)}^\dagger \ket{b_i[n]}\bra{b_i[n]}U_{(i,n)} - \Id\Big) \ket{1}_C \bra{1}_C\Bigg) \Bigg]\\
            = &  \prod_{j=0}^{n-1} \Bigg({\rm Tr}\left[  3U_{(i,j)}^\dagger \ket{b_i[j]}\bra{b_i[j]}U_{(i,j)}\right] - {\rm Tr}\left[ \Id\right]\Bigg)   \\
            & \Bigg(  3\bra{b_i[n]}U_{(i,n)}\ket{1}_C \bra{1}_C U_{(i,n)}^\dagger \ket{b_i[n]} - 1 \Bigg)\\
            =&   3\bra{b_i[n]}U_{(i,n)}\ket{1}_C \bra{1}_C U_{(i,n)}^\dagger \ket{b_i[n]} - 1,
        \end{split}
    \end{align}
    which gives the result of equation \eqref{eq:EfficientTrace5}, completing the proof.
\end{proof}

Since only single qubit operations are required, computing each trace takes $\mathcal{O}(n)$ time, and requires $\mathcal{O}(n)$ memory to store all intermediate values. In fact, if we focus on $k$-local Hamiltonians, each Hamiltonian term $H_l$ only contains at most $k$ non-identity Paulis, so only $\mathcal{O}(k)$ operations are required to compute each value.

As was the case with the random Clifford measurements, it is unnecessary to store the entire classical shadow:

\begin{definition}\label{Def: compressedShadowPauli}
    If the random Pauli measurement version of the classical shadows procedure is performed, the resulting classical shadow of size $N$ can be stored in the following way:
    \begin{align}
        \ket{\hat{\rho}} &= \{\ket{\hat{\rho}_{i}} | i \in \mathbb{Z}_N \},
    \end{align}
    where
    \begin{align}
        \ket{\hat{\rho}_{i}} &= \bigotimes_{j=1}^{n+1} U_{(i,j)}^\dagger \ket{b[j]}_{i}
    \end{align}
    Here we use Dirac notation to describe the resultant bit-strings, however note that as $\ket{b[j]}_{i}$ is a classical bit, so $\ket{\hat{\rho}_i}$ is a computational basis vector which corresponds to a one-hot unary representation of a bit string.
\end{definition}

Algorithm~\ref{algorithm: FindCoeffPauli} describes the full procedure for estimating the Hamiltonian coefficients. Note that it uses the ``compressed'' Pauli shadows given by Definition~\ref{Def: compressedShadowPauli}. Also, there is no need to use the Tableau or SIP Algorithms of \cite{Aaronson_2004} to store intermediary values or compute the expectation values since all the operations involved are on single qubits.

\begin{algorithm}[t!]
\caption{$\mathtt{FindCoeffPauli}(\rho_c^{\otimes N}, \boldsymbol{H}, n, N)$: Determine the Vector of Hamiltonian Coefficients from $\rho_c$}\label{algorithm: FindCoeffPauli}
\textbf{Input:} 

$\rho_c^{\otimes N} $: Collection of $N$ pseudo-Choi states on $n$ qubits, each on Hilbert space $\mathcal{H}_S \otimes \mathcal{H}_A \otimes \mathcal{H}_C$

$\boldsymbol{H}$: Set of M k-local Hamiltonian terms $H_l$ whose coefficients $c_m$ are desired

\Comment{$\rho_c^{\otimes N} $ is a quantum state, whereas $\boldsymbol{H}$ is a classical set of operators}

\begin{enumerate}
    \item \textbf{Initialize the array of Hamiltonian coefficients}
    
    Let $\hat{\textbf{c}} \gets \begin{bmatrix} 0, 0, ..., 0 \end{bmatrix}_{1\times M}$ and let $\hat{\textbf{c}}_l$ denote the $l^{th}$ element of $\hat{\textbf{c}}$
    
    \item \textbf{Generate and store the classical shadow}
        
    Use copies of $\rho_c$ to generate $\hat{\rho}$, a classical shadow of size $N$ of ${\rm Tr}_A(\rho_c)$, by limiting the random Pauli measurement based classical shadows procedure to the space $\mathcal{H}_S \otimes \mathcal{H}_C$ to effectively perform the partial trace over $\mathcal{H}_A$.
    
    Store $\hat{\rho}$ in classical memory (i.e. store each vector $\ket{\hat{\rho}_i
    }$)

    \Comment{See Definition~\ref{Def: compressedShadowPauli} of a classical shadow}

    \item \textbf{Generate an estimate of ${\rm Tr}(\rho O_{l}) = \frac{c_l}{\alpha^2}$ for each corresponding Hamiltonian term}
    
    Let $\hat{\textbf{o}} \gets \begin{bmatrix} 0, 0, ..., 0 \end{bmatrix}_{1\times N}$ and let $\hat{\textbf{o}}_i$ denote the $i^{th}$ element of $\hat{\textbf{o}}$
    
    $\mathtt{For}$ $l$ $\mathtt{from}$ $0$ $\mathtt{to}$ $M-1$:
    
       \qquad $\mathtt{For}$ $i$ $\mathtt{from}$ $0$ to $N-1$:
        
        \qquad \qquad $\hat{o}_i \gets \frac{3}{2} \bra{b_i[n]} U_{(i,n)} X_C U_{(i,n)}^\dagger \ket{b_i[n]}\prod_{j=0}^{n-1} \Big(3 \bra{b_i[j]} U_{(i,j)} H_{(l,j)} U_{(i,j)}^\dagger \ket{b_i[j]} - {\rm Tr}(H_{(l,j)}) \Big)$
        
        \Comment{Equal to ${\rm Tr}\left(\hat{\rho_i} O_l\right)$ as per Equation \eqref{eq:EfficientTrace4}}

        \qquad $\hat{\textbf{c}}_l \gets \mathtt{MedianOfMeans}(\hat{\textbf{o}})$

    \item \textbf{Generate an estimate of ${\rm Tr}(\rho O_{\alpha}) = \frac{1}{\alpha^2}$}
    
    $\mathtt{For}$ $i$ $\mathtt{from}$ $0$ to $N-1$:
        
    \qquad $\hat{o}_i \gets 3 \abs{\bra{b_i[n]} U_{(i,n)}\ket{1}_C }^2 - 1$
    \Comment{Equal to ${\rm Tr}(\hat{\rho}_i O_\alpha)$ as per Equation \eqref{eq:EfficientTrace5}}

    $\hat{o}_{\alpha} \gets \mathtt{MedianOfMeans}(\hat{\textbf{o}})$
    
    \item \textbf{Remove the factor of $\frac{1}{\alpha^2}$ from the Hamiltonian coefficients}
    
    Let $\hat{\textbf{c}} \gets \frac{\hat{\textbf{c}}}{\hat{o}_{\alpha}}$ 
    
    \item \textbf{Output the vector of Hamiltonian coefficients}
    
    Return $\hat{\textbf{c}} = 
    \begin{bmatrix} \hat{c}_0, \hat{c}_1, ..., \hat{c}_{M-1} \end{bmatrix}$
\end{enumerate}
\end{algorithm}

Since all the operations involved in this process are single-qubit operations, computing the required inner products is more efficient than in the case of Clifford-based shadows. In particular, the number of computations required to do so is at most $\mathcal{O}(MN)$ where a sufficient value of $N$ is given in equation~\eqref{eq:NbdPauli}. Note that this value of $N$ includes a factor of $4^k$, where $k$ is the locality of the Hamiltonian. Meanwhile, the Clifford-based procedure requires $\mathcal{O}(MNn^3)$ computations to calculate the necessary inner products, where the sufficient value of $N$ is given in Theorem~\ref{thm:samplecomplexityUB} and does not contain the factor of $4^k$. Therefore, unless the value of $k$ is large, the Pauli-based procedure will have the better post-processing performance, at the expense of having a slightly higher sample complexity. Furthermore, as in the Clifford case, these operations can be done in parallel, greatly reducing the computation time.


\subsection{Sample Complexity Upper Bound}
The sample complexity of this method is very similar to that of the case where Clifford based classical shadows was used to learn the pseudo-Choi state. The key difference is that when the random Pauli measurement version of classical shadows is used, the upper bound on the shadow norm in Lemma~\ref{lem:shadownormUB} no longer applies. Instead, we have the following result:
\begin{lemma} [Shadow Norm Upper Bound (Pauli Shadows)\label{lem:shadownormUBPauli}]{
    Let the set \{$O_{\alpha}$\} $\cup$  $\{O_l \mid l \in \mathbb{Z}_{M}\}$ be denoted by $\textbf{O} \equiv \{O_i \mid i \in \mathbb{Z}_{2M}\}$. Note that these operators act non-trivially on at most $k+1$ qubits. If using the random Pauli measurement version of the classical shadows procedure to predict the expectation values of the operators in the set, the shadow norm of each operator is at most 
    \begin{align}
        \norm{O_{i}}^2_{shadow} \leq 3^{k+1}
    \end{align}
}
\end{lemma}

\begin{proof}[Proof of Lemma~\ref{lem:shadownormUBPauli}]
    This results directly from \cite{Huang_2020}, where Huang et al. showed that for classical shadows based on random Pauli measurements
    \begin{align}
        \norm{O}^2_{shadow} \leq 3^{k}
    \end{align}
    where $O$ is a single $k$-local Pauli operator.
    
\end{proof}

By the same argument as in the case of Clifford based classical shadows, this upper bound on the shadow norm leads to a sample complexity of
\begin{align}\label{eq:NbdPauli}
    N \in \mathcal{O}\left( \frac{3^k \alpha^4 (c_{\rm max} + 1) M\log\left(M/\delta\right) }{\epsilon^2}  \right)
\end{align}

\section{Proofs and technical results}
\subsection{Proof of Equation \eqref{eq:alpha}}\label{section:randomsection2}
\begin{proof}
    \begin{align}
        \bra{\Phi_d}_{SA}(H^2\otimes \Id_A) \ket{\Phi_d}_{SA}
        & = \frac{\sum_{i=0}^{d-1}\bra{i}_S \bra{i}_A (H^2\otimes \Id_A) \sum_{j=0}^{d-1}\ket{j}_S \ket{j}_A}{d} \nonumber \\
        & = \frac{\sum_{i=0}^{d-1}\sum_{j=0}^{d-1}\bra{i}_S H^2 \ket{j}_S \bra{i}_A\ket{j}_A}{d} \nonumber \\
        & = \frac{\sum_{i=0}^{d-1} \bra{i}_S H^2\ket{i}_S}{d} \nonumber \\
        & = \frac{\text{Tr}\left(H^2\right)}{d} \nonumber \\
        & = \frac{1}{d}\text{Tr}\left(\sum_{l=0}^{M-1}c_l  H_l \sum_{m=0}^{M-1} c_m H_m \right) \nonumber \\
        & = \frac{1}{d} \sum_{l=0}^{M-1}\sum_{m=0}^{M-1}c_l c_m\text{Tr}( H_l H_m) \nonumber \\
        & =  \sum_{m=0}^{M-1}c_m^2 \nonumber \\
        &= \norm{\textbf{c}}_2^2
    \end{align}
    
    Therefore
    \begin{align}
        \alpha & := \sqrt{\bra{\Phi_d}_{SA}(H^2\otimes \Id_A) \ket{\Phi_d}_{SA} + 1} \nonumber \\
        & = \sqrt{\norm{\textbf{c}}_2^2 + 1}
    \end{align}
\end{proof}

\subsection{Proof of Proposition~\ref{Prop: EfficientTrace}}\label{section:EfficientTrace}

Proposition~\ref{Prop: EfficientTrace} contains equations showing how to explicitly compute three expectation values of interest using a classical shadow. This section outlines the derivation of each of these three expressions. 

Recall that classical shadows generated by the random Clifford measurement version of the classical shadows procedure are of the form $\hat{\rho}_i = \left( 2^{2n+1} + 1\right) U^\dagger_i \ket{b_i}\bra{b_i} U_i - \Id$, and recall the definitions of the following operators:
\begin{align*}
    O_l^{+} & = O_l + \left(O_l\right)^\dagger\\
    O_l^{-} & = iO_l - i\left(O_l\right)^\dagger
\end{align*}
where $O_l = \left(H_l \otimes \Id\right) \ket{\Phi_d}_{SA} \bra{\Phi_d}_{SA} \otimes \ket{0}_C\bra{1}_C$ is a decoding operator as in Definition~\ref{Def: CliffDecodingOperators}.

Also, note that $U_i^\dagger \ket{b_i}$ is on the Hilbert space $\mathcal{H}_S \otimes\mathcal{H}_A \otimes\mathcal{H}_C$, and $H_l$ acts on the Hilbert space $\mathcal{H}_S$.

\begin{proof}[Proof of equation \eqref{eq:EfficientTrace1}]

\begin{align}
    \begin{split}
        {\rm Tr}\left(\hat{\rho_i} O_l^+\right)=& {\rm Tr}\Big(\left[ \left( 2^{2n+1} + 1 \right)U_i^\dagger \ket{b_i}\bra{b_i} U_i - \Id \right]   \Big[ \left(H_l \otimes \Id_A \right) \ket{\Phi_d}_{SA} \bra{\Phi_d}_{SA} \otimes \ket{0}_C\bra{1}_C  \\
    	&+ \ket{\Phi_d}_{SA} \bra{\Phi_d}_{SA}\left(H_l \otimes \Id_A \right)  \otimes \ket{1}_C\bra{0}_C \Big] \Big)\\
    	=&  \left( 2^{2n+1} + 1 \right){\rm Tr}\left( \left[U_i^\dagger \ket{b_i}\bra{b_i} U_i \right] \Big[\left(H_l \otimes \Id_A \right) \ket{\Phi_d}_{SA} \bra{\Phi_d}_{SA} \otimes \ket{0}_C\bra{1}_C \Big]    \right) \\
    	&- {\rm Tr}\Big(\left(H_l \otimes \Id_A \right) \ket{\Phi_d}_{SA} \bra{\Phi_d}_{SA} \otimes \ket{0}_C\bra{1}_C\Big)\\
    	&+  \left( 2^{2n+1} + 1 \right){\rm Tr}\left( \left[U_i^\dagger \ket{b_i}\bra{b_i} U_i\right] \Big[\ket{\Phi_d}_{SA} \bra{\Phi_d}_{SA}\left(H_l \otimes \Id_A \right)  \otimes \ket{1}_C\bra{0}_C \Big]   \right) \\
    	&- {\rm Tr}\Big(\ket{\Phi_d}_{SA} \bra{\Phi_d}_{SA}\left(H_l \otimes \Id_A \right)  \otimes \ket{1}_C\bra{0}_C\Big)\\
    	=&  \left( 2^{2n+1} + 1 \right){\rm Tr}\left( \left[U_i^\dagger \ket{b_i}\bra{b_i} U_i\right] \Big[ \left(H_l \otimes \Id_A \right) \ket{\Phi_d}_{SA} \bra{\Phi_d}_{SA} \otimes \ket{0}_C\bra{1}_C  \Big]  \right) \\
    	&- {\rm Tr}\Big(\left(H_l \otimes \Id_A \right) \ket{\Phi_d}_{SA} \bra{\Phi_d}_{SA} \Big) {\rm Tr}\left( \ket{0}_C\bra{1}_C\right)\\
    	&+  \left( 2^{2n+1} + 1 \right){\rm Tr}\left( \left[U_i^\dagger \ket{b_i}\bra{b_i} U_i\right] \Big[ \ket{\Phi_d}_{SA} \bra{\Phi_d}_{SA}\left(H_l \otimes \Id_A \right) \otimes \ket{1}_C\bra{0}_C \Big]   \right) \\
    	&- {\rm Tr}\Big(\ket{\Phi_d}_{SA} \bra{\Phi_d}_{SA}\left(H_l \otimes \Id_A \right)\Big){\rm Tr}\Big( \ket{1}_C\bra{0}_C\Big)\\
    	=&  \left( 2^{2n+1} + 1 \right){\rm Tr}\left( \left[U_i^\dagger \ket{b_i}\bra{b_i} U_i \right] \Big[ \left(H_l \otimes \Id_A \right) \ket{\Phi_d}_{SA} \bra{\Phi_d}_{SA} \otimes \ket{0}_C\bra{1}_C  \Big]  \right) \\
    	&+  \left( 2^{2n+1} + 1 \right){\rm Tr}\left( \left[U_i^\dagger \ket{b_i}\bra{b_i} U_i \right] \Big[ \ket{\Phi_d}_{SA} \bra{\Phi_d}_{SA}\left(H_l \otimes \Id_A \right) \otimes \ket{1}_C\bra{0}_C  \Big]  \right) \\
    	=&  \left( 2^{2n+1} + 1 \right) \bra{b_i} U_i \Big( \left(H_l \otimes \Id_A \right) \ket{\Phi_d}_{SA} \bra{\Phi_d}_{SA} \otimes \ket{0}_C\bra{1}_C    \Big)U_i^\dagger \ket{b_i} \\
    	&+  \left( 2^{2n+1} + 1 \right)\bra{b_i} U_i\Big(  \ket{\Phi_d}_{SA} \bra{\Phi_d}_{SA}\left(H_l \otimes \Id_A \right)  \otimes \ket{1}_C\bra{0}_C    \Big)U_i^\dagger \ket{b_i} \\
    	=&  \left( 2^{2n+1} + 1 \right) \Big(\bra{b_i} U_i  \left(H_l \otimes \Id_A \right) \ket{\Phi_d}_{SA}\otimes \ket{0}_C\Big) \Big( \bra{\Phi_d}_{SA} \otimes \bra{1}_C  U_i^\dagger \ket{b_i}\Big) \\
    	&+  \left( 2^{2n+1} + 1 \right)\Big( \bra{b_i} U_i  \ket{\Phi_d}_{SA}\otimes \ket{1}_C\Big) \Big( \bra{\Phi_d}_{SA}\left(H_l \otimes \Id_A \right)  \otimes \bra{0}_C  U_i^\dagger \ket{b_i}\Big)
    \end{split}
\end{align}

\end{proof}

\begin{proof}[Proof of equation \eqref{eq:EfficientTrace2}]
\begin{align}
	\begin{split}
		{\rm Tr}\left(\hat{\rho_i} O_l^-\right)=& {\rm Tr}\Big(\left[ \left( 2^{2n+1} + 1 \right)U_i^\dagger \ket{b_i}\bra{b_i} U_i - \Id \right]   \Big[ i \left(H_l \otimes \Id_A \right) \ket{\Phi_d}_{SA} \bra{\Phi_d}_{SA} \otimes \ket{0}_C\bra{1}_C \\
		& -i \ket{\Phi_d}_{SA} \bra{\Phi_d}_{SA}\left(H_l \otimes \Id_A \right)  \otimes \ket{1}_C\bra{0}_C \Big] \Big)\\
		=&  i\left( 2^{2n+1} + 1 \right){\rm Tr}\left( \Big[ U_i^\dagger \ket{b_i}\bra{b_i} U_i\Big] \Big[ \left(H_l \otimes \Id_A \right) \ket{\Phi_d}_{SA} \bra{\Phi_d}_{SA} \otimes \ket{0}_C\bra{1}_C  \Big]   \right) \\
		&- i{\rm Tr}\left(\left(H_l \otimes \Id_A \right) \ket{\Phi_d}_{SA} \bra{\Phi_d}_{SA} \otimes \ket{0}_C\bra{1}_C\right)\\
		&-i  \left( 2^{2n+1} + 1 \right){\rm Tr}\left( \Big[ U_i^\dagger \ket{b_i}\bra{b_i} U_i \Big] \Big[ \ket{\Phi_d}_{SA} \bra{\Phi_d}_{SA}\left(H_l \otimes \Id_A \right)  \otimes \ket{1}_C\bra{0}_C  \Big]  \right) \\
		&+i {\rm Tr}\left(\ket{\Phi_d}_{SA} \bra{\Phi_d}_{SA}\left(H_l \otimes \Id_A \right)  \otimes \ket{1}_C\bra{0}_C\right)\\
		=&  i\left( 2^{2n+1} + 1 \right) \bra{b_i} U_i \Big( \left(H_l \otimes \Id_A \right) \ket{\Phi_d}_{SA} \bra{\Phi_d}_{SA} \otimes \ket{0}_C\bra{1}_C    \Big)U_i^\dagger \ket{b_i} \\
		&-i  \left( 2^{2n+1} + 1 \right)\bra{b_i} U_i\Big(  \ket{\Phi_d}_{SA} \bra{\Phi_d}_{SA}\left(H_l \otimes \Id_A \right)  \otimes \ket{1}_C\bra{0}_C    \Big)U_i^\dagger \ket{b_i} \\
		=&  i\left( 2^{2n+1} + 1 \right) \Big(\bra{b_i} U_i  \left(H_l \otimes \Id_A \right) \ket{\Phi_d}_{SA}\otimes \ket{0}_C\Big) \Big( \bra{\Phi_d}_{SA} \otimes \bra{1}_C    U_i^\dagger \ket{b_i}\Big) \\
		&-i  \left( 2^{2n+1} + 1 \right)\Big( \bra{b_i} U_i  \ket{\Phi_d}_{SA}\otimes \ket{1}_C\Big) \Big( \bra{\Phi_d}_{SA}\left(H_l \otimes \Id_A \right)  \otimes \bra{0}_C    U_i^\dagger \ket{b_i}\Big)
    \end{split}
\end{align}
\end{proof}

\begin{proof}[Proof of equation \eqref{eq:EfficientTrace3}]
Recall from Definition~\ref{Def: CliffDecodingOperators} that $O_{\alpha} = \ket{\Phi_d}_{SA} \bra{\Phi_d}_{SA} \otimes \ket{1}_C\bra{1}_C$ is the decoding operator corresponding to the normalization constant $\alpha$. From this it follows that
\begin{align}
    \begin{split}
        {\rm Tr}( \hat{\rho}_i O_\alpha) =&   {\rm Tr}\Big( \left[\left( 2^{2n+1} + 1\right) U^\dagger_i \ket{b_i}\bra{b_i} U_i - \Id\right] \Big[ \ket{\Phi_d}_{SA} \bra{\Phi_d}_{SA} \otimes \ket{1}_C\bra{1}_C \Big]   \Big)\\
        =& \left( 2^{2n+1} + 1\right) {\rm Tr}\Big( \left[ U^\dagger_i \ket{b_i}\bra{b_i} U_i\right] \Big[ \ket{\Phi_d}_{SA} \bra{\Phi_d}_{SA} \otimes \ket{1}_C\bra{1}_C \Big]   \Big)\\
        &- {\rm Tr}\Big(\ket{\Phi_d}_{SA} \bra{\Phi_d}_{SA} \otimes \ket{1}_C\bra{1}_C \Big)\\
        =& \left( 2^{2n+1} + 1\right)  \bra{b_i} U_i \Big(\ket{\Phi_d}_{SA} \bra{\Phi_d}_{SA} \otimes \ket{1}_C\bra{1}_C \Big) U^\dagger_i \ket{b_i}  - 1\\
        =& \left( 2^{2n + 1} + 1 \right) \abs{\bra{b_i}U_i \ket{\phi_d}_{SA}\ket{1}_C}^2 - 1 
    \end{split}
\end{align}
\end{proof}


\subsection{Proof of Lemma~\ref{lem:shadownormUB}}\label{App:shadownormUB}
\begin{proof}[Proof of Lemma~\ref{lem:shadownormUB}]
Recall the definitions of the Hermitian operators in Lemma~\ref{lem:shadownormUB}:
\begin{align}
    O_{\alpha} &= \ket{\Phi_d}_{SA} \bra{\Phi_d}_{SA} \otimes \ket{1}\bra{1},
\end{align}
and
\begin{align}
    O_l^{+} &= O_l + \left(O_l\right)^\dagger\\
    O_l^{-} &= iO_l - i\left(O_l\right)^\dagger
\end{align}

where
\begin{align}
    O_l = \left(H_l \otimes \Id\right) \ket{\Phi_d}_{SA} \bra{\Phi_d}_{SA} \otimes \ket{0}\bra{1}  
\end{align}
is defined for all $0 \leq l \leq M-1$.

First consider the square of the Hilbert-Schmidt norm of $O_l^{+}$:
\begin{align}
    \text{Tr}\left(\left(O_l^{+}\right)\left(O_l^{+}\right)^\dagger\right) & =\text{Tr}\left(\left(O_l^{+}\right)^2\right) \nonumber \\
    & = \text{Tr}\left(O_l^2\right) + \text{Tr}\left(O_l O_l^\dagger \right) + \text{Tr}\left(O_l^\dagger O_l \right) + \text{Tr}\left(\left(O_l^\dagger \right)^2\right) \nonumber \\
    & = \text{Tr}\left(O_l^2\right) + 2\text{Tr}\left(O_l O_l^\dagger \right) + \text{Tr}\left(\left(O_l^\dagger \right)^2\right)
\end{align}

Examining each term above individually, we see that:
\begin{align}
    \text{Tr}\left(O_l^2\right) & = \text{Tr}\left(\left( \left(H_l \otimes \Id\right) \ket{\Phi_d}_{SA} \bra{\Phi_d}_{SA} \otimes \ket{0}\bra{1}\right)^2 \right) \nonumber \\
    & = \text{Tr}\Big(\left( \left(H_l \otimes \Id\right) \ket{\Phi_d}_{SA} \bra{\Phi_d}_{SA} \otimes \ket{0}\bra{1}\right)\left( \left(H_l \otimes \Id\right) \ket{\Phi_d}_{SA} \bra{\Phi_d}_{SA} \otimes \ket{0}\bra{1}\right) \Big) \nonumber \\
    & = \text{Tr}\Big(\bra{\Phi_d}_{SA}\left( H_l \otimes \Id\right) \ket{\Phi_d}_{SA} \bra{\Phi_d}_{SA}  \left(H_l \otimes \Id\right) \ket{\Phi_d}_{SA}   \Big)\text{Tr}\Big(  \ket{0}\bra{1}\ket{0}\bra{1}\Big) \nonumber \\
    & = 0
\end{align}

Similarly, 
\begin{align}
    \text{Tr}\left(\left(O_l^\dagger \right)^2\right) &=0
\end{align}

\begin{align}
    \text{Tr}\left(O_l O_l^\dagger \right) & = \text{Tr}\left(\Big( \left(H_l \otimes \Id\right) \ket{\Phi_d}_{SA} \bra{\Phi_d}_{SA} \otimes \ket{0}\bra{1}\Big) \Big( \left(H_l \otimes \Id\right) \ket{\Phi_d}_{SA} \bra{\Phi_d}_{SA} \otimes \ket{0}\bra{1}\Big)^\dagger \right) \nonumber \\
    & = \text{Tr}\left(\Big( \left(H_l \otimes \Id\right) \ket{\Phi_d}_{SA} \bra{\Phi_d}_{SA} \otimes \ket{0}\bra{1}\Big) \Big(  \ket{\Phi_d}_{SA} \bra{\Phi_d}_{SA} \left(H_l^\dagger \otimes \Id\right) \otimes \ket{1}\bra{0}\Big) \right) \nonumber \\
    & = \text{Tr}\left( \left(H_l \otimes \Id\right) \ket{\Phi_d}_{SA}  \bra{\Phi_d}_{SA} \left(H_l^\dagger \otimes \Id\right)  \right)\text{Tr}\left( \ket{0}\bra{1}\ket{1}\bra{0} \right) \nonumber \\
    &= \bra{\Phi_d}_{SA} \left(H_l^\dagger \otimes \Id\right)\left(H_l \otimes \Id\right) \ket{\Phi_d}_{SA} \nonumber  \\
    &= \bra{\Phi_d}_{SA} \ket{\Phi_d}_{SA} \nonumber  \\
    &= 1
\end{align}
where we have used the fact that $H_l$ is part of an orthonormal basis, so $H_l^\dagger H_l = \Id$.

From this we arrive at the following upper bound:
\begin{align}
    \text{Tr}\left(\left(O_l^{+}\right)^2\right) &= 2
\end{align}
Likewise:
\begin{align}
    \text{Tr}\left(\left(O_l^{-}\right)^2\right) &= 2
\end{align}
Finally, it is obvious from the definition of $O_\alpha$ that 
\begin{align}
    \text{Tr}\left(O_{\alpha}^2\right) &= 1
\end{align}

Section 5B in \cite{Huang_2020} proves that $\norm{O_{i}}^2_{shadow} \leq 3\text{Tr}\left(\left(O_{i}\right)^2\right)$ for a Hermitian operator $O_i$. 

Since Lemma~\ref{lem:shadownormUB} defines $\textbf{O} \equiv \{O_i \mid i \in \mathbb{Z}_{2M}\}$ as the set \{$O_{\alpha}$\} $\cup$  $\{O_l^{+}$, $O_l^{-} \mid l \in \mathbb{Z}_{M}\}$, we see that 
\begin{align}
    \text{Tr}\left(\left(O_{i}\right)^2\right) & \leq 2,
\end{align}
so
\begin{align}
    \norm{O_{i}}^2_{shadow} \leq 6,
\end{align}
completing the proof of Lemma~\ref{lem:shadownormUB}.
\end{proof}


\subsection{Choosing a Sufficiently Small Classical Shadows Error}\label{App:esUB}

To keep the total error of the Hamiltonian learning problem to at most $\epsilon$, the error of the classical shadows procedure, $\epsilon_s$, cannot be too large. However, decreasing $\epsilon_s$ requires increasing the number of measurements $N$. Therefore, we seek the largest value of $\epsilon_s$ such that the total error is at most $\epsilon$. This value is given by Proposition~\ref{prop:errbd}, which is proven below.

\begin{proof}[Proof of Proposition~\ref{prop:errbd}]
    Using classical shadows in Algorithm~\ref{algorithm: FindCoeffClifford}, we compute $\hat{o}_i$ and $\hat{o}_\alpha$, which, with a probability of at least $\delta_s$, are estimates of $\frac{1}{\alpha^2}$ and $\frac{c_i}{\alpha^2}$ to within an error of $\epsilon_s$ . This means that upon success, the uncertainty in $\alpha^2$, which we get by taking the reciprocal of our estimate of $\frac{1}{\alpha^2}$, is at most
    \begin{align}
        \epsilon_{\alpha} = \alpha^4 \epsilon_s
    \end{align}
    
    The next step in Algorithm~\ref{algorithm: FindCoeffClifford} is to divide $\hat{o}_i$ by $\hat{o}_\alpha$ to get $\hat{c}_i$, our estimate of $c_i$. The uncertainty in $c_i$ is therefore at most
    \begin{align}
        \epsilon_i &= c_i \sqrt{\left(\frac{\epsilon_{\alpha}}{\alpha^2}\right)^2 + \left(\frac{\epsilon_s}{\frac{c_i}{\alpha^2}}\right)^2} \nonumber \\
        &= \epsilon_s \alpha^2 \sqrt{c_i^2 +1}
    \end{align}
    
    The total $l_2$ error in the vector of Hamiltonian coefficients is then
    \begin{align}
        \norm{\hat{\textbf{c}} - \textbf{c}}_2 & \leq \sqrt{M} \max_i \epsilon_i \nonumber \\
        & = \sqrt{M} \epsilon_s \alpha^2  \sqrt{\max_i c_i^2 +1}
    \end{align}
    For this total error to be at most $\epsilon$, as required by the Hamiltonian learning problem, it serves to choose 
    \begin{align}
        \epsilon_s & = \frac{\epsilon}{\alpha^2 \sqrt{c^2_{\rm max} +1} \sqrt{M}} \label{eq:esUB}
    \end{align}
\end{proof}


\subsection{Proof of Lemma~\ref{lemma: ProduceChoiState}}\label{App:ProduceChoiState}

Here we provide a formal proof of Lemma~\ref{lemma: ProduceChoiState} which provides our key result surrounding the cost of preparing the pseudo-Choi state that represents the Hamiltonian.
\begin{proof}[Proof of Lemma~\ref{lemma: ProduceChoiState}]

Begin by preparing the state $\ket{\psi_0} = \ket{0}_C \ket{0}_B \ket{\Phi_d}_{SA}$, and use it as the input state to the circuit in Figure~\ref{figure: ProduceChoiState}. 
Applying the Hadamard gate to register $C$ gives
\begin{align}
    \ket{\psi_1} &= \frac{1}{\sqrt{2}}\Big(  \ket{0}_C \ket{0}_B \ket{\Phi_d}_{SA} + \ket{1}_C \ket{0}_B \ket{\Phi_d}_{SA} \Big)
\end{align}

After the controlled application of $U_{block}$, which acts on registers $B$ and $S$, we have the state 
\begin{align}
    \ket{\psi_2} &= \frac{1}{\sqrt{2}}\Big(  \ket{0}_C (U_{block}\otimes \Id_A) \ket{0}_B \ket{\Phi_d}_{SA} + \ket{1}_C \ket{0}_B \ket{\Phi_d}_{SA} \Big)
\end{align}

Recall that the block-encoding of the Hamiltonian is of the form
\begin{align}
    U_{block} & = \ket{0}_B \bra{0}_B \otimes \frac{2\widetilde{H}t}{\pi} + \ket{0}_B \bra{1}_B \otimes I_S + \ket{1}_B \bra{0}_B \otimes J_S + \ket{1}_B \bra{1}_B \otimes K_S
\end{align}
where $I_S$, $J_S$, and $K_S$ are ``junk'' operators we do not care about. Thus we can write the previous state as
\begin{align}
    \ket{\psi_2} &= \frac{1}{\sqrt{2}}\Bigg(  \ket{0}_C   \ket{0}_B  \left(\frac{\widetilde{H}}{\Delta} \otimes \Id_A\right) \ket{\Phi_d}_{SA}  + \ket{0}_C   \ket{1}_B  ( J_S \otimes \Id_A) \ket{\Phi_d}_{SA}  + \ket{1}_C \ket{0}_B \ket{\Phi_d}_{SA} \Bigg)
\end{align}

The probability of obtaining the outcome $\ket{0}_B$ upon measuring the block-encoding qubit is given by:
\begin{align}
    \mathbf{Pr}(0) &= \bra{\psi_2} (\Id_C \otimes \ket{0}_B\bra{0}_B \otimes \Id_{SA})\ket{\psi_2} \nonumber \\
    & = \bra{\psi_2}  \left(\frac{  \ket{0}_C   \ket{0}_B  \left(\frac{\widetilde{H}}{\Delta} \otimes \Id_A\right) \ket{\Phi_d}_{SA}   + \ket{1}_C \ket{0}_B \ket{\Phi_d}_{SA}}{\sqrt{2}} \right) \nonumber  \\
    & =   \frac{1}{2}\left( \bra{\Phi_d}_{SA} \left(\frac{\widetilde{H}^\dagger \widetilde{H}}{\Delta^2} \otimes \Id_A\right)  \ket{\Phi_d}_{SA}   + 1 \right) \nonumber \\
    & = \frac{\norm{\widetilde{\textbf{c}}}_2^2}{2 \Delta^2} + \frac{1}{2},
\end{align}
where in the last line we inserted the representation of the Hamiltonian $\widetilde{H}$ from Definition~\ref{def:Hamdefs} (the missing steps are nearly identical to the proof of Equation \eqref{eq:alpha} in Appendix~\ref{section:randomsection2}). 

If the outcome of the measurement of register B is $\ket{0}_B$, the resulting state is
\begin{align}
    \ket{\psi_3} &= \frac{\left(\Id_C \otimes \ket{0}_B \bra{0}_B \otimes \Id_{SA}\right)\ket{\psi_2}}{\sqrt{\mathbf{Pr}(0)}} \nonumber  \\
    & = \frac{\ket{0}_C \ket{0}_B (\frac{\widetilde{H}}{\Delta} \otimes \Id_A) \ket{\Phi_d}_{SA} + \ket{1}_C \ket{0}_B \ket{\Phi_d}_{SA}}{ \sqrt{2}\sqrt{\frac{\norm{\widetilde{\textbf{c}}}_2^2}{2\Delta^2} + \frac{1}{2}}} 
\end{align}
where $\Delta = \frac{\pi}{2t}$.

Ignoring the extra qubit in register $B$, this is a pseudo-Choi state of $\frac{\widetilde{H}}{\Delta}$:

\begin{align}
    \ket{\psi_c'} &= \left(\frac{\ket{0}_C(\frac{\widetilde{H}}{\Delta} \otimes \Id_A) \ket{\Phi_d}_{SA} + \ket{1}_C\ket{\Phi_d}_{SA}}{\gamma}\right) ,
\end{align}
where
\begin{align}
    \gamma &= \sqrt{\frac{\norm{\widetilde{\textbf{c}}}_2^2}{\Delta^2} + 1}.
\end{align}

\end{proof}

To check that the success probability of producing the pseudo-Choi state (i.e. measuring the block-encoding qubit and getting the outcome $\ket{0}_B$) is a valid probability, we would like to be sure that the term $\bra{\Phi_d}_{SA} \left(\frac{\widetilde{H}^\dagger \widetilde{H}}{\Delta^2} \otimes \Id_A\right)  \ket{\Phi_d}_{SA}$ is in $[1/2,1] $. We start by relating $\norm{\widetilde{\textbf{c}}}_2^2$ to the spectral norm of the Hamiltonian:
\begin{align}
    \norm{\widetilde{\textbf{c}}}_2^2 & = \bra{\Phi_d}_{SA} (\widetilde{H}^\dagger \widetilde{H} \otimes \Id_A)  \ket{\Phi_d}_{SA} \nonumber \\
    &= \frac{1}{d}\sum_{i=1}^d \sum_{j=1}^d \bra{i}_S \bra{i}_A (\widetilde{H}^\dagger \widetilde{H} \otimes \Id_A) \ket{j}_S \ket{j}_A \nonumber  \\
    & = \frac{1}{d}\sum_{i=1}^d \sum_{j=1}^d \bra{i}_S \widetilde{H}^\dagger \widetilde{H} \ket{j}_S \bra{i}_A \ket{j}_A \nonumber  \\
    & = \frac{1}{d}\sum_{i=1}^d  \bra{i}_S \widetilde{H}^\dagger \widetilde{H} \ket{i}_S  \nonumber  \\
    & = \frac{1}{d}{\rm Tr} \left(\widetilde{H}^\dagger \widetilde{H} \right) \nonumber   \\
    & = \frac{1}{d}\norm{ \widetilde{H} }^2_F  \nonumber  \\
    & \geq \frac{\norm{\widetilde{H}}^2_2}{d}.
\end{align}
 Using this, we see that
\begin{align}
    \mathbf{Pr}(0) & \geq \frac{1}{2}\left( \frac{\norm{\widetilde{H}}^2_2}{d\Delta^2}   + 1 \right) \nonumber \\
    & =  \frac{2\norm{\widetilde{H}t}^2_2}{d\pi^2}   + \frac{1}{2} \ge \frac{1}{2}.
\end{align}
where in the second line we used $\Delta = \frac{\pi}{2t}$.  Similarly, using the fact that $\norm{\tilde{H}}_F \le \sqrt{r} \norm{\tilde{H}}_2$ we have that
\begin{align}
    \mathbf{Pr}(0) & \leq \frac{1}{2}\left( \frac{r\norm{\widetilde{H}}^2_2}{d\Delta^2}   + 1 \right) \nonumber \\
    & =  \frac{2r\norm{\widetilde{H}t}^2_2}{d\pi^2}   + \frac{1}{2}.
\end{align}

From Lemma~\ref{lemma:ProduceBlockEncoding} we inherit the following upper bound on the evolution time:
\begin{align}
    t & \leq \frac{1}{2\norm{H}},
\end{align}
which means that if the block-encoded Hamiltonian had no error (i.e. $\widetilde{H} = H$) the success probability would certainly be less than 1. However, we must take into account the error from the block-encoding procedure. Thankfully, the result of Lemma~\ref{lemma:ProduceBlockEncoding} guarantees that 
\begin{align}
    \norm{Ht - \widetilde{H}t} &\leq \epsilon_b,
\end{align}
for $\epsilon_b \in (0, \frac{1}{2})$.
Therefore, the success probability remains valid as it can never be more than 1 or less than $\frac{1}{2}$.




\subsection{Proof of Equation~\eqref{eq: chernoff}}\label{App: Hoeffding}

Since we have an upper bound on $N_s$, the number of pseudo-Choi states required to solve the Block-Encoded Hamiltonian Learning Problem~\ref{def:BELearningProblem}, we would now like to know how many queries to $U_{block}$ will suffice to generate $N_s$ pseudo-Choi states with high probability. An upper bound on this value is given by equation~\eqref{eq: chernoff}, which we now prove.

\begin{proof}[Proof of equation~\eqref{eq: chernoff}]

Consider the following Chernoff bound, from~\cite{Alon2000}:
\begin{lemma}\label{lemma: Chernoff}
    Let $X_1, ..., X_n \in [0,1]$ be independent random variables with $\mathbb{E}[X_i] = p_i$, let $X = \sum_{i=1}^n X_i$, and let $\mu = \mathbb{E}[X] = \sum_{i=1}^N X_i p_i$. For $\beta \in (0,1)$ we have
    \begin{align}
        \mathbf{Pr}(X \leq (1 - \beta)\mu) & \leq e^{- \frac{\beta^2 \mu}{2}}
    \end{align}
\end{lemma}

Now, let $X$ be the number of pseudo-Choi states produced after $\widetilde{N}$ queries to $U_{block}$. We want $X$ to be greater than or equal to $N_s$, so we seek an upper bound on $\mathbf{Pr}(X \leq N_s)$. Since the success probability of obtaining each pseudo-Choi state from one query to the block-encoding is $\frac{\gamma^2}{2}$ (see Appendix~\ref{App:ProduceChoiState}), we have $\mu = \frac{\widetilde{N} \gamma^2}{2}$.

To use the Chernoff bound~\ref{lemma: Chernoff}, we therefore choose
\begin{align}
    \beta & = 1 - \frac{2 N_s}{\gamma^2\widetilde{N}}
\end{align}
Note that since $\beta$ must be greater than zero to apply the Chernoff bound, we require $\widetilde{N} \geq \frac{2 N_s}{\gamma^2}$.

By the Chernoff bound we then have:
\begin{align}
    \mathbf{Pr}(X \leq N_s) & \leq e^{-\frac{\widetilde{N} \gamma^2}{4}\left(1 - \frac{4 N_s}{\gamma^2\widetilde{N}} + \frac{4 N_s^2}{\gamma^4\widetilde{N}^2}\right)}\\
    &= e^{\left(-\frac{\widetilde{N} \gamma^2}{4} + N_s - \frac{N_s^2}{\gamma^2\widetilde{N}}\right)}\\
    & \leq e^{\left(-\frac{\widetilde{N} \gamma^2}{4} + N_s \right)}
\end{align}
We want our failure probability to be at most $\delta_{N_s}$, so we let
\begin{align}
    \delta_{N_s} & \equiv e^{\left(-\frac{\widetilde{N} \gamma^2}{4} + N_s \right)}, 
\end{align}
from which we see that
\begin{align}
    \ln\left({\frac{1}{\delta_{N_s}}}\right) & = \frac{\widetilde{N} \gamma^2}{4} - N_s
\end{align}
Finally, solving for $\widetilde{N}$ gives
\begin{align}
    \widetilde{N} & = 4\frac{ \ln\left({\frac{1}{\delta_{N_s}}}\right)}{\gamma^2} + 4 \frac{N_s}{\gamma^2}
\end{align}

In general, we assume $N_s$ will be much larger than $\ln\left({\frac{1}{\delta_{N_s}}}\right)$. Therefore, to generate at least $N_s$ pseudo-Choi states with probability at least $1 - \delta_{N_s}$, the number of queries to $U_{block}$ needed is
\begin{align}
    \widetilde{N} & \in \mathcal{O} \left(\frac{N_s}{\gamma^2} \right)
\end{align}

\end{proof}

\end{document}